\documentclass[10pt,rqno, a4paper]{article}

\usepackage{xr}
\usepackage{titlesec}
\titleformat*{\section}{\bf\centering\large} 
\titleformat*{\subsection}{\bf} 
\titleformat*{\subsubsection}{\it} 

\usepackage{psfrag}
\usepackage{cite}
\usepackage[small]{caption}
\usepackage{amsmath}
\usepackage{amssymb}
\usepackage{amsfonts}
\usepackage{setspace}
\usepackage{latexsym}
\usepackage{mathrsfs}
\usepackage{epsfig}
\usepackage{amsthm}
\usepackage{yfonts}
\usepackage{graphicx}
\usepackage{color}
\usepackage{verbatim}

\usepackage{enumitem}

\usepackage[margin=2cm]{caption}

\newcommand{\be}{\begin{equation}}
\newcommand{\ee}{\end{equation}}

\newcommand{\vs}{\vspace{0.2cm}}

\usepackage{fancyhdr}
\AtBeginDocument{\thispagestyle{plain}}
\fancypagestyle{plain}
	{\fancyhead{}
	\fancyfoot{}
	\fancyhead[LE,RO]{}
	\fancyhead[LO,RE]{}
	\fancyfoot[C]{\thepage}
	
	\headsep = 20pt}

\numberwithin{equation}{subsection}
\newtheorem{Theorem}{Theorem}[subsection]
\newtheorem{Definition}[Theorem]{Definition}
\newtheorem{Proposition}[Theorem]{Proposition}

\newtheorem{Lemma}[Theorem]{Lemma}
\newtheorem{Corollary}[Theorem]{Corollary}

\newtheorem{Problem}[Theorem]{Problem}

\newcommand{\dist}{d} 

\newcommand{\sg}{g} 
\newcommand{\hg}{{\mathfrak{g}}} 

\newcommand{\length}{L} 
\newcommand{\diam}{{\rm  diam}} 
\newcommand{\area}{A}

\newcommand{\Sa}{{\rm S}^{1}} 

\newcommand{\qM}{S} 
\newcommand{\stM}{\mathcal{M}} 
\newcommand{\sM}{\Sigma} 

\newcommand{\gcur}{\kappa} 
\newcommand{\hqg}{q} 

\newcommand{\qstM}{\mathcal{N}}

\newcommand{\hgls}{\star}
\newcommand{\T}{{\rm T}}
\newcommand{\K}{\mathbb{K}}
\newcommand{\U}{\mathcal{U}}

\usepackage{tocloft} 

\usepackage{setspace, tocloft}

\allowdisplaybreaks

\begin{document}

\thispagestyle{empty}

\begin{center}
{\Large\bf A classification theorem for static vacuum black holes
\vspace{.2cm}

Part II: the study of the asymptotic
}

\vspace{.5cm}

{\sc Mart\'in Reiris Ithurralde}

{mreiris@cmat.edu.uy}

\vs

{\it Centro de Matem\'atica/Universidad de la Rep\'ublica}

{\it Montevideo, Uruguay}
\vs

\begin{abstract}
This is the second article of a series or two, proving a generalisation of the uniqueness theorem of the Schwarzschild solution. The theorem to be shown classifies all (metrically complete) solutions of the static vacuum Einstein equations with compact but non-necessarily connected horizon without any further assumption on the topology or the asymptotic. Specifically, it is shown that any such solution is either: (i) a Boost, (ii) a Schwarzschild black hole, or (iii) is of Myers/Korotkin-Nicolai type, that is, it has the same topology and Kasner asymptotic as the Myers/Korotkin-Nicolai black holes.  

In this Part II we show that the only end of a static black hole data set is either asymptotically flat or asymptotically Kasner. This proves the third and last step required in the proof of the classification theorem, as was explained in Part I. The analysis requires a thorough study of static data sets with a free $\Sa$-symmetry and a delicate investigation of the geometry of static black hole ends with sub-cubic volume growth. Many of the conclusions of this article are hitherto unknown and have their own interest. 
\end{abstract}

\newpage

\thispagestyle{plain}

\begin{minipage}[l]{11cm}
{\small
\tableofcontents
}
\end{minipage}
\end{center}
\vs

\section{Introduction}  

This is the second of two articles intended to prove the following classification theorem of vacuum static black holes,
\begin{Theorem}[The classification Theorem]\label{TCTHM3} Any static black hole data set is either,
\begin{enumerate}[labelindent=\parindent, leftmargin=*, label={\rm (\Roman*)}, widest=a, align=left]
\item\label{FTII} a Schwarzschild black hole, or,
\item\label{FTI} a Boost, or,
\item\label{FTIII} is of Myers/Korotkin-Nicolai type.
\end{enumerate}
\end{Theorem}
For a contextual discussion of Theorem \ref{TCTHM3} including the notion of static black hole data set and a detailed description of each family see the introduction of Part I,\cite{PartI}.

In Part I it was proved that static black holes data sets have only one end and that the horizons are weakly outermost. This accomplished the first two of the three steps required for the proof of Theorem \ref{TCTHM3}, as was explained in subsection \ref{TSOTP} of Part I. In this Part II we prove the third step, namely that the end is either asymptotically flat or asymptotically Kasner. The results of this article are found in sections \ref{S1S} and \ref{VWAE}. Section \ref{S1S}, which is independent of section \ref{VWAE}, has interest in itself and gives a thorough discussion of free $\Sa$-symmetric data sets. This section is used in section \ref{VWAE} where it is proved that black hole ends are either asymptotically flat or asymptotically Kasner. The techniques introduced for the asymptotic study are so far new and are based upon a careful analysis of static solutions on metrically collapsed annuli. Many of the conclusions are hitherto unknown and have their own interest.
\vs

Before discussing in subsection \ref{TCSTA} the structure of the article and the different proofs, let us make a succinct summary of the background material. We hope it will help the presentation and the reading. Complementary information is found in the background section \ref{BACKGROUNDMATERIAL2} which contains in particular the background material of Part I.
\vs

Formally, a (vacuum) static black hole data set $(\Sigma;g,N)$ consists of a non-compact orientable three-manifold $\Sigma$ with compact and non-empty boundary $\partial \Sigma$, a three-metric $g$ such that $(\Sigma;g)$  is metrically complete, and a non-negative lapse function $N$ that is zero on $\partial \Sigma$ (the horizons) and positive in the interior $\Sigma^{\circ}=\Sigma\setminus \partial \Sigma$ of $\Sigma$, satisfying the static vacuum Einstein equations
\be\label{STATSTAT}
NRic=\nabla\nabla N,\quad \Delta N=0.
\ee
A static black hole data set $(\Sigma;g,N)$ gives rise to a vacuum static black hole spacetime\footnote{The exterior communication region of it.} (${\bf Ric}=0$), 
\be\label{SPTE2}
{\bf \Sigma}=\mathbb{R}\times \Sigma,\quad {\bf g}=N^{2}dt^{2}+g,
\ee
where $\partial_{t}$ is the static Killing field. Conversely, a static black hole spacetime of the form (\ref{SPTE2}), gives rise to a static black hole data set $(\Sigma;g,N)$. Throughout this article we will work with data sets rather than spacetimes. In other words, we work at the `initial data level'.

The data sets of the Schwarzschild black holes are,
\be
\Sigma=\mathbb{R}^{3}\setminus B(0,2m),\quad g=\frac{1}{1-2m/r}dr^{2}+r^{2}d\Omega^{2}\quad {\rm and}\quad N=\sqrt{1-2m/r}
\ee
where $m>0$ is the ADM-mass and $B(0,2m)$ is the open ball of radius $2m$. The Boost data sets are,
\be
\Sigma=[0,\infty)\times \T^{2},\quad g=dx^{2}+h,\quad N=x
\ee
where $h$ is any flat metric on the two-torus $\T^{2}=\Sa\times \Sa$. Finally a data set $(\Sigma;g,N)$ is of Myers/Korotkin-Nicolai type if $\Sigma$ has the topology of an open solid three-torus minus a finite number of open three-balls, and if the asymptotic is Kasner. 

The Kasner spaces (that define the asymptotic) are defined as any $\mathbb{Z}\times \mathbb{Z}$-quotient of the data,

\be
\tilde{\Sigma}=(0,\infty)\times \mathbb{R}^{2};\quad \tilde{g}= dx^{2}+x^{2\alpha}dy^{2}+x^{2\beta}dz^{2},\quad \tilde{N}=x^{\gamma},
\ee
where $y$ and $z$ are coordinates on each of the factors $\mathbb{R}$ of $\mathbb{R}^{2}$, and $\alpha, \beta$ and $\gamma$ are any numbers satisfying,
\be\label{CIRCLE1}
\alpha+\beta+\gamma=1,\qquad \alpha^{2}+\beta^{2}+\gamma^{2}=1
\ee
(see Figure \ref{Figure21}).
\begin{figure}[h]
\centering
\includegraphics[width=6cm, height=6cm]{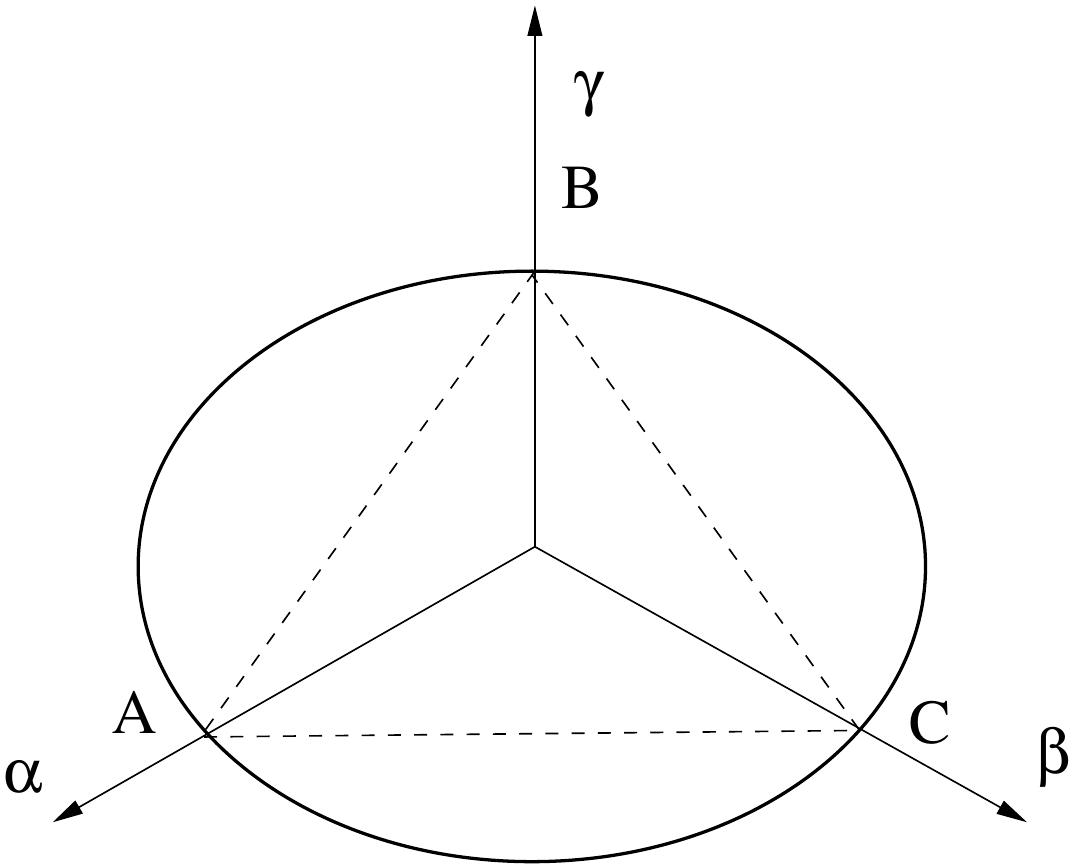}
\caption{The circle that defines the range of the Kasner parameters $\alpha$, $\beta$, $\gamma$.}
\label{Figure21}
\end{figure}
The group $\mathbb{Z}\times \mathbb{Z}$ acts freely on the factor $\mathbb{R}^{2}$ by translations and therefore the quotient manifold is diffeomorphic to $(0,\infty)\times \T^{2}$. Observe that the diameters of the `transversal tori' $\{x\}\times \T^{2}$ are determined by the $\mathbb{Z}\times \mathbb{Z}$-action and can be therefore arbitrarily small.

The Kasner spaces with $(\alpha,\beta,\gamma)=(0,0,1)$ are the Boosts and are the Kasner data with faster growth of the lapse (linear). They are the only Kasner that are static black hole data sets\footnote{One must include $\{0\}\times \T^{2}$ to the manifold.}, the other being singular as $x\rightarrow 0$. We denote the Boosts by the letter $B$. The Kasner spaces $(\alpha,\beta,\gamma)=(1,0,0)$ and $(\alpha,\beta,\gamma)=(0,1,0)$, that have constant lapse and are therefore flat, are denoted respectively by the letters $A$ and $C$. In simple terms, a data set is asymptotically Kasner if it approaches a particular Kasner data, at any order of differentiability, faster than any inverse power of the distance (see Definition \ref{KADEF}).

In this article we will use mainly the harmonic presentation of data sets, namely we will use $(\Sigma;\hg,U)$ instead of $(\Sigma; g,N)$ where,
\begin{equation}
\hg=N^{2}g,\quad U=\ln N.
\end{equation}
The static equations (\ref{STATSTAT}) now are,
\be
Ric_{\hg}=2\nabla U\nabla U,\quad \Delta_{\hg} U=0,
\ee
and have several geometric advantages. In particular, the Ricci curvature of $\hg$ is non-negative, and is zero iff $U$ is constant. 

The Kasner spaces in the harmonic presentation are now $\mathbb{Z}\times \mathbb{Z}$-quotients of,
\begin{align}
\tilde{\Sigma}=(0,\infty)\times \Sigma,\quad \tilde{\hg}=dx^{2}+x^{2a}dy^{2}+x^{2b}d z^{2},\quad \tilde{U}=c\ln x
\end{align}
where $a, b$ and $c$ satisfy
\be
2c^{2}+(a-\frac{1}{2})^{2}+(b-\frac{1}{2})^{2}=\frac{1}{2}\quad \text{and}\quad a+b=1
\ee
Thus, the circle (\ref{CIRCLE1}), (see Figure \ref{Figure21}), is seen as an ellipse in the plane $a+b=1$, (see Figure \ref{Figure31}). The $g$-flat solutions $A, C$ and $B$ are now $(a,b,c)=(1,0,0), (0,1,0)$,  and $(1/2,1/2,1/2)$, respectively.
\begin{figure}[h]
\centering
\includegraphics[width=6cm, height=6cm]{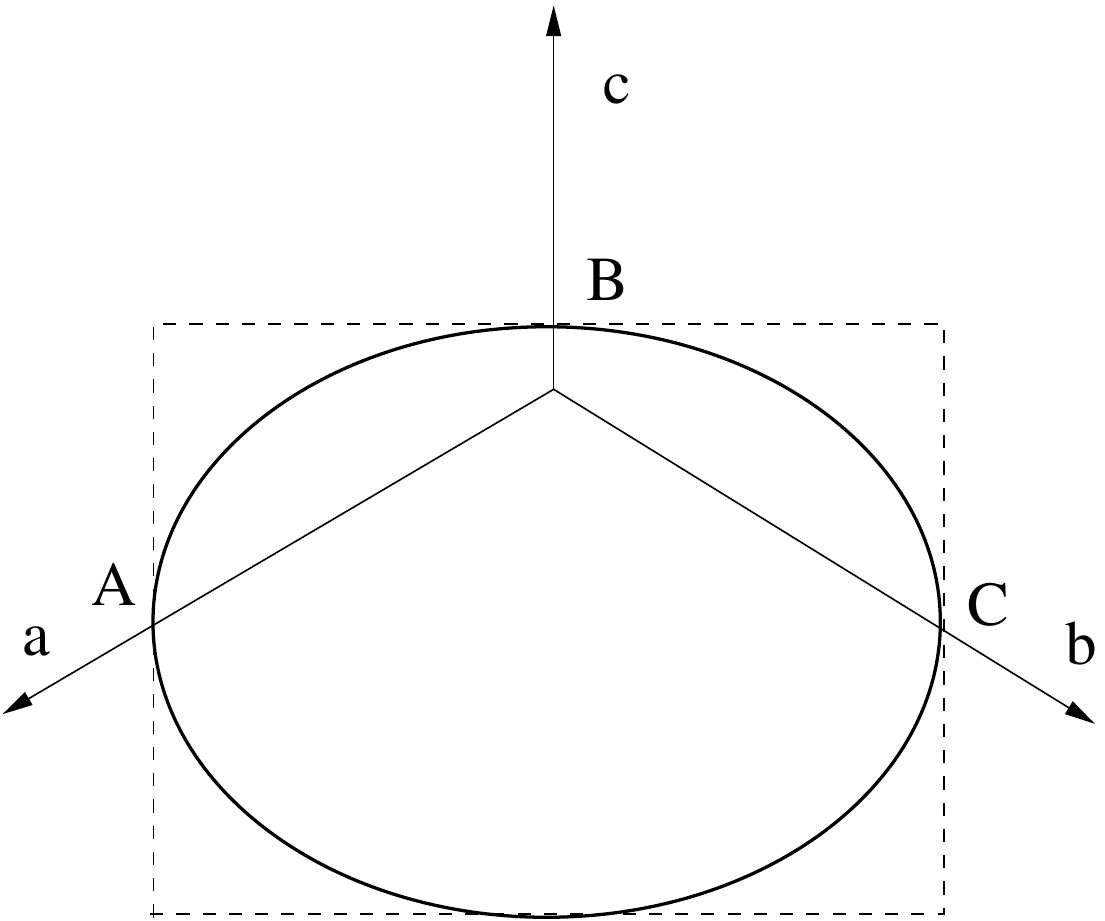}
\caption{The ellipse that defines the range of the parameters $a, b$ and $c$.}
\label{Figure31}
\end{figure}
\vs

The study of the asymptotic will be done by looking at rescaled annuli. Let $k\geq 1$, $r>0$ and let $\mathcal{A}_{\hg}(2^{-k}r,2^{k}r)$ be the annulus,
\be 
\mathcal{A}_{\hg}(2^{-k}r,2^{k}r):=\{p\in \Sigma: 2^{-k}r<\dist_{\hg}(p,\partial \Sigma)<2^{k}r\}
\ee
where $d_{\hg}(p,\partial \Sigma)$ is the $\hg$-distance from $p$ to the boundary $\partial \Sigma$. Fixed $k$, we will let $r$ increase, and, over the annulus $\mathcal{A}_{\hg}(2^{-k}r,2^{k}r)$ we will look at the rescaled metrics $\hg_{r}:=\hg/r^{2}$. Anderson's estimates (see Theorem \ref{LACD2} in Part I) show that $Ric_{\hg_{r}}$ and $\nabla U$ are uniformly bounded (i.e. their $\hg_{r}$-norms have bounds that do not depend on $r$), as so are any derivatives of them. This fundamental property will permit first the analysis of the geometry of rescaled annuli, and then, by suitable concatenation, the analysis of the asymptotic. We will discuss all that in the next subsection.

As a result of the different proofs it will be clear that, if the asymptotic is Kasner, then the parameter $c$ is positive, $c>0$. Hence the Kasner spaces with $c>0$ will be the ones more relevant to us. Still, the Kasner $A$ and $C$, that have $c=0$, will also come into play often but for technical reasons. There is a significant difference between $A$ and $C$ on one side, and the Kasner with $c>0$ on the other side: the diameter of the transversal tori $\{x\}\times \T^{2}$ grows linearly with the distance $x$ in the first case but sub-linearly in the second case. Thus, if $c>0$, the diameters of the tori $\{x\}\times \T^{2}$ with respect to $\hg_{x}=\hg/x^{2}$, tend to zero as $x\rightarrow \infty$. Another way to say this is: fixed $k\geq 1$, the Riemannian annuli $(\mathcal{A}_{\hg_{x}}(2^{-k},2^{k});\hg_{x})$ metrically collapse, as $x$ tends to infinity, to a segment of length $2^{k}-2^{-k}$ but there is no such type of collapse if the Kasner is $A$ or $C$, they metrically collapse to a two-dimensional flat annulus, (for metric-collapse see subsection \ref{SFCCRM}). As said, these global differences will cause technical difficulties while studying Kasner asymptotic. We will resume this point in subsection \ref{TCSTA}.
\vs

We move now to explain the structure of the article and the route behind the series of results and their proofs. In particular the claims of section \ref{VWAE} are somehow interrelated and therefore it is useful to have a clear overview.    

\subsection{The content and the structure of this article (Part II)}\label{TCSTA}
Section \ref{BACKGROUNDMATERIAL2} contains the background material, including notation and terminology. Subsection \ref{SDSMT2} contains the main definitions, as the one of static black hole data set or Kasner asymptotic, and states again the classification theorem as Theorem \ref{TCTHM3}. Subsection \ref{SAP2} defines annuli and partitions cuts, that are technically useful to study asymptotic properties. All that is the background material that was already introduced in Part I. The rest of the background material is special for this Part II and is the following. Also inside subsection \ref{SAP}, we pay special attention to `scaling', and notations related to it that will be used massively when studying ends in section \ref{VWAE} (it is important to keep track of them). Scaling techniques are useful due to the scale invariance of Anderson's decay estimates for the curvature and for the gradient of the logarithm of the lapse, see the Theorems \ref{LACD1} and \ref{LACD2} in Part I. Furthermore the study of ends through scaling requires a minimum material on the Cheeger-Gromov-Fukaya theory of convergence and collapse of Riemannian manifolds under curvature bounds that is shortly introduced in subsection \ref{SFCCRM}. Subsection \ref{SSKK} contains a careful account of Kasner spaces and a suitable proof of their (well known) uniqueness, which will be used throughout section \ref{VWAE} when we discuss asymptotic. This ends the background section.
\vs

Section \ref{S1S} marks the beginning of the results of Part II, whose final goal is to describe, in section \ref{VWAE}, the asymptotic of ends of black hole data sets. Section \ref{S1S} studies various aspects of data sets which are free $\Sa$-symmetric. It has interest in itself and proves a number of novel results on these types of spaces. The contents are as follows. Subsection \ref{RDRE} presents the reduced equations, that is, the fields and the equations that are obtained after the quotient by a $\Sa$-symmetric static data set by $\Sa$, Proposition \ref{PRED}. The reduced data $(\qM;q,U,V)$ of a static data set $(\Sigma; \hg,U)$, with a $\Sa$-symmetry generated by a Killing field $\xi$, consists of a two-manifold $\qM$, a Riemannian metric $q$ on $\qM$, and two fields, the usual field $U=\ln N$, and $V=\ln \Lambda$ with $\Lambda=|\xi|_{\hg}$. Relevant examples of reduced data sets are discussed in the subsections \ref{KSOL} and \ref{CIGARSOL}. Subsection \ref{KSOL} discusses, as a natural example, the reduced Kasner spaces (this subsection can be skipped). Subsection \ref{CIGARSOL} makes a thorough description of another particular reduced data that we call the `cigars' (due to their geometric shape). Subsection \ref{CIGUNIQU} proves a uniqueness statement for the cigars and subsection \ref{CIGNHCP} characterises the cigars as the data that model high-curvature regions. These properties of the cigars play an essential role in subsection \ref{DFIAT}, where it is proved that $|\nabla U|^{2}$, $|\nabla V|^{2}$ and $\kappa$ (the Gaussian curvature of $q$), have quadratic decay at infinity on $(\qM; q)$, (provided $(S;q)$ is metrically complete and $\partial S$ is compact). A few comments are in order here. The discussion of such decay depends on whether the twist $\Omega$ of $\xi$, which is a constant, is zero or not. When it is zero, the quadratic decay can be obtained using the same techniques \'a la Bakry-\'Emery used to prove a generalised Anderson's decay in Part I for the gradient of the logarithm of the lapse, Proposition \ref{REDCUR}. However, it turns out that such techniques do not entirely apply when the twist $\Omega$ is not zero. For this reason, quadratic decay in such case is proved arguing by contradiction, which explains why we study high-curvature regions in subsection \ref{CIGNHCP}. In the same subsection \ref{DFIAT} it is shown, using the decay previously proved, that $S$ has only a finite number of simple ends, each diffeomorphic to $[0,\infty)\times \Sa$. Furthermore it is proved in Proposition \ref{LPRO} that $U$ has a limit $U_{\infty}$ at infinity, $-\infty\leq U_{\infty}\leq \infty$. These are the most important results of section \ref{S1S}.
\vs

Altogether section \ref{VWAE} proves that black hole data sets are either asymptotically flat or asymptotically Kasner. The two types of asymptotic, that are discussed separately in subsections \ref{ENDSAF} and \ref{ENDSAK}, are distinguished by the type of volume growth of the end $(\Sigma;\hg)$, which is at most cubic by the Bishop-Gromov volume comparison. In subsection \ref{ENDSAF} it is shown that if the volume growth of the static black hole end is cubic then the end is asymptotically flat, whereas in subsection \ref{ENDSAK}, which has four subsections, it is shown the fundamental Theorem \ref{KAFR} stating that, if the volume growth is sub-cubic, then the asymptotic is Kasner. It is important to remark that as a byproduct of the proofs it will be shown that the Kasner asymptotic is indeed different from the flat Kasner $A$ or $C$, and of course from any Kasner with parameter $\gamma$ less than zero, (or $c<0$ if we work in the harmonic presentation), that are ruled out by the maximum principle (as then $N\rightarrow 0$ at infinity). This behaviour is compatible with the asymptotic of Myers/Korotkin-Nicolai black holes that can be that of any Kasner different from $A$, $B$, $C$ and different from those with $\gamma<0$. We leave it as an open problem to prove that the only static black hole data sets asymptotic to a Boost are in fact the Boosts. A more elaborated discussion on this point can be found in the introduction of Part I. 

Let us describe more in detail now the four subsections \ref{SPTWOP0},  \ref{SPTWOP1},  \ref{FTKASS}, \ref{POKA}.

The preliminary subsection \ref{SPTWOP0} discusses metrics on two-tori under conditions on the curvature and the diameter, and is used in the next subsection to study the geometry of the level sets of the lapse on (almost) one-collapsed annuli (that is, annuli whose geometry is `near' one-dimensional in the Gromov-Hausdorff metric, see subsection \ref{RDSACL}). 

Subsection \ref{SPTWOP1} proves a sufficient condition for a static end (non-necessarily a static black hole end) to have Kasner asymptotic different from $A$ or $C$, Theorem \ref{KASYMPTOTIC}. Roughly speaking, the sufficient criteria says that if the rescaled geometry of a sufficiently one-collapsed annulus has $|\nabla U|_{\hg_{r}}\neq 0$, then the end is asymptotic to a Kasner different from $A$ or $C$. The proof of Theorem \ref{KASYMPTOTIC} necessitates of two propositions, that we now comment in big terms. First, Proposition \ref{PORSI} shows, roughly (see the hypothesis), that any rescaled annulus that is sufficiently one-collapsed and has $|\nabla U|_{\hg_{r}}\neq 0$, is `$C^{k}$-close' to a Kasner space. More importantly, it  estimates the `$C^{k}$-distance' to the Kasner space, to any order of differentiability $k$, in terms of any power of the diameter of the transversal tori (i.e. the level sets of the lapse). The proof requires a very detailed study of one-collapsed annuli, done (in part) by carefully inspecting the geometry of the transversal tori, whose second fundamental forms is fully controlled as an easy consequence of Anderson's estimates. Second, Proposition \ref{PORE} proves a type of `geometric bootstrap'-procedure, basically stating (see the hypothesis) that if the rescaled geometry of an annulus is `close in $C^{k}$' to a Kasner different from $A$ and $C$, then the rescaled geometry of a next annulus (following the one before) is even `closer in $C^{k}$' to an annulus of a Kasner different from $A$ and $C$. Thus, once we are in the hypothesis of Theorem \ref{KASYMPTOTIC}, a proof can be reachd by (roughly) using Proposition \ref{PORSI} first, and then using repeatedly Proposition \ref{PORE}.

The sufficient criteria for Kasner asymptotic of Theorem \ref{KASYMPTOTIC} is the one leading ultimately to the proof of Theorem \ref{KAFR} showing, as said, the KA of black hole data sets with sub-cubic volume growth. But to apply the theorem one must grant first that the geometry of at least one rescaled annulus satisfies its hypothesis, and this is not easy. The proof of Theorem \ref{KAFR} is elaborated and requires subsections \ref{FTKASS} and \ref{POKA}. 

Subsection \ref{FTKASS} returns to the study of free $\Sa$-symmetric data sets $(\Sigma;\hg,U)$ by analysing the asymptotic of their ends under the natural condition $U\leq U_{\infty}$. It is shown in Theorem \ref{SSKAA} that, for such data, either the asymptotic is Kasner different from $A$ or $C$ or the whole data is flat and $U$ is constant. The proof of this follows after various steps. First we use the results of section \ref{S1S} to prove that such ends, when non-flat, are $\star$-static (Definition \ref{DEFSSS}), meaning in this case that the level sets of $U$ near $U_{\infty}$ are connected, compact and of genus greater than zero. It is then proved that either the asymptotic is Kasner different from $A$ and $C$ or it has sub-quadratic curvature decay. Finally sub-quadratic curvature decay is ruled out by making use of the monotonic quantity (\ref{GMONOTONIC}) along the level sets of $U$ (that we know are not spheres because the data is $\star$-static).

Theorem \ref{SSKAA}, on free $\Sa$-symmetric data sets, is needed in several parts of subsection \ref{POKA} to show the aforementioned non-trivial Theorem \ref{KAFR}.  As we explain more in detail below, free $\Sa$-symmetric reduced data sets show up often as the `collapsed limit' of rescaled annuli. Thus, as we study ends precisely via rescaled annuli, it is natural to expect Theorem \ref{SSKAA} to enter into scene in the analysis at some moment. Let us elaborate on this point a bit more. If an end has sub-cubic volume growth, then (sub)sequences of rescaled annuli either metrically collapse to a one-dimensional space and unwrap (i.e. after considering covers) to a ${\rm T}^{2}$-symmetric (Kasner) static data or metrically collapse to a two-dimensional orbifold $(\qM;q)$ and locally unwrap to a free $\Sa$-symmetric static data (this comes from a general fact about convergence and collapse of Riemannian manifolds that we discuss in subsection \ref{SFCCRM}).  Two-dimensional reduced ends arising as such scaled limits are described in subsection \ref{RDSACL}, (which is the last of section \ref{S1S} earlier discussed), and can have only a finite number of orbifold points. Thus, by the results of subsection \ref{FTKASS}, the asymptotic is Kasner. This crucial information is used suitably to prove Theorem \ref{DFNKA} in subsection \ref{POKA}, stating that the asymptotic of static ends with sub-cubic volume growth is Kasner different from $A$ and $C$ or the curvature decays sub-quadratically along a suitable set (precisely, a ray union a simple cut). Thus, to prove Theorem \ref{KAFR} after having Theorem \ref{DFNKA}, one must rule out the sub-quadratic decay. This is done by showing through various propositions that static black hole ends with sub-cubic volume growth are $\star$-static, and then using Proposition \ref{CORONA} that forbids such type of decay for $\star$-static data by appealing to the monotonic quantity (\ref{GMONOTONIC}). 

The proof of the classification theorem is done in section \ref{TCTH} by carefully putting together all the previous results as was explained in section \ref{TSOTP} of Part I. 
\vspace{.3cm}

{\bf Acknowledgment} I would like to thank Herman Nicolai, Marc Mars, Marcus Kuhri, Gilbert Weinstein, Michael Anderson, Greg Galloway, Miguel Sanchez, Carla Cederbaum, Lorenzo Mazzieri, Virginia Agostiniani and John Hicks for discussions and support. Also my gratefulness to Carla Cederbaum for inviting me to the conference `Static Solutions of the Einstein Equations' (T\"ubingen, 2016), to Piotr Chrusciel for inviting me to the meeting in `Geometry and Relativity' (Vienna, 2017) and to Helmut Friedrich for the very kind invitation to visit the Albert Einstein Institute (Max Planck Institute, Potsdam, 2017). This work has been largely discussed at them. Finally my gratefulness to the support received from the Mathethamical Center at the Universidad de la Rep\'ublica, Uruguay.

\section{Background material}\label{BACKGROUNDMATERIAL2}

\subsection{Static data sets and the main Theorem}\label{SDSMT2}

Manifolds will always be smooth ($C^{\infty}$). Riemannian metrics as well as tensors will also be smooth.  If $g$ is a Riemannian metric on a manifold $\Sigma$, then 
\be
\dist_{g}(p,q)= \inf\big\{\length_{g}(\gamma_{pq}):\gamma_{pq}\ \text{smooth curve joining $p$ to $q$}\big\},
\ee
is a metric, where $L_{g}$ is the notation we will use for length (when it is clear from the context we will remove the sub-index $g$ and write simply $\dist$ and $L$). A Riemannian manifold $(\Sigma;g)$ is {\it metrically complete} if the metric space $(\Sigma; \dist)$ is complete. If $(\Sigma;g)$ is metrically complete and $\partial \Sigma=\emptyset$ then the manifold is geodesically complete and we say simply as usual that $(\Sigma;g)$ is complete. 

\begin{Definition}[Static data sets]\label{SDS} A static (vacuum) data set $(\Sigma;\sg,N)$ consists of an orientable three-manifold $\Sigma$, possibly with boundary, a Riemannian metric $\sg$, and a function $N$, such that,
\begin{enumerate}[labelindent=\parindent, leftmargin=*, label={\rm (\roman*)}, widest=a, align=left]
\item $N$ is strictly positive in the interior $\Sigma^{\circ}(=\Sigma\setminus \partial\Sigma)$ of $\Sigma$,
\item $(\sg,N)$ satisfy the vacuum static Einstein equations,
\be
\label{SEQ}  N Ric = \nabla\nabla N,\qquad \Delta N=0
\ee
\end{enumerate}
\end{Definition}
The definition is quite general. Observe in particular that $\Sigma$ and $\partial \Sigma$ could be compact or non-compact. To give an example, a data set $(\Sigma;\sg,N)$ can be simply the data inherited on any region of the Schwarzschild data. This flexibility in the definition of static data set allows us to write statements with great generality.

A horizon is defined as usual.
\begin{Definition}[Horizons] Let $(\Sigma;\sg,N)$ be a static data set. A horizon is a connected component of $\partial \Sigma$ where $N$ is identically zero.
\end{Definition}
Note that the Definition \ref{SDS} doesn't require $\partial \Sigma$ to be a horizon, though the data sets that we classify in this article are those with $\partial \Sigma$ consisting of a finite set of compact horizons ($\Sigma$ is a posteriori non compact). It is known that the norm $|\nabla N|$ is constant on any horizon and different from zero. It is called the surface gravity.

It is convenient to give a name to those spaces that are the final object of study of this article. Naturally we will call them {\it static black hole} data sets.

\begin{Definition}[Static black hole data set]\label{DWO} A static data set $(\Sigma;\sg,N)$ with $\partial \Sigma=\{N=0\}$ and $\partial \Sigma$ compact, is called a static black hole data set.
\end{Definition}  

In order to study the asymptotic of ends of black hole data sets it will be more convenient to work with `static data ends' that are simply data sets with one end and compact boundary.

\begin{Definition}[Static data end]\label{SDE} A metrically complete static data set $(\Sigma;g,N)$ with $\partial \Sigma$ compact and $\Sigma$ containing only one end, will be called a static data end.
\end{Definition}

As will be shown, black hole data sets have only one end and so they are static data ends themselves. On the other hand, static data ends not necessarily arise from them. Hence, several of the theorems in this article, that are proved for static data ends, have a large range of applicability. 

The following definition, taken from \cite{4b6cb19bc94d4cf485e58571e3062f77}, recalls the notion of {\it weakly outermost} horizon.

\begin{Definition}[Galloway, \cite{4b6cb19bc94d4cf485e58571e3062f77}] Let $(\Sigma; \sg, N)$ be a static black hole data set. Then, a horizon $H$ is said weakly outermost if there are no embedded surfaces $S$ homologous to $H$ having negative outwards mean curvature. 
\end{Definition}

The following is the definition of Kasner asymptotic. It requires a decay into a background Kasner space faster than any inverse power of the distance. The definition follows the intuitive notion and it is written in the coordinates of the background Kasner, very much in the way AF is written in Schwarzschildian coordinates.

\begin{Definition}[Kasner asymptotic]\label{KADEF} A data set $(\Sigma; g,N)$ is asymptotic to a Kasner data $(\Sigma^{\mathbb{K}};g^{\mathbb{K}},N^{\mathbb{K}})$, $\Sigma^{\mathbb{K}}=(0,\infty)\times \T^{2}$, if for any $m\geq 1$ and $n\geq 0$ there is $c>0$, a bounded closed sets $K\subset \Sigma$, $K^{\mathbb{K}}\subset \Sigma^{\mathbb{K}}$ and a diffeomorphism $\phi:\Sigma\setminus K\rightarrow \Sigma^{\mathbb{K}}\setminus K^{\mathbb{K}}$ such that,
\begin{gather}
|\partial_{I}(\phi_{*}g)_{ij}-\partial_{I}g^{\mathbb{K}}_{ij}|\leq \frac{c}{x^{m}},\\
|\partial_{I}(\phi_{*}N)-\partial_{I}N^{\mathbb{K}}|\leq \frac{c}{x^{m}},
\end{gather}
for any multi-index $I=(i_{1},i_{2},i_{3})$ with $|I|=i_{1}+i_{2}+i_{3}\leq n$, where, if $x, y$ and $z$ are the coordinates in the Kasner space, then $\partial_{I}=\partial_{x}^{i_{1}}\partial_{y}^{i_{2}}\partial_{z}^{i_{3}}$. 
\end{Definition}

The next is the definition of data set of Myers/Korotkin-Nicolai type that we use.

\begin{Definition}[Black holes of M/KN type]\label{KNTDEF} A static data set $(\Sigma;\sg,N)$ is of Myers/Korotkin-Nicolai type if
\begin{enumerate}
\item $\partial \Sigma$ consist of $h\geq 1$ weakly outermost (topologically) spherical horizons,
\item $\Sigma$ is diffeomorphic to a solid three-torus minus  $h$-open three-balls,
\item the asymptotic is Kasner.
\end{enumerate}
\end{Definition}

It is worth to restate now the main classification theorem that we shall prove

\begin{Theorem}[The classification Theorem]\label{TCTHM4} Any static black hole data set is either,
\begin{enumerate}[labelindent=\parindent, leftmargin=*, label={\rm (\Roman*)}, widest=a, align=left]
\item\label{FTII} a Schwarzschild black hole, or,
\item\label{FTI} a Boost, or,
\item\label{FTIII} is of Myers/Korotkin-Nicolai type.
\end{enumerate}
\end{Theorem}

As an outcome of the proof it will be shown that the Kasner asymptotic of the static black holes of type \ref{FTIII}, that is of M/KN type, is different from the Kasner $A$ and $C$ (of course it can't be asymptotic to a Kasner with $\gamma<0$ by the maximum principle). We leave it as an open problem to prove that the only static black hole data sets asymptotic to a $B$ are the Boosts. 

\begin{Problem} Prove that the Boosts are the only static black hole data sets asymptotic to a Boost.
\end{Problem}

We do not know if the only solutions of type \ref{FTIII} are the Myers/Korotkin-Nicolai solutions. We state this as an open problem.

\begin{Problem}\label{OPENPRO2} Prove or disprove that the only static solutions of type \ref{FTIII} are the Myers/Korotkin-Nicolai solutions.
\end{Problem}

On a large part of the article we will use the variables $(\hg,U)$ with $\hg=N^{2}g$ and $U=\ln N$, instead of the natural variables $(g,N)$. The data $(\Sigma;\hg,U)$ is the {\it harmonic presentation} of the data $(\Sigma;g,N)$. The static equations in these variables are,
\be
Ric_{\hg}=2\nabla U\nabla U,\quad \Delta_{\hg} U=0
\ee
and therefore the map $U:(\Sigma;\hg)\rightarrow \mathbb{R}$ is harmonic, (hence the name).

\subsection{Scaling, annuli and partitions}\label{SAP2}

\begin{enumerate}[leftmargin=*, label={\rm \arabic*}, widest=a, align=left]

\item {\sc Metric balls}. If $C$ is a set and $p$ a point then $\dist_{g}(C,p)=\inf\{\dist_{g}(q,p):q\in C\}$. Very often we take $C=\partial \Sigma$. If $C$ is a set and $r>0$, then, define the open ball of `center' $C$ and radius $r$ as,
\be
B_{g}(C,r)=\{p\in \Sigma:\dist_{g}(C,p)<r\}
\ee

\item {\sc Scaling}. Very often we will work with scaled metrics. To avoid a cumbersome notation we will use often the subindex $r$ (the scale) on scaled metrics, tensors and other geometric objects. Precisely, let $r>0$, then for the scaled metric $g/r^{2}$ we use the notation $g_{r}$, namely,
\be
g_{r}:=\frac{1}{r^{2}}g
\ee
Similarly, $d_{r}(p,q)=d_{g_{r}}(p,q)$, $\langle X,Y\rangle_{r}=\langle X,Y\rangle_{g_{r}}$, $|X|_{r}=|X|_{g_{r}}$, and for curvatures and related tensors too, for instance if $R$ is the scalar curvature of $g$, then $R_{r}$ is the scalar curvature of $g_{r}$. 

This notation will be used very often and is important to keep track of it. 

\item {\sc Annuli}. Let $(\Sigma;g)$ be a metrically complete and non-compact Riemannian manifold with non-empty boundary $\partial \Sigma$. 

- Let $0<a<b$, then we define the open annulus $\mathcal{A}_{g}(a,b)$ as
\be
\mathcal{A}_{g}(a,b)=\{p\in \Sigma: a<\dist_{g}(p,\partial \Sigma)<b\}
\ee
We write just $\mathcal{A}(a,b)$ when the Riemannian metric $g$ is clear from the context. 

- When working with scaled metrics $g_{r}$, we will alternate often between the following notations
\be
\mathcal{A}_{r}(a,b),\quad \mathcal{A}_{g_{r}}(a,b),\quad \mathcal{A}_{g}(ra,rb),
\ee
(to denote the same set), depending on what is more simple to write or to read. For instance we could write $\mathcal{A}_{2^{j}}(1,2)$ instead of $\mathcal{A}_{g_{2^{j}}}(1,2)$ or $\mathcal{A}_{g}(2^{j},2^{1+j})$. If the metric is clear from the context we will use the first notation $\mathcal{A}_{r}(a,b)$.

- If $C$ is a connected set included in $\mathcal{A}_{g}(a,b)$, then we define,
\be
\mathcal{A}^{c}_{g}(C;a,b)
\ee
to denote the connected component of $\mathcal{A}_{g}(a,b)$ containing $C$. The set $C$ could be for instance a point $p$ in which case we write $\mathcal{A}^{c}_{g}(p;a,b)$.

\item {\sc Partitions cuts and end cuts}. To understand the asymptotic geometry of data sets, we will study the geometry of scaled annuli. Sometimes however it will be more convenient and transparent to use certain sub-manifolds instead of annuli that can have a rough boundary. For this purpose we define partitions, partition cuts, end cuts, and simple end cuts.

{\it Assumption}: Below (inside this part) we assume that $(\Sigma;g)$ is a metrically complete and non-compact Riemannian manifold with non-empty and compact boundary $\partial \Sigma$. 

\begin{Definition}[Partitions] A set of connected compact three-submanifolds of $\Sigma$ with non-empty boundary 
\be
\{\mathcal{P}^{m}_{j,j+1},\  j=j_{0},j_{0}+1,\ldots;\ m=1,2,\ldots,m_{j}\geq 1\}, 
\ee
($j_{0}\geq 0$), is a {\it partition} if,
\begin{enumerate}
\item $\mathcal{P}^{m}_{j,j+1}\subset \mathcal{A}(2^{1+2j},2^{4+2j})$ for every $j$ and $m$.
\item $\partial \mathcal{P}^{m}_{j,j+1}\subset (\mathcal{A}(2^{1+2j},2^{2+2j})\cup \mathcal{A}(2^{3+2j},2^{4+2j}))$ for every $j$ and $m$.
\item The union $\cup_{j,m}\mathcal{P}^{m}_{j,j+1}$ covers $\Sigma\setminus B(\partial \Sigma,2^{2+2j_{0}})$.
\end{enumerate}
\end{Definition}

\begin{figure}[h]
\centering
\includegraphics[width=7cm, height=9cm]{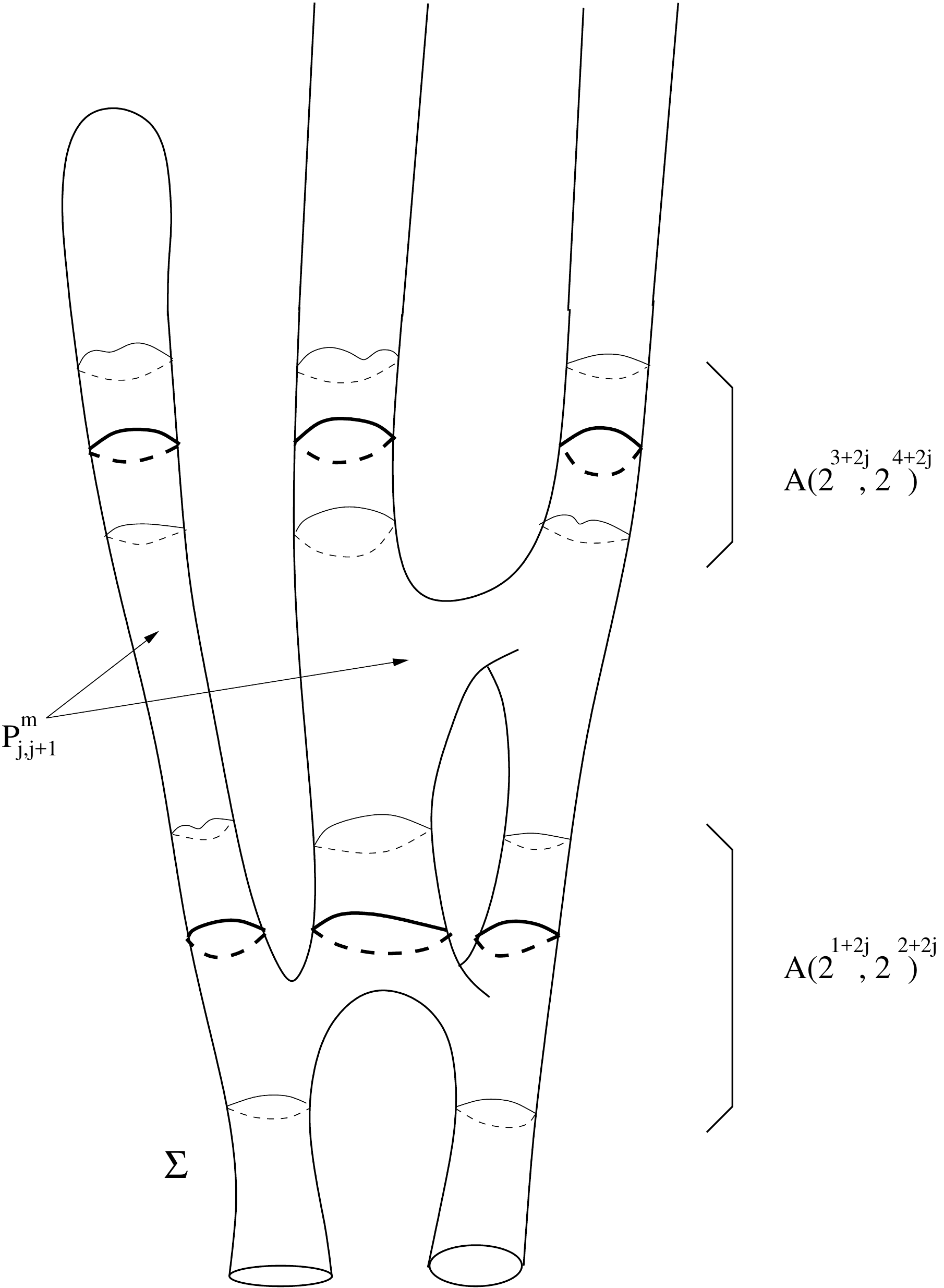}
\caption{The figure shows the annuli $\mathcal{A}(2^{1+2j},2^{2+2j})$, $\mathcal{A}(2^{3+2j},2^{4+2j})$ and the two components, for $m=1,2$ of $\mathcal{P}^{m}_{j,j+1}$.}
\label{PARTITIONF}
\end{figure}

Figure \ref{PARTITIONF} shows schematically a partition. The existence of partitions is done (succinctly) as follows. 
Let $j_{0}\geq 0$ and let $j\geq j_{0}$. Let $f:\Sigma\rightarrow [0,\infty)$ be a (any) smooth function such that $f\equiv 1$ on $\{p:\dist(p,\partial \Sigma)\leq 2^{1+2j}\}$ and $f\equiv 0$ on $\{p: \dist(p,\partial \Sigma)\geq 2^{2+2j}\}$, 
\footnote{Consider a partition of unity $\{\chi_{i}\}$ subordinate to a cover $\{\mathcal{B}_{i}\}$ where the neighbourhoods $\mathcal{B}_{i}$ are small enough that if $\mathcal{B}_{i}\cap \{p:\dist(p,\partial \Sigma)\leq 2^{1+2j}\}\neq \emptyset$ then $\mathcal{B}_{i}\cap \{p: \dist(p,\partial \Sigma)\geq 2^{2+2j}\}=\emptyset$. Then define $f=\sum_{i\in I}\chi_{i}$, where $i\in I$ iff $\mathcal{B}_{i}\cap \{p:\dist(p,\partial \Sigma_{i})\leq 2^{1+2j}\}\neq \emptyset$.}. 
Let $x$ be any regular value of $f$ in $(0,1)$. For each $j$ let $\mathcal{Q}_{j}$ be the compact manifold obtained as the union of the closure of the connected components of $\Sigma\setminus \{f=x\}$ containing at least a component of $\partial \Sigma$. Then the manifolds $\mathcal{P}^{m}_{j,j+1}$, $m=1,\ldots,m_{j}$, are defined as the connected components of $\mathcal{Q}_{j+1}\setminus \mathcal{Q}_{j}^{\circ}$. 

We let $\partial^{-}\mathcal{P}^{m}_{j,j+1}$ be the union of the connected components of $\partial \mathcal{P}^{m}_{j,j+1}$ contained in $\mathcal{A}(2^{1+2j},2^{2+2j})$. Similarly, we let $\partial^{+}\mathcal{P}^{m}_{j,j+1}$ be the union of the connected components of $\partial \mathcal{P}^{m}_{j,j+1}$ contained in $\mathcal{A}(2^{3+2j},2^{4+2j})$. 
\begin{Definition}[Partition cuts]
If $\mathcal{P}$ is a partition, then for each $j$ we let 
\be
\{\mathcal{S}_{jk},k=1,\ldots,k_{j}\} 
\ee
be the set of connected components of the manifolds $\partial^{-}\mathcal{P}^{m}_{j,j+1}$ for $m=1,\ldots,m_{j}$. The set of surfaces $\{\mathcal{S}_{jk},j\geq j_{0},\ldots,k=1,\ldots,k_{j}\}$ is called a {\it partition cut}. 
\end{Definition}

\begin{Definition}[End cuts] 
Say $\Sigma$ has only one end. Then, a subset, $\{\mathcal{S}_{jk_{l}}, l=1,\ldots,l_{j}\}$ of a partition cut
 $\{\mathcal{S}_{jk},k=1,\ldots,k_{j}\}$ is called an end cut if when we remove all the surfaces $\mathcal{S}_{jk_{l}}$, $l=1,\ldots,l_{j}$, from $\Sigma$, then every connected component of $\partial \Sigma$ belongs to a bounded component of the resulting manifold, whereas if we remove all but one of the surfaces $\mathcal{S}_{jk_{l}}$, then at least one connected component of $\partial \Sigma$ belongs to an unbounded component of the resulting manifold.
\end{Definition}

If $\Sigma$ has only one end, then one can always remove if necessary manifolds from a partition cut $\{\mathcal{S}_{jk},k=1,\ldots,k_{j}\}$ to obtain an end cut. End cuts always exist.
 
\begin{Definition}[Simple end cuts]
Say $\Sigma$ has only one end. If an end cut $\{\mathcal{S}_{jk_{l}},j\geq j_{0},l=1,\ldots,l_{j}\}$ has $l_{j}=1$ for each $j\geq j_{0}$ then we say that the end is a {\it simple end cut} and write simply $\{\mathcal{S}_{j}\}$.
\end{Definition}

Simple end cuts do not always exist. If $\{\mathcal{S}_{j}\}$ is a simple end cut and $j_{0}\leq j<j'$ we let $\mathcal{U}_{j,j'}$ be the compact manifold enclosed by $\mathcal{S}_{j}$ and $\mathcal{S}_{j'}$. This notation will be used very often.

\end{enumerate}

\subsection{The ball-covering property and a Harnak-type of estimate for the Lapse}\label{TBCPHTE}

Let $(\Sigma;\sg,N)$ be a static data set with $\partial \Sigma$ compact. In \cite{MR1809792}, Anderson observed that, as the four-metric $N^{2}dt^{2}+\sg$ is Ricci-flat, then Liu's ball-covering property holds, \cite{MR1216638} (the compactness of $\partial \Sigma$ is necessary here because Liu's theorem is for manifolds with non-negative Ricci curvature outside a compact set). Namely, for any $b>a>\delta>0$ there is $n$ and $r_{0}$ such that for any $r\geq r_{0}$ the annulus $\mathcal{A}(ra,rb)$ can be covered by at most $n$ balls of $g$-radius $r\delta$ centred in the same annulus (equivalently $\mathcal{A}_{r}(a,b)$ can be covered by at most $n$ balls of $g_{r}$-radius $\delta$ centred in the same annulus). Hence any two points $p$ and $q$ in a connected component of $\mathcal{A}_{r}(a,b)$ can be joined through a chain, say $\alpha_{pq}$, of at most $n+2$ radial geodesic segments of the balls of radius $\delta$ covering $\mathcal{A}_{r}(a,b)$. On the other hand Anderson's estimate implies that the $g_{r}$-gradient $|\nabla \ln N|_{r}$ is uniformly bounded (i.e. independent on $r$) on $\mathcal{A}_{r}(a-\delta,a+\delta)$ and therefore uniformly bounded over any curve $\alpha_{pq}$. Integrating $|\nabla \ln N|_{r}$ along the curves $\alpha_{pq}$ and using the bound we arrive at a relevant Harnak estimate controlling uniformly (i.e. independently of $r$) the quotients $N(p)/N(q)$. The estimate is due to Anderson and is summarised in the next Proposition (for further details see, \cite{0264-9381-32-19-195001}).
 
\begin{Proposition}{\rm (Anderson, \cite{MR1809792})}\label{MAXMINU1} Let $(\Sigma;\sg,N)$ be a metrically complete static data set with $\partial \Sigma$ compact, and let $0<a<b$. Then,
\begin{enumerate}
\item\label{anterior} There is $r_{0}$ and $\eta>0$, such that for any $r>r_{0}$ and for any set $Z$ included in a connected component of $\mathcal{A}_{r}(a,b)$ we have,
\be\label{EQHARN}
\max\{N(p):p\in Z\}\leq \eta \min\{N(p):p\in Z\}
\ee
\item\label{posterior} Furthermore, if $r_{i}\rightarrow \infty$ and if $Z_{i}$ is a sequence of sets included, for each $i$, included in a connected component $\mathcal{A}^{c}_{r_{i}}(a,b)$ of $\mathcal{A}_{r_{i}}(a,b)$, and we have, 
\be
\max\{|\nabla \ln N|_{r_{i}}(p):p\in \mathcal{A}^{c}_{r_{i}}(a/2,2b)\}\rightarrow 0
\ee
then,
\be
\frac{\max\{N(p):p\in Z_{i}\}}{\min\{N(p):p\in Z_{i}\}}\rightarrow 1.
\ee
as $i\rightarrow \infty$.
\end{enumerate}
\end{Proposition}
Let $(\Sigma;\hg,U)$ be a static data set in the harmonic presentation (assume $N>0$ and $\partial \Sigma$ compact). We have shown in Part I that $(\Sigma;\hg)$ is metrically complete and that $|\nabla U|_{\hg}^{2}$ decays quadratically. But as $Ric_{\hg}\geq 0$ Liu's ball covering property \cite{MR1216638} also holds on $(\Sigma;\hg)$ provided it is metrically complete and $\partial \Sigma$ is compact. Repeating then Anderson's argument we arrive at the following Harnak estimate but in the harmonic presentation. 
\begin{Proposition}[Anderson, \cite{MR1809792}]\label{MAXMINU} Let $(\Sigma;\hg,U)$ be a metrically complete static data set with $\partial \Sigma$ compact and let $0<a<b$. Then,
\begin{enumerate}
\item\label{anterior} There is $r_{0}>0$ and $\eta>0$, such that for any $r>r_{0}$ and set $Z$ included in a connected component of $\mathcal{A}_{r}(a,b)$ we have,
\be
\max\{U(q):q\in Z\}\leq \eta+\min\{U(q):q\in Z\},
\ee
\item\label{posterior} Furthermore, if $r_{i}\rightarrow \infty$ and if $Z_{i}$ is a sequence of sets included for each $i$ in a connected component $\mathcal{A}^{c}_{r_{i}}(a,b)$ of $\mathcal{A}_{r_{i}}(a,b)$, and we have,  
\be
\max\{|\nabla U|_{r_{i}}(q):q\in \mathcal{A}^{c}_{r_{i}}(a/2,2b)\}\rightarrow 0
\ee
then,
\be
\max\{U(q):q\in Z_{i}\}-\min\{U(q):q\in Z_{i}\}\rightarrow 0
\ee
as $i\rightarrow \infty$.
\end{enumerate}
\end{Proposition}

Both propositions will be used later.

\subsection{Facts about convergence and collapse of Riemannian manifolds}\label{SFCCRM}

In some parts of this article we will use well known techniques in convergence and collapse of Riemannian manifolds. We recall here the concepts and the results that we will use.

We first recall the basic definition of $C^{\infty}$-convergence (the presentation is in the category of tensors adjusted to our needs). We refer the reader to \cite{MR2243772} for more general definitions. 

A sequence of smooth compact Riemannian manifolds with smooth boundary $(M_{i};g_{i})$ converges in $C^{\infty}$ to a smooth compact Riemannian manifold with smooth boundary $(M_{\infty};g_{\infty})$, if there are smooth diffeomorphisms $\phi_{i}:M_{\infty}\rightarrow M_{i}$ such that $\phi^{*}g_{i}$ converges to $g_{\infty}$ in $C^{k}_{g_{\infty}}$ for all $k\geq 0$. That is,
\be
\|\phi_{i}^{*}g_{i}-g_{\infty}\|_{C^{k}_{g_{\infty}}(M_{\infty})}\rightarrow 0  
\ee
where the $C^{k}_{g_{\infty}}(M)$-norm of a smooth tensor field $W$ on a manifold $M$ is defined as usual as,
\be
\|W\|^{2}_{C^{k}_{g}(M)}:=\sup_{x\in M}\bigg\{\sum_{i=0}^{i=k}|\nabla^{(i)}W|^{2}_{g}(x)\bigg\} \quad \text{where}\quad \nabla^{(i)}W=\underbrace{\nabla\ldots\nabla}_{\text{i-times}} W
\ee
To fix ideas, the sequence of Riemannian manifolds,
\be
M_{i}=[1/2,3/4]\times \Sa\times \Sa,\quad g_{i}=(1+x^{i})dx^{2}+d\theta_{1}^{2}+d\theta_{2}^{2}
\ee
converges in $C^{\infty}$ to,
\be
M_{\infty}=[1/2,3/4]\times \Sa\times \Sa,\quad g_{\infty}=dx^{2}+d\theta_{1}^{2}+d\theta_{2}^{2}
\ee

If a sequence of manifolds $(M_{i};g_{i})$ grow in diameter, then there is no convergence in the previous sense but there can be convergence in the pointed sense to a pointed non-compact manifold $(M_{\infty}, p_{\infty};g_{\infty})$. This means that there is a sequence of points $p_{i}\in M_{i}$ and for each compact sub-manifold $N\subset M_{\infty}$ containing $p_{\infty}$, there are diffeomorphisms into the image $\phi_{i}:N\rightarrow M_{i}$ such that $\phi_{i}(p_{\infty})=p_{i}$ and such that $(N;\phi^{*}_{i}g_{i})$ converges in $C^{\infty}$ to $(N,g_{\infty})$. For instance, the sequence of manifolds,
\be
M_{i}=[0,i]\times \Sa\times \Sa,\quad g_{i}=(1+\frac{1}{(1+x)^{i}})dx^{2}+d\theta_{1}^{2}+d\theta_{2}^{2}
\ee
converges in $C^{\infty}$ and in the pointed sense to,
\be
M_{\infty}=[0,\infty)\times \Sa\times \Sa,\quad g_{\infty}=dx^{2}+d\theta_{1}^{2}+d\theta_{2}^{2}
\ee

It can happen that a sequence of manifolds metrically collapses into a manifold of lower dimension. For example, consider the sequence of Riemannian manifolds,
\be
M_{i}=[0,1/2]\times \Sa\times \Sa,\quad g_{i}=dx^{2}+x^{i}d\theta_{1}^{2}+d\theta_{2}^{2}
\ee 
where the coefficient $x^{i}$, over the first factor $\Sa$ tends uniformly to zero as $i\rightarrow \infty$. This sequence of manifolds metrically collapse to the two-dimensional Riemannian manifold,
\be
M_{\infty}=[0,1/2]\times \Sa,\quad g_{\infty}=dx^{2}+d\theta_{2}^{2}.
\ee
Similarly, the sequence of Riemannian manifolds,
\be
M_{i}=[0,1/2]\times \Sa\times \Sa,\quad g_{i}=dx^{2}+x^{i}d\theta_{1}^{2}+x^{i}d\theta_{2}^{2}
\ee 
metrically collapse to the one-dimensional Riemannian manifold,
\be
M_{\infty}=[0,1/2],\quad g_{\infty}=dx^{2}
\ee
that is, to the interval $[0,1/2]$ with the usual metric. Metric collapse means that the Gromov-Hausdoff distance (GH-distance and denoted by $d_{GH}$) between them, as metric spaces, tends to zero (see \cite{MR2243772}). It is a general fact that collapse with bounded curvature is always into a one dimensional manifold, or a two dimensional orbifold. We discuss below the only results that we will use in this respect. 

The context will be always that of metrically complete static data sets $(\Sigma;\hg,U)$ with $\Sigma$ non-compact and $\partial \Sigma$ compact. Let $\gamma$ be a ray emanating from $\partial \Sigma$, that is, an infinite geodesic $\gamma(s)$ such that $\gamma(0)\in \partial \Sigma$ and $\dist(\gamma(s),\partial \Sigma)=s$ (when the data is a static black hole data sets then we assume, because $\hg$ is singular on $\partial \Sigma$, that $\gamma$ is a ray from the boundary of a compact neighbourhood of $\partial \Sigma$). 

The first result we will use is the following. Suppose that for a divergent sequence of points $p_{i}\in \gamma$, the rescaled annuli $(\mathcal{A}^{c}_{r_{i}}(p_{i};a,b);\hg_{r_{i}})$ metrically collapse to $([a,b];dx^{2})$. Note that by Anderson's estimates (see Part I), the collapse is with bounded curvature (and bounded derivatives of the curvature). Then, there is a sequence $\mathcal{B}_{i}$ of neighbourhoods of $\mathcal{A}^{c}_{r_{i}}(p_{i};a,b)$ and finite covers $\tilde{\mathcal{B}}_{i}$ such that $(\tilde{\mathcal{B}}_{i}; \tilde{\hg}_{i})$ converges in $C^{\infty}$ to a $\T^{2}$-symmetric Riemannian space $([a,b]\times \T^{2};\tilde{\hg})$, \footnote{Another way to state this is the following. Given $\epsilon>0$ there are $\delta>0$ and $r_{0}>0$ such that for any $p\in \gamma$ with $r=r(p)\geq r_{0}$, such that the annulus $(\mathcal{A}^{c}_{r}(p;a,b);\hg_{r})$ is $\delta$-close in the GH-distance to the segment $[a,b]$, then there is a neighbourhood $\mathcal{B}$ of $\mathcal{A}^{c}_{r}(p;a,b)$ and a finite cover $\tilde{\mathcal{B}}$ such that $(\tilde{\mathcal{B}};\tilde{\hg}_{r})$ is $\epsilon$-close in $C^{k}$ to a $\T^{2}$-symmetric flat space $([a,b]\times \T^{2};\tilde{\hg})$.}. Here it is important that the points $p_{i}$ belong to $\gamma$ otherwise the existence of such coverings may be not true (this is well known, see examples in \cite{MR3302042} for instance). 

The second result we will use is the following. Suppose that for a divergent sequence of points $p_{i}\in \gamma$, the rescaled annuli $(\mathcal{A}^{c}_{r_{i}}(p_{i};a,b); \hg_{r_{i}})$ metrically collapse, but not into a segment. By Anderson's estimates again, the collapse is with bounded curvature (and bounded derivatives of the curvature). Then, there is a sequence $\mathcal{B}_{i}$ of neighbourhoods of $\mathcal{A}^{c}_{r_{i}}(p_{i};a,b)$ collapsing into a two dimensional Riemannian orbifold with orbifold points of angles $2\pi/2, 2\pi/3, 2\pi/4,\ldots$ Furthermore, if a sequence of points $q_{i}$ converges to a non-orbifold point $q$ then there are neighbourhoods $\mathcal{U}_{i}$ of $q_{i}$ and finite covers $(\tilde{\mathcal{U}}_{i};\tilde{\hg}_{i})$ converging in $C^{\infty}$ to an $\Sa$-symmetric Riemannian manifold, whose quotient by $\Sa$ is isometric to a neighbourhood of the limit point $q$ in the limit Riemannian manifold.

For collapse of two-dimensional manifolds the situation is similar but simpler. We will use the following. Let $(S;q)$ be a non-compact Riemannian manifold with non-empty boundary and let $\gamma$ be a ray from $\partial S$. Let $p_{i}\in \gamma$ be a divergent sequence of points. Suppose that $(\mathcal{A}^{c}_{r_{i}}(p_{i};a,b);q_{i})$ metrically collapses with bounded curvature. Then it does so into an interval $[a,b]$ and there is a sequence of neighborhoods $\mathcal{B}_{i}$ of $\mathcal{A}^{c}_{r_{i}}(p_{i};a,b)$ and finite covers $\tilde{\mathcal{B}}_{i}$, such that $(\tilde{B}_{i};q_{i})$ converges in $C^{\infty}$ to a $\Sa$-symmetric Riemannian manifold, whose quotient by $\Sa$ is $[a,b]$.

The existence of the coverings for each case follows from Theorem 12.1 in \cite{MR1145256}. The orbifold structure when there is two-dimensional metric collapse follows from Proposition 11.5 in \cite{MR1145256}. 

\subsection{The Kasner solutions}\label{SSKK}
\subsubsection{Explicit form and parameters} The Kasner data, denoted by $\mathbb{K}$, are $\mathbb{R}^{2}$-symmetric solutions explicitly given by
\be
\label{KASNERO} \sg=dx^{2} +x^{2\alpha}d y^{2} +x^{2\beta}d z^{2},\quad N=x^{\gamma}
\ee
with $(x,y,z)$ varying in the manifold $\mathbb{R}^{+}\times \mathbb{R}\times \mathbb{R}$, and where $(\alpha,\beta,\gamma)$ satisfy
\be\label{CIRCLE}
\alpha+\beta+\gamma=1\quad \text{and}\quad \alpha^{2}+\beta^{2}+\gamma^{2}=1
\ee
but are otherwise arbitrary (see Figure \ref{Figure21}). The solutions corresponding to two different triples $(\alpha,\beta,\gamma)$ and $(\alpha',\beta',\gamma')$ are equivalent (i.e. isometric) iff $\alpha=\beta'$, $\beta=\alpha'$ and $\gamma=\gamma'$. 

The metrics (\ref{KASNERO}) are flat only when $(\alpha,\beta,\gamma)=(1,0,0), (0,1,0)$ or $(0,0,1)$. We will give them the following names,
\be
A:\ (\alpha,\beta,\gamma)=(1,0,0),\quad C:\ (\alpha,\beta,\gamma)=(0,1,0),\quad B:\ (\alpha,\beta,\gamma)=(0,0,1)
\ee
The solution $B$ is the Boost.

$\mathbb{Z}$-actions, $\mathbb{Z}\times (0,\infty)\times \mathbb{R}^{2}\rightarrow (0,\infty)\times \mathbb{R}^{2}$, are given by fixing a (non-zero) vector field $X$, combination of $\partial_{y}$ and $\partial_{z}$, and letting $n\times p\rightarrow p+nX$. The quotients are ${\rm S}^{1}$-symmetric static solutions. 
Similarly, $\mathbb{Z}^{2}$ quotients give ${\rm S}^{1}\times {\rm S}^{1}$-symmetric static solutions. 
$\mathbb{Z}^{2}$-quotient of the Kasner space will also be called Kasner spaces and denoted too by $\mathbb{K}$. These are the spaces defining Kasner asymptotic.

\subsubsection{The harmonic presentation}

The Kasner spaces in the harmonic presentation are
\be
\label{KHARMAP1} \hg=dx^{2}+x^{2a}dy^{2}+x^{2b}d z^{2},\quad U=c\ln x
\ee
where $a, b$ and $c$ satisfy
\be
2c^{2}+(a-\frac{1}{2})^{2}+(b-\frac{1}{2})^{2}=\frac{1}{2}\quad \text{and}\quad a+b=1
\ee
Thus, the circle (\ref{CIRCLE}), (see Figure \ref{Figure21}), is seen now as an ellipse in the plane $a+b=1$, (see Figure \ref{Figure31}). The $g$-flat solutions $A, B$ and $C$ are,
\be
A:\ (a,b,c)=(1,0,0),\quad  C:\ (a,b,c)=(0,1,0),\quad B:\ (a,b,c)=(1/2,1/2,1/2)
\ee

The Kasner solutions (\ref{KHARMAP1}) are scale invariant. Namely, for any $\lambda>0$, $(\mathbb{R}^{+}\times \mathbb{R}^{2};\lambda^{2}\hg)$ represents the same Kasner space as $(\mathbb{R}^{+}\times \mathbb{R}^{2};\hg)$ does. This can be seen by making the change
\be
{\tt x}=\lambda x,\quad {\tt y}=\lambda^{1-a}y,\quad {\tt z}= \lambda^{1-b}z
\ee
that transforms (\ref{KHARMAP1}) into  
\be\label{KASNER}
\hg=d{\tt x}^{2} +{\tt x}^{2a}d {\tt y}^{2} +{\tt x}^{2b}d {\tt z}^{2},\quad U=c\ln {\tt x} -c\ln \lambda
\ee
Another way to say this is that $(1-2c)t\partial_{t}+x\partial_{x}+(1-a)y\partial_{y}+(1-b)z\partial_{z}$ is a homothetic Killing of the space-time. The scale invariance can of course be seen also in the original space $(\mathbb{R}^{+}\times \mathbb{R}^{2};g,N)$. Note that in general, the isometry that exists between $(\mathbb{R}^{+}\times \mathbb{R}^{2};\hg)$ and $(\mathbb{R}^{+}\times \mathbb{R}^{2};\lambda^{2}\hg)$ does not pass to the quotient by a $\mathbb{Z}\times \mathbb{Z}$-action.    

\subsubsection{Uniqueness}\label{UNIQ}

The Kasner data are the only data with a free $\mathbb{R}\times \mathbb{R}$-symmetry other than the Minkowski data 
\be
\Sigma=\mathbb{R}^{3},\quad g=dx^{2}+dy^{2}+dz^{2},\quad N=1.
\ee
We give now a proof of this fact in a way that becomes useful when we study the Kasner asymptotic later in Section \ref{ENDSAK}.

The proof is as follows. We work in the harmonic presentation $(\Sigma;\hg,U)$, therefore geometric tensors are defined with respect to $\hg$. If the data set $(\Sigma;\hg,U)$ has a free $\mathbb{R}^{2}$-symmetry, and is not the Minkowski solution, then $U$ can be taken as a harmonic coordinate with range in an interval $I$. Then, on $\mathbb{R}^{2}\times I$ we can write
\be\label{UDECESTIM}
\hg=\lambda^{2}dU^{2}+h
\ee
where $\lambda=\lambda(U)$, and where $h=h(U)$ is a family of flat metrics on $\mathbb{R}^{2}$. Without loss of generality assume that $U=0$ at the left end of $I$. Let $(z_{1},z_{2})$ be a (flat) coordinate system on $\mathbb{R}^{2}\times \{0\}$. In the coordinate system $(z_{1},z_{2},U)$ the static equation $Ric_{\hg}=2\nabla U\nabla U$ reduces to
\begin{align}
\label{RTE1} & \partial_{U}h_{AB}=2\lambda \Theta_{AB},\\
\label{RTE2} & \partial_{U}\Theta_{AB}=\lambda(-\theta\Theta_{AB}+2\Theta_{AC}\Theta^{C}_{\ B}),\\
\label{RTE3} & \Theta_{AB}\Theta^{AB}-\theta^{2}=-\frac{2}{\lambda^{2}},
\end{align}
where $\Theta$ is the second fundamental form of the leaves $\mathbb{R}^{2}\times \{U\}$ and $\theta=\Theta_{A}^{\ A}$ is the mean curvature. The static equation $\Delta_{\hg} U=0$ reduces to
\be\label{RTE4} 
\partial_{U}\bigg(\frac{\sqrt{|h|}}{\lambda}\bigg)=0
\ee
where $|h|$ is the determinant of $h_{AB}$. Hence 
\be\label{LEQO}
\Gamma \sqrt{|h|}= \lambda
\ee
for a constant $\Gamma>0$. This can be inserted in (\ref{RTE1})-(\ref{RTE2}) to get the autonomous system of ODE
\begin{align}
\label{RTE21} & \partial_{U}h_{AB}=2\Gamma \sqrt{|h|} \Theta_{AB},\\
\label{RTE22} & \partial_{U}\Theta_{AB}=\Gamma \sqrt{|h|}(-\theta\Theta_{AB}+2\Theta_{AC}\Theta^{C}_{\ B}),
\end{align}
The equation (\ref{RTE3}) transforms into
\be\label{RTE23} 
\Theta_{AB}\Theta^{AB}-\theta^{2}=-\frac{2}{\Gamma^{2}|h|},
\ee
and (it is direct to see) that it holds for all $U$ provided it holds for $U=0$ and (\ref{RTE21}) and (\ref{RTE22}) hold for all $U$. The (\ref{RTE23}) is thus only a `constraint" equation. Therefore the system (\ref{RTE1})-(\ref{RTE3}) is solved by giving $h_{AB}(0), \Theta_{AB}(0)$ and $\Gamma>0$ satisfying (\ref{RTE23}), then running (\ref{RTE21})-(\ref{RTE22}) and finally obtaining $\lambda$ from (\ref{LEQO}).

To solve (\ref{RTE21})-(\ref{RTE22}) first change variables from $U$ to $s$, where $ds=\Gamma\sqrt{|h|}dU$. The system (\ref{RTE21})-(\ref{RTE22}) now reads
\begin{align}
\label{RTE31} & \partial_{s}h_{AB}=2 \Theta_{AB},\\
\label{RTE32} & \partial_{s}\Theta_{AB}=-\theta\Theta_{AB}+2\Theta_{AC}\Theta^{C}_{\ B},
\end{align}
Use these equations to check that
\begin{align}
\label{THE13} & \partial_{s}\theta=-\theta^{2},\\
\label{THE12} & \partial_{s}\Theta_{12}=(\Theta_{11}h^{11}+\Theta_{22}h^{22}-2\Theta_{12}h^{12})\Theta_{12}
\end{align}
Thus, $\theta$ has its own evolution equation which gives $\theta(s)=1/(s+1/\theta(0))$. Moreover if we choose $(z_{1},z_{2})$ on $\{U=0\}$ to diagonalise $h(0)$ and $\Theta(0)$ simultaneously (i.e. $h_{11}(0)=1, h_{22}=1, h_{12}(0)=0$ and $\Theta_{12}(0)=0$), then (\ref{THE12}) shows that $\Theta_{12}=0$ and $h_{12}=0$ for all $s$ and therefore that the evolutions for $h_{11}$ and $h_{22}$ decouple to independent ODEs. With this information it is straightforward to see that the solutions to (\ref{THE13})-(\ref{THE12}), which at the initial times satisfy also (\ref{RTE23}) are only the Kasner solutions.   

We will use all the previous discussion later in Section \ref{ENDSAK}.

\section{Free $\Sa$-symmetric solutions}\label{S1S}

This section studies various aspects of data sets which are free $\Sa$-symmetric. The contents are as follows. Subsection \ref{RDRE} presents the reduced equations, Proposition \ref{PRED}. Subsection \ref{KSOL} discusses the reduced Kasner spaces and subsection \ref{CIGARSOL} describes thoroughly a reduced data that we call the `cigars' (due to their geometric shape). Subsection \ref{CIGUNIQU} proves the cigar's uniqueness and subsection \ref{CIGNHCP} characterises the cigars as the data that model high-curvature regions. These properties of the cigars play an essential role in subsection \ref{DFIAT}, where it is proved that $|\nabla U|^{2}$, $|\nabla V|^{2}$ and $\kappa$ (the Gaussian curvature of $q$), have quadratic decay at infinity on $(\qM; q)$, provided $(\qM;q)$ is metrically complete and $\partial S$ is compact. The discussion of such decay depends on whether the twists $\Omega$ of $\xi$, which is a constant, is zero or not. In the same subsection \ref{DFIAT} it is shown, using the decay previously proved, that $S$ has only a finite number of simple ends, each diffeomorphic to $[0,\infty)\times \Sa$. Furthermore it is proved in Proposition \ref{LPRO} that $U$ has a limit $U_{\infty}$ at infinity, $-\infty\leq U_{\infty}\leq \infty$. Finally, subsection \ref{RDSACL} describes the global structure of reduced data sets arising as collapsed limits that will be relevant to study asymptotic of static ends through scaling.

\subsection{The reduced data and the reduced equations}\label{RDRE}

Let $(\sM;g,N)$ be a static data set invariant under a free $\Sa$-action. The action induces a foliation of $\Sigma$ by $\Sa$-invariant circles. Let $(\Sigma;\hg,U)$ be the harmonic presentation. We will quotient the data $(\Sigma;\hg,U)$ by the Killing field and study the reduced system. 

The complete list of reduced variables and other necessary notation, is the following.
\begin{enumerate}[labelindent=\parindent, leftmargin=*, label={\rm -}, widest=a, align=left]
\item As usual let $\hg=N^{2}g$,
\item let $\xi$ be the Killing field generating the $\Sa$-action.
\item let $\Lambda=|\xi|_{\hg}$ be the $\hg$-norm of $\xi$, 
\item let $\Omega=\epsilon^{\hg}_{abc}\xi^{a}\nabla^{b}\xi^{c}$ be the $\hg$-twist of $\xi$ ($\epsilon^{\hg}$ is the $\hg$-volume form and $\nabla$ any cov. der.),
\item let $U=\ln N$, 
\item let $V=\ln \Lambda$,
\item let $\qM$ be the quotient manifold of $\sM$ by the $\Sa$-action,
\item let $\hqg$ be the quotient two-metric of $\hg$,
\item let $\gcur$ be the Gaussian curvature of $\hqg$.
\end{enumerate}

With all this at hand the following is the definition of a reduced static data set.
\begin{Definition}[Reduced static data set] A data set $(S;q,U,V)$ arising from reducing a $\Sa$-invariant static data set is a reduced static data set. 
\end{Definition}
The next proposition presents the reduced equations of a reduced data set\footnote{We haven't found a reference for these equations though most likely they are given somewhere}. The equations involve only $q$, $U$ and $V$, therefore the tensor $Ric$ and the operators, $\Delta$, $\nabla$ and $\langle\ ,\ \rangle$ are with respect to $q$.

\begin{Proposition}\label{PRED} The (reduced) static equations of a reduced data set $(\qM;q,U,V)$ are,
\begin{align}
\label{ES1} & Ric=\nabla\nabla V +\nabla V\nabla V +\frac{1}{2}\Omega^{2}e^{-4V}q+2\nabla U\nabla U,\\
\label{ES2} & \Delta V +\langle \nabla V,\nabla V\rangle=\frac{1}{2}\Omega^{2}e^{-4V},\\
\label{ES3} & \Delta U +\langle \nabla U,\nabla V\rangle=0.
\end{align}
where $\Omega$ (introduced earlier) is constant. Moreover $\Omega$ is zero iff $\xi$ is hypersurface orthogonal inside $\sM$. 
\end{Proposition}

Before passing to the proof let us make some comments on the reduced equations. 
\vs

- When $\Omega=0$ the system (\ref{ES1})-(\ref{ES2}) is locally equivalent to the Weyl equations around any point where $\nabla \Lambda\neq 0$. We won't use this information however in the rest of the article.
\vs

- The solutions to (\ref{ES1})-(\ref{ES3}) are invariant under the simultaneous transformations
\be\label{TRANSSCA}
q\rightarrow \lambda^{2}q,\quad V\rightarrow V+\frac{1}{2}\ln \nu,\quad U\rightarrow U+\mu, \quad \Omega\rightarrow \frac{\nu}{\lambda} \Omega
\ee
for any $\lambda>0, \nu>0$ and $\mu$. Namely, if we replace $(q,V,U)$ and $\Omega$ in (\ref{ES1})-(\ref{ES3}) for $(\lambda^{2}q, V+\frac{1}{2}\ln \nu, U+\mu)$ and $\nu\Omega/\lambda$ respectively, then the equations are still verified. We will call them simply `scalings" and denote them by $(\lambda,\nu,\mu)$.
\vs

- Given a solution to (\ref{ES1})-(\ref{ES2}), the metric $\hg$ can be recovered using the expression
\be\label{OVERED}
\hg=h_{ab}dx^{a}dx^{b}+\Lambda^{2}(d\varphi+\theta_{i}dx^{i})^{2}
\ee
where $(x^{1},x^{2})$ are coordinates on $S$ and where the one form $\theta$ is found by solving
\be\label{MFORM}
d(\theta_{i}dx^{i}) = \frac{\Omega}{\Lambda^{3}}\sqrt{|q|}dx^{1}\wedge dx^{2}
\ee
where $|q|$ is the determinant of $q_{ij}$ and where $\partial_{\varphi}=\xi$ is the original Killing field. As $\xi$ is the generator of a $\Sa$-action, the range of $\varphi$ is $[0,2\pi)$. Without this information the range of $\varphi$ is undetermined. This is related to the fact that, locally, the reduction procedure requires only that $\xi$ is a non-zero Killing field. If the orbits of $\xi$ do not close up in parametric time $2\pi$, still the reduced equations (\ref{ES1})-(\ref{ES3}) hold, and to recover $\hg$ using (\ref{OVERED}) and (\ref{MFORM}) the right range of $\varphi$ needs to be provided.

This indeterminacy gives rise to two globally inequivalent ways to scale data $(\Sigma;\hg,U;\xi)$ giving rise to the same reduced variables and equations. We assume that $\xi\neq 0$ and has closed orbits. The first is the scaling,
\be\label{FIRSCA}
\hg\rightarrow \lambda^{2}\hg,\qquad \xi\rightarrow \frac{\sqrt{\nu}}{\lambda}\xi
\ee
the second is (recall $\hg=q_{ij}dx^{i}dx^{j}+\Lambda^{2}(d\varphi+\theta_{i}dx^{i})^{2}$),
\begin{align}\label{SECSCA}
\hg\rightarrow \lambda^{2}q_{ij}dx^{i}dx^{j}+\nu\Lambda^{2}(d\varphi +\frac{\lambda}{\nu^{1/2}}\theta_{i}dx^{i})^{2},\quad \xi\rightarrow \xi
\end{align}
In either case, the reduced variables $(q,U,V)$ scale in the same way (\ref{TRANSSCA}). The two new three-metrics are locally isometric but the new length of the orbits of the killing field $\xi$ do not necessarily coincide. The length of the orbits is scaled by $\lambda$ in the first case, and by $\sqrt{\nu}$ in the second case.
\vs

- As in dimension two  we have $Ric=\gcur q$, then (\ref{ES1})-(\ref{ES2}) imply that the Gaussian curvatures acquires the expression
\be\label{KAPPAF}
\gcur =\frac{3}{4}\Omega^{2}e^{-4V}+|\nabla U|^{2}. 
\ee
In particular $\gcur$ is non-negative. This will be an important property when analysing the geometry of the reduced data. 
\vs

The proof of Proposition \ref{PRED} is just computational and relies on formulae in \cite{Dain:2008xr}. We include it for the sake of completeness, but it can be skipped otherwise.

\begin{proof}[Proof of Proposition \ref{PRED}] We use calculations from \cite{Dain:2008xr}, but the notation is different. Precisely we use the following notation: $\qstM$ is the quotient of the spacetime manifold $\stM$ by the $\Sa$-action, $\omega_{a}$ is the twist one form of the Killing field $\xi$ in the spacetime and $\lambda$ its norm. Naturally, we have the commutative diagram
\vs\vs
\begin{center}
\setlength{\unitlength}{.75mm}
\begin{picture}(50,25)
\put(12,5){\vector(0,1){15}}
\put(52,5){\vector(0,1){15}}  
\put(10,0){$\sM$}
\put(50,0){$\stM$}
\put(17,2){\vector(1,0){30}}
\put(30,5){$i_{\sM}$}
\put(10,22){$\qM$}
\put(50,22){$\qstM$}
\put(17,24){\vector(1,0){30}}
\put(30,27){$i_{\qM}$}
\put(7,12){$\pi$}
\put(54,12){$\pi$}
\end{picture}   
\end{center}
where the $\pi$'s are the projections into the quotient spaces and the inclusions $i_{\sM}$ and $i_{\qM}$ are totally geodesic, namely the second fundamental form $K$ of $\sM$ in $\stM$ and the second fundamental form $\chi$ of $\qM$ in $\qstM$, are both zero. Let $\boldsymbol{n}$ be the normal to $\qM$ in $\qstM$.  

Equation (45) from \cite{Dain:2008xr} implies $\boldsymbol{n}(\lambda)=0$ and $i^{*}_{S}\omega_{a}=0$. Using this information inside (18) of \cite{Dain:2008xr} we obtain,
\be
\boldsymbol{\tilde{\nabla}}_{a}\boldsymbol{\tilde{\nabla}}^{a}\lambda=\frac{\omega(\boldsymbol{n})^{2}}{2\lambda^{3}}
\ee
where $\boldsymbol{\tilde{\nabla}}_{a}$ is the covariant derivative of the quotient metric on $\mathcal{N}$. We compute
\be
\boldsymbol{\tilde{\nabla}}_{a}\boldsymbol{\tilde{\nabla}}^{a}\lambda=-\boldsymbol{n}^{a}\boldsymbol{n}^{b}\boldsymbol{\tilde{\nabla}}_{a}\boldsymbol{\tilde{\nabla}}_{b} \lambda +\Delta\lambda=\langle \frac{\nabla N}{N},\nabla\lambda\rangle +\Delta\lambda
\ee
where now $\Delta$ and $\langle\ ,\ \rangle$ are defined with respect to the quotient two-metric over $\qM$ that we denote by $h$. Thus
\be\label{LEQQ}
\Delta\lambda+\langle \frac{\nabla N}{N},\nabla\lambda\rangle=\frac{\omega(\boldsymbol{n})^{2}}{2\lambda^{3}}
\ee 
On the other hand as $N$ is harmonic in $(\sM,g)$ we have
\be\label{NEQ}
\Delta N + \langle \nabla N, \frac{\nabla \lambda}{\lambda}\rangle=0
\ee
where the operators are again with respect to $h$. Finally, the equations (26) and (30) in \cite{Dain:2008xr} give
\be
\gcur_{h}=\frac{\Delta\lambda}{\lambda}+\frac{1}{4}\frac{\omega(\boldsymbol{n})^{2}}{\lambda^{4}}
\ee
where $\gcur_{h}$ is the gaussian curvature of $h$.
Now, $q=N^{2}h$, hence
\be\label{GCUREQ}
N^{2}\gcur=\gcur_{h}-\Delta \ln N=\hat{\gcur}-\frac{\Delta N}{N}+\frac{|\nabla N|^{2}}{N^{2}}
\ee
where again $\Delta$ and $|\ \ |$ are with respect to $h$. Combining (\ref{LEQ}), (\ref{NEQ}) and (\ref{GCUREQ}) we obtain
\be\label{GCUREQII}
\gcur=\frac{3}{4}\frac{\omega(\boldsymbol{n})^{2}}{N^{2}\lambda^{4}}+\frac{|\nabla N|^{2}}{N^{4}}
\ee

Now, the spacetime expression
\be
\partial_{t}^{a}\boldsymbol{\epsilon}_{abcd}\xi^{b}\boldsymbol{\nabla}^{c}\xi^{d}=N\omega(\boldsymbol{n})
\ee
is well known to be constant where $\boldsymbol{\nabla}$ is the spacetime covariant derivative and $\boldsymbol{\epsilon}$ the spacetime volume form (see \cite{MR757180} Theorem 7.1.1). On the other hand
\be\label{OEQ}
\Omega=N\epsilon_{abc}\xi^{a}\nabla^{b}\xi^{c}=\partial_{t}^{a}\boldsymbol{\epsilon}_{abcd}\xi^{b}\boldsymbol{\nabla}^{c}\xi^{d}
\ee
where $\epsilon^{\sg}_{abc}$ is the $\sg$-volume form. Expressing (\ref{LEQQ}), (\ref{NEQ}), (\ref{GCUREQII}) and (\ref{OEQ}) in terms of $U,V$, and expressing the Laplacians and norms in terms of $q$ we obtain (\ref{ES2})-(\ref{ES3}). To obtain (\ref{ES1}) use 
\be
\kappa_{h}h_{ab}=\frac{\nabla_{a}\nabla_{b}\lambda}{\lambda}+\frac{\omega(\boldsymbol{n})^{2}}{2\lambda^{4}}h_{ab}+\frac{\nabla_{a}\nabla_{b}N}{N}
\ee
taken from eqs. (20) and (25) in \cite{Dain:2008xr}, and re-express it in terms of $q_{ab}$ and its covariant derivative. 
\end{proof}

\subsection{Example: the reduced Kasner}\label{KSOL}

The most simple examples of reduced static data sets come from reducing the Kasner solutions through suitable Killing fields. Below we describe the reduced Kasner in detail.

Recall that the Kasner data sets (in the harmonic representation) are
\begin{align}
\hg=dx^{2}+x^{2a}dy^{2}+x^{2b}d z^{2},\qquad U=U_{1}+c\ln x
\end{align}
where $c,a$ and $b$ satisfy $c^{2}+(a-\frac{1}{2})^{2}=\frac{1}{4}$ and $a+b=1$. If we reduce these metrics through the Killing field $\xi=\lambda \partial_{z}$ we obtain the reduced data $(q,U,V)$,
\begin{align}
& q=dx^{2}+x^{2a}d\varphi^{2},\\
& U=U_{1}+c\ln x,\\
& V=V_{1}+b\ln x
\end{align}
where of course
\be
c^{2}+(a-\frac{1}{2})^{2}=\frac{1}{4},\qquad a+b=1.
\ee
and also
\be
\Omega=0
\ee
Above we made $V_{1}=\ln \lambda$, (note that $V_{1}=V(1)$ and that $U_{1}=U(1)$). If we make this solution periodic along $\varphi$ and vary $a$, (hence $b$ and $c$) and $\lambda$ we obtain all the possible reduced solutions with $\Omega=0$ and with a $\Sa$-symmetry (in $\varphi$). 

More general than this we can quotient the Kasner solutions by the Killing field
\be
\xi=\lambda(\cos\omega\ \partial_{y}+\sin\omega\ \partial_{z})
\ee 
for any $\lambda>0$ and $\omega\in [0,2\pi)$, (fixed). A direct calculation shows that the reduced data set $(q,U,V)$ is
\begin{align}
& q=dx^{2}+\bigg[\frac{x^{2}}{x^{2a}\cos^{2}\omega+x^{2b}\sin^{2}\omega}\bigg]d\varphi^{2},\\
& U=U_{1}+c\ln x,\\
& V=V_{1}+\frac{1}{2}\ln (x^{2a}\cos^{2}\omega+x^{2b}\sin^{2}\omega),
\end{align}
where of course
\be
c^{2}+(a-\frac{1}{2})^{2}=\frac{1}{4},\qquad a+b=1.
\ee
and furthermore
\be
\Omega^{2}=4e^{4V_{1}}(a-b)^{2}\cos^{2}\omega\sin^{2}\omega
\ee
Above we made $e^{V_{1}}=\lambda$, (note that $V_{1}=V(1)$ and that $U_{1}=U(1)$). If we make this solution periodic along $\varphi$ and vary $a$, (hence $b$ and $c$) and $\lambda$ and $\omega$, we obtain all the possible reduced solutions with a $\Omega\neq 0$ and with a $\Sa$-symmetry (in $\varphi$).

A simple computation shows that as long as $\Omega\neq 0$ the norm $\Lambda$ of the Killing field $\xi$ grows at least as fast as the square root of the distance to the boundary of the data set. More precisely we have
\be
\Lambda^{2}\geq \eta |\Omega| x
\ee
where $\eta$ does not depend on the data set. As we will see later this is indeed a general property for the asymptotic of any reduced data set.

\subsection{A subclass of the reduced Kasner: the cigars}\label{CIGARSOL}

When either $(a,b)=(1,0)$ or $(a,b)=(0,1)$ and $\omega\notin \{0,\pi/2,\pi,3\pi/2\}$ we obtain an important class of solutions that we will call the {\it cigars} (motivated by their shape, see Figure \ref{Figure1}). Their metrics are complete in $\mathbb{R}^{2}$. After a convenient change of variables, the cigars are given by
\be\label{CIGAR}
U=U_{0},\quad V=V_{0}+\frac{1}{2}\ln (1+r^{2})\quad \text{and}\quad q=4\Omega^{-2}e^{4V_{0}}\big(dr^{2}+\frac{r^{2}}{1+r^{2}}d\varphi^{2}\big)
\ee
where $U_{0}$ and $V_{0}$ are arbitrary constants and where $r$ is the radial coordinate from the origin and $\varphi$ is the polar angle ranging in $[0,2\pi)$, (note that $V_{0}=V(r=0)$). The asymptotic metric is $q=4\Omega^{-2}e^{4V_{0}}(dr^{2}+d\varphi^{2}$), hence cylindrical of section equal to $4\pi \Omega^{-1}e^{2V_{0}}$. 
\vs\vs

\begin{figure}[h]
\centering
\includegraphics[width=10cm,height=.8cm]{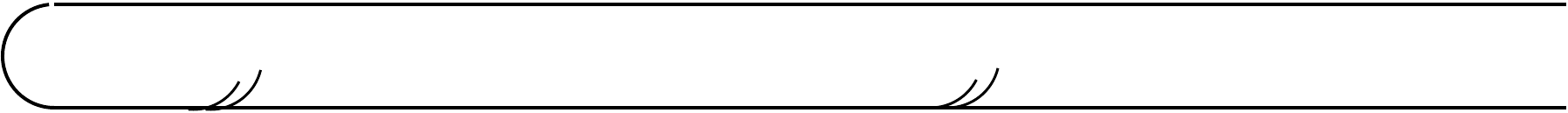}
\caption{Representation of the cigar.}
\label{Figure1}
\end{figure}
\vs\vs

As $U$ is constant, then the lapse $N$ is also constant and the original static solution, (from where the data (\ref{CIGAR}) is coming from), is flat. Let us explain now which quotient of $\mathbb{R}^{3}$ gives rise to the cigars. For any positive $\delta$ we let $T_{\delta}$ be the translation in $\mathbb{R}^{3}$ of magnitude $\delta$ along the $z$-axis and for any $\varphi$ we let $R_{\varphi}$ be the rotation in $\mathbb{R}^{3}$ of angle $\varphi$ around the $z$-axis. Consider the isometric $\mathbb{R}$-action $I$ on $\mathbb{R}^{3}$ given by
\be
I: (t)\times (x,y,z)\longrightarrow T_{te^{V_{0}}}\big( R_{t\Omega (e^{-V_{0}})/2}(x,y,z)\big)
\ee
Now, we quotient $\mathbb{R}^{3}$ as follows: two points $(x,y,z)$ and $(x',y',z')$ are identified iff $(x',y',z')=I(2\pi n,(x,y,z))$ for some $n\in \mathbb{Z}$. The quotient is free $\Sa$-symmetric where the action is by restricting $I$ to $[0,2\pi)$. A straight forward calculation shows that the quotient data $(q,U,V)$ is the cigar solution.  

\subsubsection{The cigars's uniqueness}\label{CIGUNIQU}

The cigars (\ref{CIGAR}) are the only complete non-compact boundary-less solutions to (\ref{ES1})-(\ref{ES3}) with $\Omega\neq 0$. To see this observe that any complete non-compact solution must have $U$ constant because $U$ satisfies
\be\label{UEST}
|\nabla U|(p)\leq \frac{\eta}{\dist(p,\partial \qM)}
\ee
and if $\qM$ is complete and non-compact then $\dist(p,\partial \qM)=\infty$ and $U$ is constant (this decay is direct from Anderson's estimate; We will make another proof of it in Proposition \ref{REDCUR}). Thus, as before, the original static $(\sM; g, N)$ solution is flat (and a $\Sa$-bundle). It is not difficult to see that the only possibility must be a quotient of $\mathbb{R}^{3}$ as described above.  However in Proposition \ref{LLPP} we give an alternative proof whose technique will be useful later when we present the cigar as the singularity model. Before and for the sake of completeness we prove that the only complete (reduced) data set with $\Omega=0$ is flat with $U$ constant.

\begin{Proposition} The only complete boundary-less (reduced) static data with $\Omega = 0$ is flat with $U$ constant .
\end{Proposition}

\begin{proof} As $U=U_{0}$ and $\Omega=0$ then $\nabla\nabla \Lambda=0$ (eq. (\ref{ES1})). This implies that $\Lambda$ is linear along geodesics. Thus, as the space is complete and $\Lambda>0$ then $\Lambda$ must be constant and $q$ flat. The result follows.
\end{proof}

\begin{Proposition}\label{LLPP} The only complete boundary-less (reduced) static data with $\Omega\neq 0$ are the cigars.
\end{Proposition}

\begin{proof} The estimate (\ref{UEST}) shows that $U$ must be constant, i.e. $U=U_{0}$. Hence, making $\overline{\Lambda}=\sqrt{2/\Omega}\ \Lambda$ we have
\begin{align}
\nabla\nabla \overline{\Lambda}=\frac{1}{\overline{\Lambda}^{3}}q,\qquad \kappa=\frac{3}{\overline{\Lambda}^{4}}
\end{align}
The first is an equation of Killing type and can be integrated easily along geodesics. If $\gamma(s)$ is a geodesic parametrised by arc-length then we have $\overline{\Lambda}''=\overline{\Lambda}^{-4}$ (make $\overline{\Lambda}(\gamma(s))=\overline{\Lambda}(s)$) which has the solutions
\be\label{SOLLAM}
\overline{\Lambda}^{2}(s)=\frac{1}{(\overline{\Lambda}{'}_{0}^{2}+1/\overline{\Lambda}_{0}^{2})}\big(1+(\overline{\Lambda}_{0}\overline{\Lambda}'_{0}+(\overline{\Lambda}_{0}'^{2}+1/\overline{\Lambda}_{0}^{2})s)^{2}\big)
\ee
where $\overline{\Lambda}_{0}=\overline{\Lambda}(0)$ and $\overline{\Lambda}_{0}'=\overline{\Lambda}'(0)$. We have thus the bound
\be
\overline{\Lambda}^{2}(s)\geq\frac{1}{(|\nabla \overline{\Lambda}_{0}|^{2}+1/\overline{\Lambda}_{0}^{2})}
\ee
where $|\nabla \overline{\Lambda}_{0}|=|\nabla\overline{\Lambda}|(0)$. This lower bound is achieved only at $s=\overline{\Lambda}_{0}|\nabla \overline{\Lambda}_{0}|/(|\nabla \overline{\Lambda}_{0}^{2}+1/\overline{\Lambda}_{0}^{2})$ on the geodesic that points in the direction of least ${\Lambda}_{0}'$, i.e. when it is equal to $-|\nabla \overline{\Lambda}_{0}|$. Therefore at the point $p$ where the minimum is achieved we have $\nabla\overline{\Lambda}(p)=0$. Hence, along any geodesic $\gamma(s)$ emanating from $p$, (i.e. $\gamma(0)=p$), we have
\be
\overline{\Lambda}^{2}=\overline{\Lambda}^{2}_{0}\bigg(1+\frac{s^{2}}{\overline{\Lambda}_{0}^{4}}\bigg)
\ee
Thus, near $p$ we can write
\be
q=ds^{2}+\ell^{2}d\varphi^{2}
\ee
with $\ell=\ell(s)$ satisfying 
\be
\ell''=-\kappa \ell=-\frac{3}{\overline{\Lambda}^{4}}\ell
\ee
and with $\ell(0)=0$ and $\ell'(0)=1$. The solution is 
\be
\ell^{2}=\frac{s^{2}}{\big(1+s^{2}/\overline{\Lambda}_{0}^{4}\big)}
\ee
recovering (\ref{CIGAR}) at least near $p$. It is simple to see that this $q$ indeed represents the metric all over $S$ which in turn must be diffeomorphic to $\mathbb{R}^{2}$.
\end{proof}

\subsubsection{The cigars as models near high-curvature points}\label{CIGNHCP}

\begin{Lemma}\label{LOCMOD} Let $(S_{i};p_{i};q_{i},V_{i},U_{i})$ be a pointed sequence of metrically complete (reduced) static data sets all having the same $\Omega\neq 0$. Suppose that 
\be
\dist_{q_{i}}(p_{i},\partial S_{i})\geq d_{0}>0
\ee
and that either 
\be
\kappa_{q_{i}}(p_{i})\rightarrow \infty,\quad{\rm or}\quad |\nabla V_{i}|_{q_{i}}(p_{i})\rightarrow \infty
\ee
Then, there are scalings $(\hat{\lambda}_{i},\hat{\nu}_{i},\hat{\mu}_{i})$ such that the
scaled sequence $(S_{i};p_{i};\hat{q}_{i},\hat{V}_{i},\hat{U}_{i})$ converges in $C^{\infty}$ and in the pointed sense to either a flat cylinder or a cigar with the same $\Omega$.
\end{Lemma} 

{\it Notation}: To simplify notation inside the proof, we will use the notation $\kappa_{i}$ for $\kappa_{q_{i}}$ and $|\nabla V_{i}|$ for $|\nabla V_{i}|_{q_{i}}$, (the index `$i$' is from the sequence and of course does not represent a scaling).

\begin{proof} The proof is divided in various cases.

{\it Case I}. Suppose that $|\nabla V_{i}|(p_{i})$ diverges but that $\kappa_{i}(p_{i})$ remains uniformly bounded. To start on we make scalings $(\overline{\lambda}_{i},\overline{\nu}_{i},\overline{\mu}_{i})$ where 
\be
\overline{\lambda}_{i}=|\nabla V_{i}|(p_{i}),\quad \overline{\nu}_{i}=e^{-2V_{i}(p_{i})},\quad \overline{\mu}_{i}=-U_{i}(p_{i}). 
\ee
Let $(\overline{q}_{i},\overline{V}_{i},\overline{U}_{i})$ be the scaled variables. Observe that $\Omega$ scales to $\overline{\Omega}_{i}=\overline{\nu}_{i}\Omega/\overline{\lambda}_{i}$. We have 
\be\label{LEQ}
\overline{\Lambda}_{i}(p_{i})=1,\quad |\nabla \overline{\Lambda}_{i}|(p_{i})=1,
\ee
where, recall, $\overline{\Lambda}_{i}=e^{\overline{V}_{i}}$. Consider now the three-dimensional static pointed data set $(\Sigma_{i};o_{i};\overline{\hg}_{i},\overline{U}_{i})$ whose reductions are the $(S_{i};p_{i};\overline{q}_{i},\overline{V}_{i},\overline{U}_{i})$. The $o_{i}$ are points in $\Sigma_{i}$ projecting into the $p_{i}$'s. Let $\overline{\xi}_{i}$ be the scaling of $\xi_{i}$. In this context the relations (\ref{LEQ}) are
\be
|\overline{\xi}_{i}|(o_{i})=1,\quad |\nabla |\overline{\xi}_{i}||(o_{i})=1,
\ee
where the norms are with respect to $\overline{\hg}_{i}$. Moreover, $\overline{\Omega}_{i}=\overline{\nu}_{i}\Omega/\overline{\lambda}_{i}\rightarrow 0$ because the $\overline{\nu}_{i}$ are bounded and the $\overline{\lambda}_{i}$ tend to infinity. Let us study now the convergence of the derivatives $(\overline{\nabla}\, \overline{\xi}_{i})(o_{i})$ of the Killings $\overline{\xi}_{i}$ at the points $o_{i}$. For notational simplicity we will remove for a moment the subindexes `i' (but we keep them in mind). For the calculation we consider $\overline{\hg}$-orthonormal basis $\{e_{1},e_{2},e_{3}\}$ around the points $o$, with $e_{3}(o)=\overline{\xi}(o)/|\overline{\xi}|(o)$ and $(\overline{\nabla}_{e_{i}}e_{j})(o)=0$. Then, using the relation $\overline{\Omega}=\overline{\epsilon}_{abc}\overline{\xi}^{a}\overline{\nabla}^{b}\overline{\xi}^{c}$ and the Killing condition $\overline{\nabla}_{a}\overline{\xi}_{b}+\overline{\nabla}_{b}\overline{\xi}_{a}=0$, the components of $\overline{\nabla}\, \overline{\xi}$ are computed as,
\begin{align}
& \langle \nabla_{e_{j}}\overline{\xi}, e_{j}\rangle=0,\\
& \langle \nabla_{e_{1}}\overline{\xi}, e_{2}\rangle=-\langle \nabla_{e_{2}}\overline{\xi},e_{1}\rangle=\frac{\overline{\Omega}}{|\overline{\xi}|},\\
& \langle \nabla_{e_{3}}\overline{\xi}, e_{j}\rangle=-\langle \nabla_{e_{j}}\overline{\xi}, e_{3}\rangle=-\nabla_{e_{j}}|\overline{\xi}|.
\end{align} 
If furthermore $e_{1}(o)$ and $e_{2}(o)$ are chosen such that $\nabla_{e_{1}(o)}|\overline{\xi}|=0$ and $\overline{\nabla}_{e_{2}(o)}|\overline{\xi}|=1$ then, (restoring now the indexing `i'), the components $\langle\overline{\nabla}_{e_{j}}\overline{\xi}_{i},e_{k}\rangle(o_{i})$ are either zero or tend to zero as $i$ goes to infinity except for $\langle\overline{\nabla}_{e_{1}}\overline{\xi}_{i},e_{3}\rangle(o_{i})$ and $\langle \overline{\nabla}_{e_{3}}\overline{\xi}_{i},e_{1}\rangle(o_{i})$ that are constant and equal to one and minus one respectively.  

Now we observe that 
\be
\dist_{\overline{\hg}_{i}}(o_{i},\partial \Sigma_{i})=\overline{\lambda}_{i}\dist_{\hg_{i}}(o_{i},\partial \Sigma_{i})=\overline{\lambda}_{i}\dist_{q_{i}}(p_{i},\partial S_{i})\geq \overline{\lambda}_{i}d_{0}\rightarrow \infty.
\ee
Therefore by Anderson's estimates, the curvature of the $\overline{\hg}_{i}$ over balls of centers $o_{i}$ and any fixed radius tend to zero. Hence, there are neighbourhoods $\mathcal{B}_{i}$ of $o_{i}$ and covers $\tilde{\mathcal{B}}_{i}$ such that the pointed sequence $(\tilde{\mathcal{B}}_{i};\tilde{o}_{i};\overline{\hg}_{i})$ converges in $C^{\infty}$ and in the pointed sense to the Euclidean three-space (for the cover metric we use also $\overline{\hg}_{i}$). We claim that the lift of the Killing fields $\overline{\xi}_{i}$ to the $\tilde{\mathcal{B}}_{i}$, (that we will denote too by $\overline{\xi}_{i}$) converge in $C^{\infty}$ to the generator of a (non-trivial) rotation of $\mathbb{R}^{3}$. To see this recall first that for any Killing field $\chi$ it holds $\nabla_{a} \nabla_{b} \chi_{c}=-Rm_{bca}^{\ \ \ d}\chi_{d}$. Thus, at any point $x$ we can find $\overline{\xi}_{i}(x)$ by integrating a second order linear ODE along a geodesic that extends from $\gamma(0)=\tilde{o}_{i}$ to $x$, given the initial data $\overline{\xi}_{i}(\gamma(0))$ and $\overline{\nabla}_{\gamma'(0)}\overline{\xi}_{i}$. As it was shown earlier that the data $\overline{\xi}_{i}(\tilde{o}_{i})$ and $(\overline{\nabla}\overline{\xi}_{i})(\tilde{o}_{i})$ converges, hence so does $\overline{\xi}_{i}$ and the perpendicular distribution of the limit Killing field $\overline{\xi}_{\infty}$ is integrable because $\lim \overline{\Omega}_{i}=0$. Thus, $\overline{\xi}_{\infty}$ generates a rotation in $\mathbb{R}^{3}$. As $|\overline{\xi}_{\infty}|(\tilde{o}_{\infty})=1$ and $|\nabla |\overline{\xi}_{\infty}||(\tilde{o}_{\infty})=1$ it must be that $\tilde{o}_{\infty}$ is at a distance one from the  rotational axis. In coordinates $(x,y,z)$ of $\mathbb{R}^{3}$ the limit vector field would be, (for instance), $x\partial_{y}-y\partial_{z}$ and the limit point would be, (for instance), $(1,0,0)$. 

This convergence of $\overline{\xi}_{i}$ to the generator of a rotation will be used in the following to extract a pair of relevant informations.

First we show that inside the surfaces $S_{i}$ there are geodesic loops $\ell_{i}$, based at the points $p_{i}$, whose $\overline{q}_{i}$ length tends to zero. Let us see this. For $i$ large enough, the orbit of the Killing $\overline{\xi}_{i}$ inside $\tilde{\mathcal{B}}_{i}$, that starts at the point $\tilde{o}_{i}$, twists around an `axis' and come very close to close up into a circle when it approaches again the point $\tilde{o}_{i}$ (see Figure \ref{S11}). Hence, a small two-dimensional disc formed by short geodesic segments emanating perpendicularly to $\overline{\xi}_{i}(\tilde{o}_{i})$ at $\tilde{o}_{i}$ must intersect the orbit at a nearby point $\tilde{o}'_{i}$. Moreover the geodesic segment joining $\tilde{o}_{i}$ and $\tilde{o}'_{i}$, projects into a geodesic loop $\ell_{i}$ on $S_{i}$ based at $p_{i}$. The length of the loops $\ell_{i}$ clearly tend to zero as i goes to infinity.

Second, for $i\geq i_{0}$ large enough, the norm of the Killings $\overline{\xi}_{i}$ over the balls $B_{\overline{\hg}_{i}}(o_{i},1/2)$ $\subset \tilde{\mathcal{B}}_{i}$ is bounded below by $1/4$. Hence, $\overline{\Lambda}_{i}$ is bounded below by $1/4$ over the balls $B_{\overline{q}_{i}}(p_{i},1/2)$ in $S_{i}$. More importantly the Gaussian curvature $\overline{\kappa}_{i}$ is bounded above by $100\overline{\Omega}_{i}^{2}$ also on $B_{\overline{q}_{i}}(p_{i},1/2)$.

\begin{figure}[h]
\centering
\includegraphics[width=9cm, height=5.5cm]{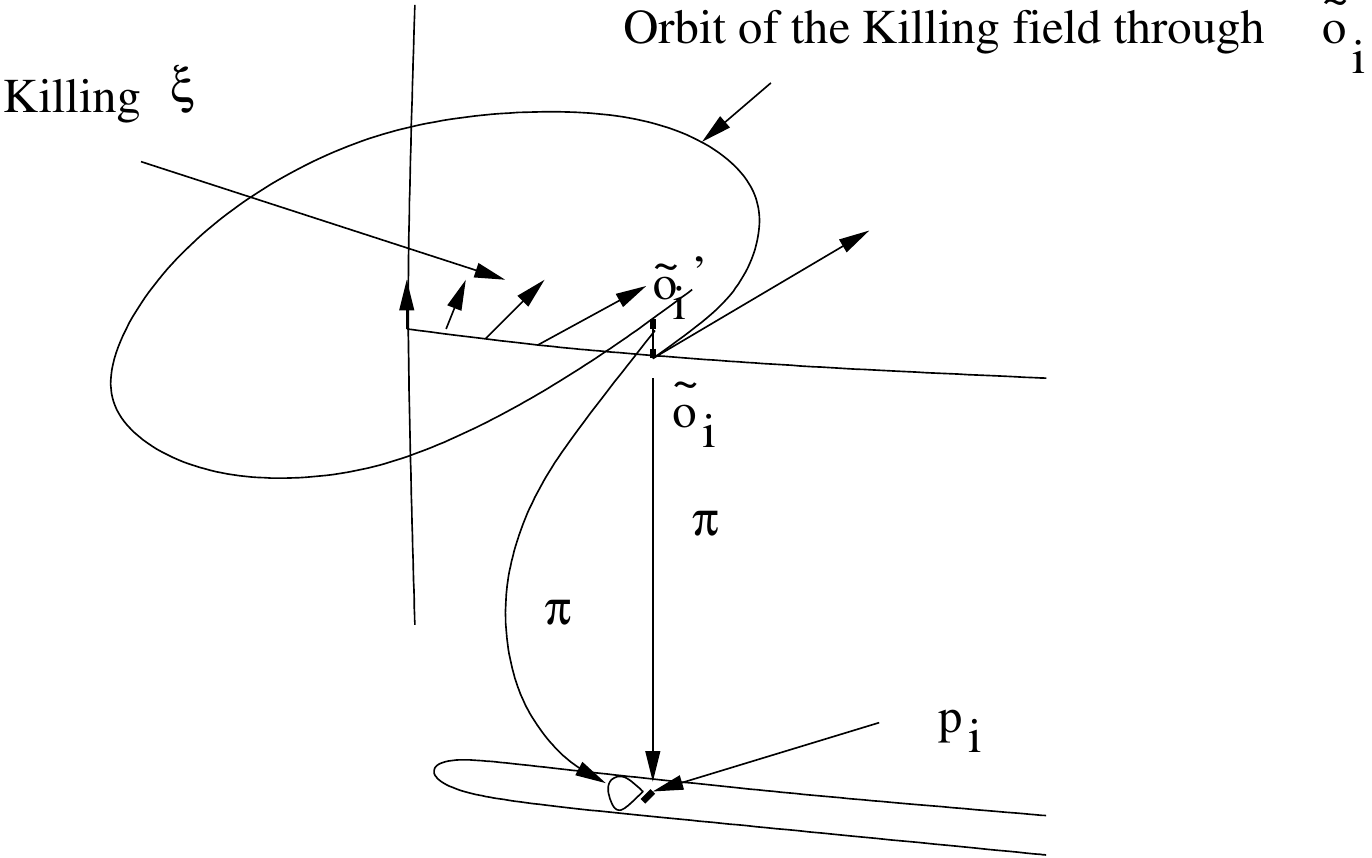}
\caption{}
\label{S11}
\end{figure}

From these two facts we conclude that the geometry near the points $o_{i}$ is collapsing with bounded curvature. This implies that if we scale up $\overline{q}_{i}$ to have the injectivity radius at $o_{i}$ equal to one, then the new scaled spaces converge in the pointed sense to a flat cylinder.
The composition of this last scaling and the one we performed first is the scaling $(\hat{\lambda}_{i},\hat{\nu}_{i},\hat{\mu}_{i})$ we were looking for.

{\it Case II}. Suppose now that both $|\nabla V_{i}|(p_{i})$ and $\kappa_{i}(p_{i})$ are diverging. If the quotient $\kappa_{i}(p_{i})/|\nabla V_{i}|^{2}(p_{i})$ tends to zero, then we can perform a scaling $(\overline{\lambda}_{i},\overline{\nu}_{i},\overline{\mu}_{i})$ 
that leaves $\Omega$ invariant and that makes $\overline{\kappa}_{i}(p_{i})$ bounded and $|\nabla \overline{V}_{i}|(p_{i})$ diverging. We can then repeat the step in {\it Case I} with $(\overline{q}_{i},\overline{V}_{i},\overline{U}_{i})$ instead of $(q_{i},V_{i},U_{i})$ to prove the Lemma in this case too. 

Assume therefore that the quotient $\kappa_{i}(p_{i})/|\nabla V_{i}|^{2}(p_{i})$ remains bounded. Perform again a scaling $(\overline{\lambda}_{i}, \overline{\nu}_{i},\overline{\mu}_{i})$ that leaves $\Omega$ invariant and makes $\overline{\kappa}_{i}(p_{i})=1$ and therefore makes $|\nabla \overline{V}_{i}|(p_{i})$ bounded because $\kappa_{i}(p_{i})/|\nabla V_{i}|^{2}(p_{i})$ is invariant. Note that as $\dist_{\overline{q}_{i}}(p_{i},\partial S_{i})\rightarrow \infty$, the estimate (\ref{UDECESTIM}) impose that $|\nabla \overline{U}_{i}|$ must tend uniformly to zero over balls of centers $p_{i}$ and fixed but arbitrary radius. We claim that the curvature $\overline{\kappa}_{i}$ remains uniformly bounded on balls of centers $p_{i}$ and fixed radius. Let $L>0$, let $x$ be a point in $B_{\overline{q}_{i}}(p_{i},L)$ and let $\gamma(s)$ be a length-minimising geodesic joining $p_{i}$ to $x$. Let $\overline{\Lambda}_{i}(s)=\overline{\Lambda}_{i}(\gamma(s))$. Then, the value of $\overline{\Lambda}_{i}$ at $x$ is found by solving the second order ODE
\be
\overline{\Lambda}''_{i}=\frac{\Omega^{2}}{4\overline{\Lambda}^{3}}+(|\nabla U|^{2}-2U'^{2})\overline{\Lambda}
\ee
subject to the initial data $\overline{\Lambda}_{i}(0)=\overline{\Lambda}_{i}(\gamma(0))$ and $\overline{\Lambda}'_{i}(0)=\nabla_{\gamma'(0)}\overline{\Lambda}_{i}$, and evaluating at $s=\dist_{\overline{q}_{i}}(x,p_{i})$. If $\nabla \overline{U}_{i}$ were identically zero then the solutions would be exactly (\ref{SOLLAM}) and we would have the bound
\be
\overline{\Lambda}_{i}^{2}(s)\geq \frac{1}{(\overline{\Lambda}{'}_{i}(0))^{2}+1/(\overline{\Lambda}_{i}(0))^{2}}
\ee
for all $s\geq 0$. In particular, if $\overline{\Lambda}_{i}(0)$ is bounded below by $A$ and $|\overline{\Lambda}_{i}'(0)|$ is bounded above by $B$ then $\overline{\Lambda}_{i}(s)$ is bounded below by $\sqrt{1/(B^{2}+1/A^{2})}$. But as $|\nabla \overline{U}_{i}|$ tends to zero uniformly over balls or radius $L$, then the solutions to the ODE tend to (\ref{SOLLAM}) with initial data $\overline{\Lambda}_{i}(0)$ and $\overline{\Lambda}'_{i}(0)$. Now, as $\overline{\kappa}_{i}(p_{i})=1$ and $|\nabla \overline{V}_{i}|(p_{i})$ is bounded, there are constants $A$ and $B$ such that
\be
\overline{\Lambda}_{i}(0)\leq A,\quad {\rm and}\quad |\overline{\Lambda}'_{i}(0)|\leq B
\ee
no matter which the geodesic $\gamma$ is. Therefore if $i\geq i_{0}(L)$ is big enough then $\overline{\Lambda}_{i}(x)\leq 2\sqrt{1/(B^{2}+1/A^{2})}$. Hence, $\overline{\kappa}_{i}\geq 3\overline{\Omega}_{i}(B^{2}+1/A^{2})^{2}/32$ everywhere on $B_{\overline{q}_{i}}(p_{i},L)$.

The bound we proved for the curvature implies that if for a certain subsequence the injectivity radius at the points $p_{i}$ tends to zero then there are finite covers that converge to a cigar. But this is impossible because the cigars do not admit any non-trivial quotient. Hence the injectivity radius remains bounded away from zero and the pointed sequence $(S_{i};p_{i};\overline{q}_{i},\overline{V}_{i},\overline{U}_{i})$ must sub-converge in the pointed sense to a solution with $U$ constant. By uniqueness it is always a cigar and we are done.
\end{proof}

Let us make an extra observation about a construction made inside the proof. Recall that the spaces $(\tilde{\mathcal{B}}_{i},\overline{\hg}_{i})$ converge to $\mathbb{R}^{3}$ and the Killings $\overline{\xi}_{i}$ converge to the generator of a rotation. Let $z_{i}$ be points where $(\nabla |\overline{\xi}_{i}|)(z_{i})=0$. These points one can think that lie in the `axis' of rotation. Naturally if we quotient the balls of centers $z_{i}$ and radius two we obtain a two-disc. This disc projects into a `cup' on $S_{i}$ containing $p_{i}$ (see Figure \ref{S11}). In the metric $q_{i}$, the `radius' of this cup (i.e. the maximum distance from a point to the boundary) goes to zero. 

The Lemma \ref{LOCMOD} provides models for the scaled geometry near points of high curvature or high $V$-gradient, but it does not say how such points affect the unscaled geometry nearby. This is an important information that we will need later. In rough terms, what occurs is that at any finite distance from such a point the (unscaled) geometry becomes one dimensional, pretty much like a cigar highly scaled down. The next Lemma \ref{66} explains the phenomenon. In few words it explains how the geometry looks like near geodesics that join points of high curvature or high $V$-gradient and the boundary of the surfaces $S_{i}$. This basic information will be sufficient to extract conclusions later. 

The scaled geometry around points in such geodesics will be model essentially as regions of the cigar whose curvature at the origin is conventionally $\kappa_{0}=3(2\pi)^{2}$ and therefore whose metric is
\be
q_{0}=\frac{1}{(2\pi)^{2}}\big(dr^{2}+\frac{r^{2}}{1+r^{2}}d\varphi^{2}\big)
\ee
where $r\geq 0$. Let us describe the models more explicitly. A pointed space $(\{0\leq r\leq 40\};x;q_{0})$, where $x$ be a point in this cigar with $r(x)\leq 25$, is a model of type ${\rm Ci}$ (from `cigar"). A pointed space $(\{r(x)-10\leq r\leq r(x)+10\};x;q_{0})$, where $x$ be a point with $r(x)>25$, is a model of type ${\rm Cy}$ (from `cylinder"). The Figure \ref{S12} sketches these two types of models.

\begin{figure}[h]
\centering
\includegraphics[width=9cm, height=5cm]{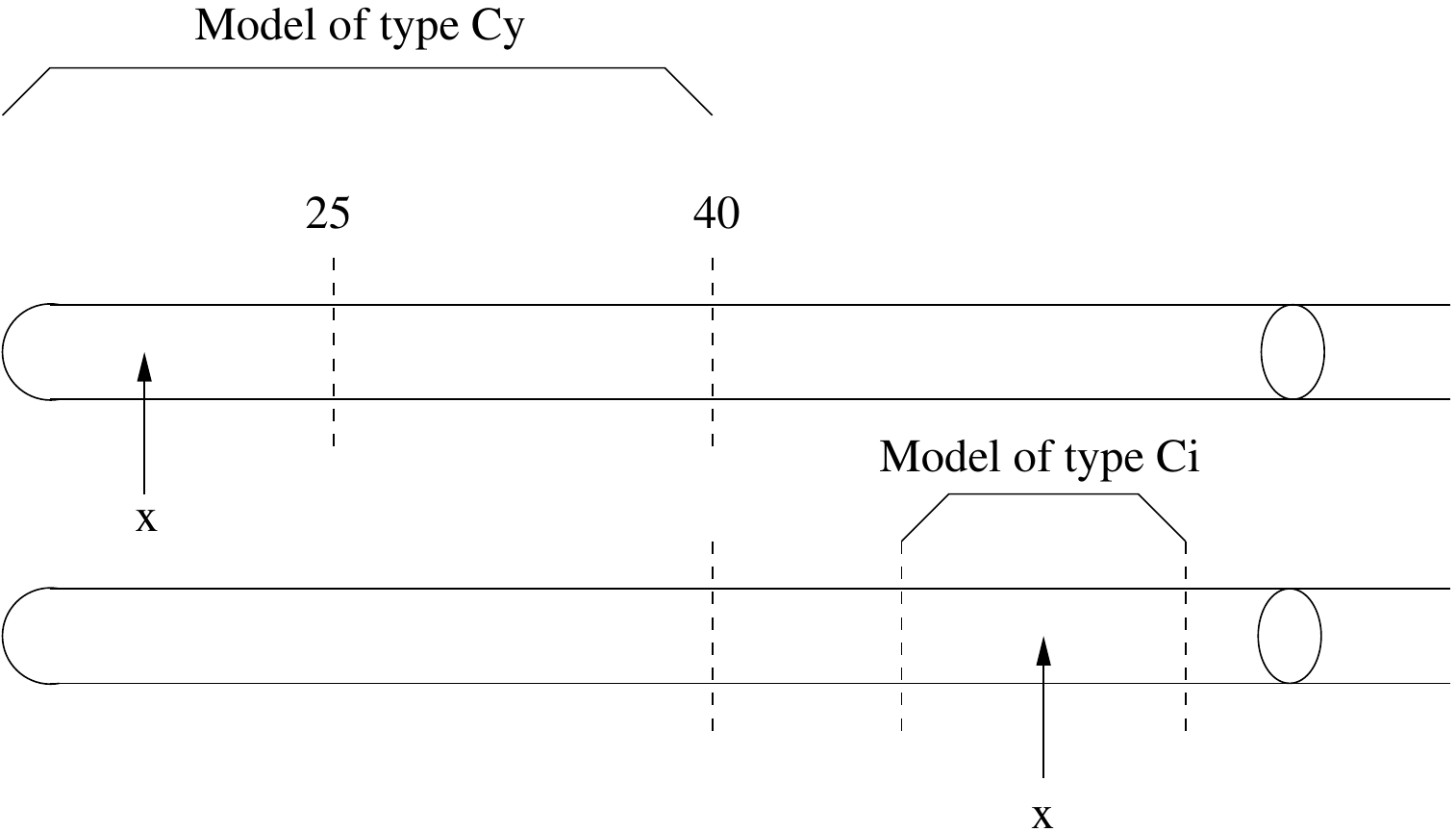}
\caption{}
\label{S12}
\end{figure}

\begin{Lemma}\label{66} Let $(S_{i};p_{i};q_{i},V_{i},U_{i})$ be a pointed sequence of metrically complete (reduced) static data sets all having the same $\Omega\neq 0$ and suppose that 
\be
\dist_{q_{i}}(p_{i},\partial S_{i})\geq d_{0}>0
\ee
and that either 
\be
\kappa_{q_{i}}(p_{i})\rightarrow \infty,\quad{\rm or}\quad |\nabla V_{i}|_{q_{i}}(p_{i})\rightarrow \infty.
\ee 
For every $i$ let $\gamma_{i}$ be a geodesic segment joining $p_{i}$ to $\partial S_{i}$ and minimising the distance between them (if $\partial S_{i}=\emptyset$ let $\gamma_{i}$ be an infinite ray). Fix a positive $d_{1}$ less than $d_{0}$.

Then, for every $k\geq 1$, $\epsilon>0$ there exists $i_{0}$ such that for any $i\geq i_{0}$ and for any $x_{i}\in \gamma_{i}$ with $\dist_{q_{i}}(x_{i}, p_{i})\leq d_{1}$ there exist a neighbourhood $\mathcal{B}_{i}$ of $x_{i}$ and a scaled metric $\overline{q}_{i}=\overline{\lambda}_{i}^{2}q_{i}$ such that $(\mathcal{B}_{i};x_{i}; \overline{q}_{i})$ is $\epsilon$-close in $C^{k}$ to either a model space ${\rm Ci}$ or a model space ${\rm Cy}$.
\end{Lemma}  

{\it Notation}: Again to simplify notation inside the proof, we will use the notation $\kappa_{i}$ for $\kappa_{q_{i}}$ and $|\nabla V_{i}|$ for $|\nabla V_{i}|_{q_{i}}$.

\begin{proof} Half of the work has been done essentially already in Lemma \ref{LOCMOD} because the geometry near points of high curvature or high $V$-gradient are model locally (at a right scale) by a space Ci or a space Cy. We say this formally as follows: given $\epsilon>0$ and $k\geq 1$ there are $K_{0}>0$ and $i_{1}>0$ such that for any $i\geq i_{1}$ and $x_{i}\in \gamma_{i}$ such that $\dist_{i}(x_{i},p_{i})\leq d_{1}$ and either $\kappa_{i}(x_{i})\geq K_{0}$ or $|\nabla V_{i}|(x_{i})\geq K_{0}$, then the conclusions of the Lemma hold. Thus, it is left to show that the conclusions hold too for points on $\gamma_{i}$ that do not have `high" curvature or high gradient, that is for which $\kappa_{i}(x_{i})\leq K_{0}$ and $|\nabla V_{i}|(x_{i})\leq K_{0}$. We prove that in what follows.

We will show that there is $i_{2}\geq i_{1}$ such that for any $i\geq i_{2}$  and for any $x_{i}\in \gamma_{i}$ such that $\dist_{i}(x_{i},p_{i})\leq d_{1}$, $\kappa_{i}(x_{i})\leq K_{0}$ and $|\nabla V_{i}|(x_{i})\leq K_{0}$, the conclusion of the Lemma also holds and the local model is of type Cy. 

Given $i$, let $x_{i}$ be a point such that $x_{i}\in \gamma_{i}$ such that $\dist_{i}(x_{i},p_{i})\leq d_{1}$, $\kappa_{i}(x_{i})\leq K_{0}$ and $|\nabla V_{i}|(x_{i})\leq K_{0}$. We begin claiming that there are $r_{0}<(d_{0}-d_{1})/2$ and $K_{1}>0$ independent of $i$ such that $\kappa_{i}(x)\leq K_{1}$ for all $x\in B_{q_{i}}(x_{i},r_{0})$. Let $r_{0}$ be any number less than $(d_{0}-d_{1})/2$ and let $x$ be a point such that $\dist(x,x_{i})\leq r_{0}$. Let $\alpha_{i}(s)$ be a length minimising geodesic joining $x_{i}$ to $x$ ($\alpha_{i}(0)=x_{i}$). Denote $V_{i}(s):=V_{i}(\alpha_{i}(s))$. Let,
\be
\hat{V}_{i}(s)=V_{i}(s)-V_{i}(0)
\ee
Then we have,
\be\label{INITIALDATAC}
\hat{V}_{i}(0)=0,\quad {\rm and}\quad |\hat{V}_{i}'(0)|\leq K_{0}
\ee
where the first equation is by the definition of $\hat{V}_{i}(0)$ and the second follows by assumption. On the other hand $\hat{V}_{i}(s)$ satisfies the differential equation (\ref{ES1}), namely,
\be\label{ODEFORV}
\hat{V}_{i}''+\hat{V}_{i}{'}^{2}=\big(\frac{1}{2}\Omega^{2}e^{-4V_{i}(0)}\big)e^{-4\hat{V}_{i}}+(|\nabla U|^{2}-2U'^{2})
\ee
where the last expression in parenthesis is evaluated of course on $\alpha_{i}(s)$. 

Let us make two comments on this equation. First, the coefficient $\Omega^{2}e^{-4V_{i}(0)}/2$ is less or equal than $\kappa_{i}(x_{i})$ and thus less or equal than $K_{0}$ by assumption. Second, the summand $(|\nabla U|^{2}-2U'^{2})(s)$ is uniformly bounded, say by $K_{2}>0$, independently of $s$, $x$, $x_{i}$ and $i$. This follows from the estimate (\ref{UEST}) and $\dist_{q_{i}}(\alpha_{i}(s),\partial S_{i})\geq (d_{1}-d_{0})/2$; This last inequality is due to,
\be
\dist_{q_{i}}(\alpha_{i}(s),\partial S_{i})\geq \dist_{q_{i}}(x_{i},\partial S_{i})-\dist_{q_{i}}(\alpha_{i}(s),x_{i})
\ee
and the inequalities $\dist_{q_{i}}(x_{i},\partial S_{i})\geq (d_{1}-d_{0})$ and $\dist_{q_{i}}(\alpha_{i}(s),x_{i})\leq \dist_{q_{i}}(x,x_{i})\leq (d_{1}-d_{0})/2$. 

Until now we have shown control on the ODE (\ref{ODEFORV}) and the initial data (\ref{INITIALDATAC}). Therefore by standard ODE analysis, it follows that one can chose $r_{0}$ small enough such that $|\hat{V}_{i}(s)|\leq K_{1}$, (i.e. preventing blow up), for a $K_{1}$ independent on $s$, $x$, $x_{i}$ and $i$. This bound on $V_{i}(x)$ (we removed the hat now) and the bound on $|\nabla U|^{2}(x)$ gives the desired bound on $\kappa_{i}(x)$.

We have proved a curvature bound $\kappa_{i}(x)\leq K_{1}$ for all $x\in B_{q_{i}}(x_{i},r_{0})$. Using this bound we are going to show that the injectivity radius at $x_{i}$, namely $inj_{q_{i}}(x_{i})$, tends to zero as $i$ tends to infinity. Indeed, if on the contrary $inj_{q_{i}}(x_{i})\geq r_{1}>0$ for some $r_{1}>0$, then because the curvature is bounded on $B_{q_{i}}(x_{i}, r_{0})$, there is $v>0$ and $r_{2}\leq \min\{r_{0},r_{1}\}/2$ such that the area of the ball $B_{q_{i}}(x_{i},r_{2})$ is greater or equal than $v$. As $B_{q_{i}}(x_{i},r_{2})\subset B_{q_{i}}(p_{i},d_{0})$ then we have
\be
\frac{\area_{i}(B_{q_{i}}(p_{i},d_{0}))}{d_{0}^{2}}\geq \frac{v}{d_{0}^{2}}
\ee
On the other hand observe that by Lemma \ref{LOCMOD} the geometry near the points $p_{i}$ is locally collapsing (at a right scale) to a line or to half a line. Thus, there is $i_{3}$ such that for $i\geq i_{3}$ there is $\delta_{i}\rightarrow 0$, such that the quotient 
\be
\frac{\area_{i}(B_{q_{i}}(p_{i},\delta_{i}))}{\delta_{i}^{2}}
\ee
is less or equal than $v/(2d_{0}^{2})$ (in fact the quotient tends to zero). But by Bishop-Gromov the function 
\be
s\rightarrow \frac{\area_{i}(B_{q_{i}}(p_{i},s))}{s^{2}}
\ee
is monotonically decreasing and therefore we should have
\be
\frac{v}{2d_{0}^{2}}\geq \frac{\area_{i}(B_{q_{i}}(p_{i},d_{0}))}{d_{0}^{2}}\geq \frac{v}{d_{0}^{2}}
\ee
which is impossible. Thus the injectivity radius at $x_{i}$ tends to zero. Therefore the balls $B_{q_{i}}(x_{i},r_{0})$ collapse with bounded curvature 
 and the existence of a scaling whose limit is a cylinder (Cy) is now direct.
\end{proof}

The Lemma \ref{66} gives a local model for the collapsed geometry around points on the geodesics $\gamma_{i}$. The concatenation of the local models provide a global picture that is summarised in the next corollary (whose proof is now direct), see Figure \ref{S13}. 

\begin{figure}[h]
\centering
\includegraphics[width=7cm, height=5cm]{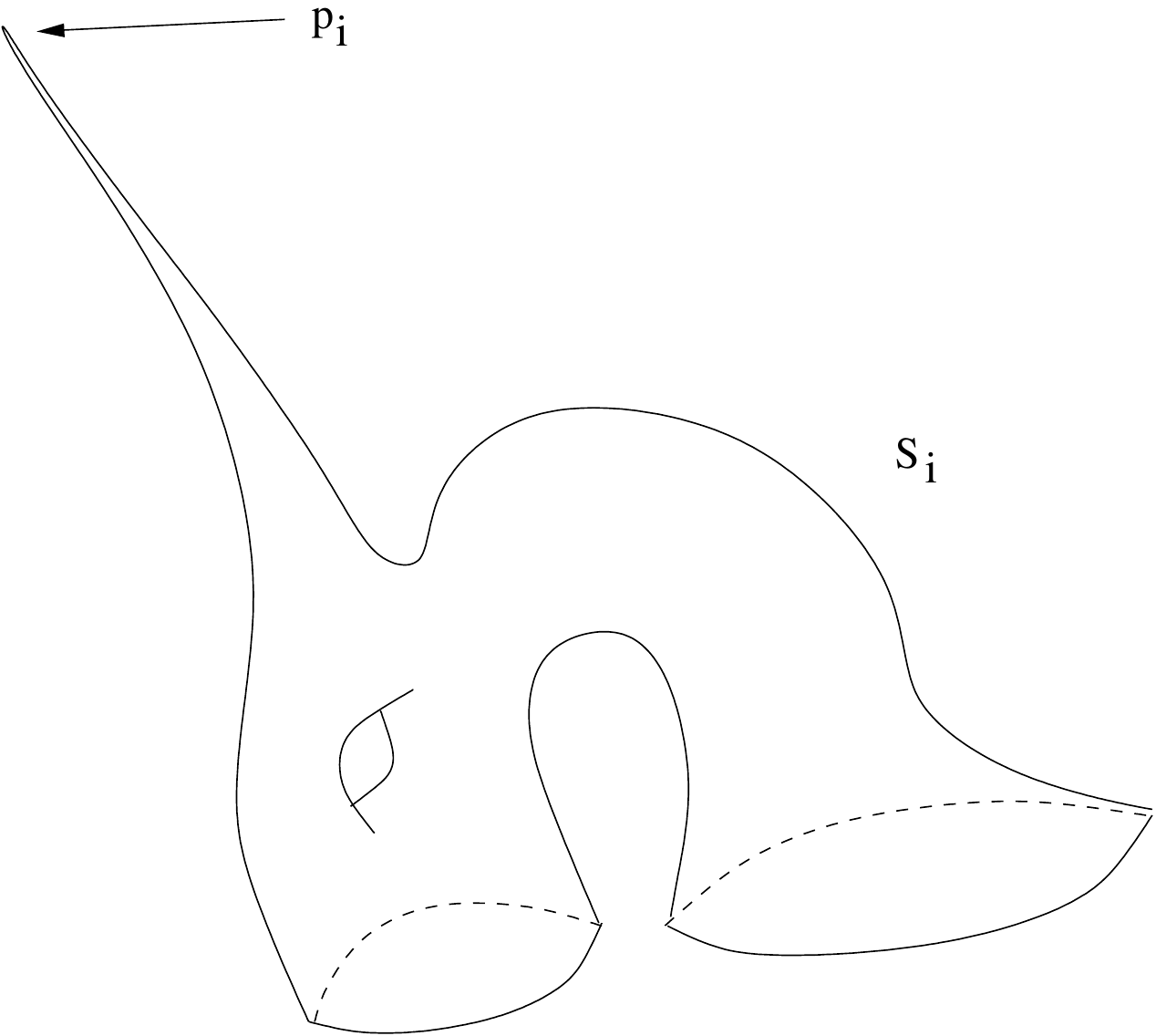}
\caption{}
\label{S13}
\end{figure}

\begin{Corollary}\label{SINGMODELUNDER} Let $(S_{i};p_{i};q_{i},V_{i},U_{i})$ be a pointed sequence of metrically complete (reduced) static data sets all having the same $\Omega\neq 0$ and suppose that 
\be
\dist_{q_{i}}(p_{i},\partial S_{i})\geq d_{0}>0
\ee
and that either 
\be
\kappa_{q_{i}}(p_{i})\rightarrow \infty,\quad{\rm or}\quad |\nabla V_{i}|_{q_{i}}(p_{i})\rightarrow \infty.
\ee 
For every $i$ let $\gamma_{i}$ be a geodesic segment joining $p_{i}$ to $\partial S_{i}$ and minimising the distance between them (if $\partial S_{i}=\emptyset$ let $\gamma_{i}$ be an infinite ray). Fix a positive $d_{1}$ less than $d_{0}$.

Then there is $i_{0}$ such that for any $i\geq i_{0}$ there is a neighbourhood $\mathcal{B}_{i}$ of the ball $B_{q_{i}}(p_{i},d_{1})$, diffeomorphic to a disc and metrically collapsing to a segment of length $d_{1}$ as $i$ goes to infinity.
\end{Corollary}

\subsection{Decay of the fields at infinity and asymptotic topology}\label{DFIAT}

We know already that the gradient of $U$ decays quadratically at infinity. In this section we show that also de gradient of $V$ and the Gaussian curvature $\kappa$ decay quadratically. The proof depends on whether $\Omega$ is zero or not. The case $\Omega=0$ is simple and relies only on the techniques a la Bakry-\'Emery used earlier. As a by product we re-prove the quadratic decay of the gradient of $U$, valid when $\Omega=0$ or not. When $\Omega\neq 0$, the proof requires the use of Corollary \ref{SINGMODELUNDER}. 

\subsubsection{Case $\Omega=0$}\label{DAIFS}

\begin{Proposition}\label{REDCUR} There is a constant $\eta>0$ such that for every metrically complete (reduced) static data set we have
\be\label{DEC1}
|\nabla U|^{2}(p)\leq\frac{\eta}{\dist^{2}(p,\partial \qM)}.
\ee
Moreover when $\Omega=0$ we have
\be\label{DEC2}
|\nabla V|^{2}(p)\leq\frac{\eta}{\dist^{2}(p,\partial \qM)},
\ee
hence also 
\be
\gcur(p) \leq\frac{\eta}{\dist^{2}(p,\partial \qM)}.
\ee
\end{Proposition}
\vs

\begin{proof} Write (\ref{ES1}) as
\be
Ric_{f}^{\alpha}=\frac{1}{2}\Omega^{2}e^{-4V}q+2\nabla U\nabla U\geq  0
\ee
with $f=-V$, $\alpha=1$, and recall from (\ref{ES3}) that $\Delta_{f}U=0$. Then, using (\ref{BOCHNERF}) with $\psi=U$ we obtain
\be
\Delta_{f} |\nabla U|^{2}\geq 4|\nabla U|^{4}
\ee
and hence (\ref{DEC1}) by Lemma \ref{LEMMAME}. 

Similarly, if $\Omega=0$ we have $\Delta_{f}V=0$ and using (\ref{BOCHNERF}) again but with $\psi=V$ we obtain 
\be
\Delta_{f} |\nabla V|^{2}\geq 2|\nabla V|^{4}
\ee
and hence (\ref{DEC2}) by Lemma \ref{LEMMAME}. 
\end{proof}

The next proposition describes in simple form the asymptotic topology of data sets $(\qM;q,U,V)$ when $\Omega=0$. Observe however that we require $\partial \qM$ compact and of course $(\qM;q)$ metrically complete. 

\begin{Proposition}\label{SUPO} Let $(\qM;q,U,V)$ be a metrically complete (reduced) static data set with $\Omega=0$, $\qM$ non-compact and $\partial \qM$ compact. Then there is a set $K$ with compact closure, such that 
\be
\qM=K\cup\big(\cup_{i=1}^{i=n} E_{i}\big)
\ee
where every $E_{i}$ is diffeomorphic to $[0,\infty)\times \Sa$.
\end{Proposition}

\begin{proof} First we observe that as $\kappa\geq 0$, the ball covering property holds (indeed regardless of whether $\Omega=0$ or not). Hence, $\mathcal{S}$ has a finite number of ends. In particular we can write $\mathcal{S}$ as the union of a set with compact closure and a finite number of surfaces $E_{i}$, $i=1,\ldots,i_{S}$, each with compact boundary and containing only one end. 

It is sufficient to work with the surfaces $E_{i}$, that we denote generically as $E$. By Bishop-Gromov we have $\frac{A(B(\partial E,r))}{r^{2}}\searrow\ \mu$. The analysis depends on whether $\mu=0$ or $\mu>0$.

{\it Case $\mu=0$.} Let $\gamma$ be a ray from $\partial E$ and let $p_{i}\in \gamma$ with $r(p_{i})=r_{i}=2^{i}$, for $i=0,1,2,\ldots$. If $\mu=0$, then the sequence of annuli $(\mathcal{A}^{c}_{r_i}(p_{i},1/4,4);q_{r_i})$ collapses in volume (in area) with bounded curvature. As we have explained earlier, this type of collapse is only through thin (finite) cylinders. Thus, (outside a compact set) $E$ is formed by an infinite concatenation of finite cylinders, (i.e. each diffeomorphic to $[0,1]\times {\rm S}^{1}$). 

{\it Case $\mu>0$.} As $\gcur\geq 0$ and $\gcur$ has quadratic decay, if $\mu>0$ then $(E;q)$ is asymptotic to a flat cone $(\mathcal{C};q_{\mu})$ where
\be\label{FLATANN}
\mathcal{C}:=\mathbb{R}^{2}\setminus \{(0,0)\},\qquad q_{\mu}=dr^{2}+4\mu^{2}r^{2}d\varphi^{2}
\ee 
($r$ is the radius and $\varphi$ is the polar angle in $\mathbb{R}^{2}$). It then follows that, outside a compact set of compact closure, $E$ is diffeomorphic to $[0,\infty)\times \Sa$ as wished.
\end{proof}

\subsubsection{Case $\Omega\neq 0$}

The following lemma is the analogous to Lemma \ref{REDCUR} in the case $\Omega=0$. Note however that, contrary to the case $\Omega=0$, we assume that $\partial \qM$ is compact. We do not know if this condition can be removed or not.

\begin{Lemma}\label{SUPO2} Let $(\qM;q,U,V)$ be a metrically complete (reduced) static data set with $\Omega\neq 0$, $\qM$ non-compact and $\partial \qM$ compact. Then, 
\begin{align}
\label{ONZD3} |\nabla U|^{2}(p)\leq \frac{\eta}{\dist^{2}(p,\partial S)}, \qquad |\nabla V|^{2}(p)\leq \frac{\eta}{\dist^{2}(p,\partial S)},
\end{align}
and,
\be
\kappa(p)\leq \frac{\eta}{\dist^{2}(p,\partial S)}
\ee
where $\eta>0$ is independent on the data. In particular
\be
\Lambda^{2}(p)\geq \eta'\Omega\, \dist(p,\partial S)
\ee
where $\eta'>0$ is also independent on the data.
\end{Lemma}

\begin{proof} The proof requires using Corollary \ref{SINGMODELUNDER}. Without loss of generality assume that $S$ is an end. Let $\gamma$ be a ray from $\partial S$. For every $j\geq 0$ let $r_{j}=2^{2j}$ and let $p_{j}\in \gamma$ be such that $\dist(p_{j},\partial S)=r_{j}$. 

The first goal will be to prove that $\kappa$ and $|\nabla V|^{2}$ decay quadratically along the union of annuli $\cup_{j\geq 0} \mathcal{A}^{c}_{r_{j}}(p_{j};1/8,8)$. We will prove later that this union covers $\gamma$ except for a finite segment of it (a priori that may not be the case). 

Let $x_{j_{i}}$ be any sequence of points such that $x_{j_{i}}\in \mathcal{A}^{c}_{r_{j_{i}}}(p_{j_{i}};1/8,8)$ for every $i\geq 0$. Each $x_{j_{i}}$ can be joined to $p_{j_{i}}$ through a continuous curve $\alpha_{i}$ entirely inside the annulus $\mathcal{A}^{c}_{r_{j_{i}}}(p_{j_{i}};1/8,8)$. Concatenating $\alpha_{i}$ with the part of $\gamma$ extending from $p_{j_{i}}$ to infinity, we obtain a curve, say $\hat{\alpha}_{i}$, extending from $x_{j_{i}}$ to infinity, and never entering the ball $B(\partial S,1/8)$, namely, keeping at a $q_{r_{j}}$-distance of $1/8$ from $\partial S$. We will use the existence of this curve below to reach a contradiction.

Suppose now that either, 
\be\label{DIVERGENCEP}
\kappa(x_{j_{i}})\dist^{2}(x_{j_{i}},\partial S)\rightarrow \infty,\quad {\rm or}\quad |\nabla V|^{2}(x_{j_{i}})\dist^{2}(x_{j_{i}},\partial S)\rightarrow \infty
\ee
We perform a sequence of scalings $(\lambda_{i},\nu_{i},\mu_{i})=(r_{j_{i}},r_{j_{i}},0)$ leading to the new fields,
\be
q\rightarrow q_{i}=\frac{1}{r_{j_{i}}^{2}}q,\quad V\rightarrow V_{i}=V+\frac{1}{2}\ln r_{j_{i}},\quad U\rightarrow U_{i}=U
\ee
With this scaling we obtain then a sequence of reduced data $(S;q_{i},V_{i},U_{i})$ all having the same $\Omega$ (recall $\Omega\rightarrow \Omega_{i}=(\nu_{i}/\lambda_{i})\Omega=\Omega$). At the same time we have $1/8\leq \dist_{i}(x_{j_{i}},\partial S)\leq 8$. Because of this, we can rewrite (\ref{DIVERGENCEP}) as,
\be
\kappa_{i}(x_{j_{i}})\rightarrow \infty,\quad {\rm or}\quad |\nabla V_{i}|_{i}^{2}(x_{j_{i}})\rightarrow \infty,
\ee
(where $\kappa_{i}=\kappa_{q_{i}}$ and $|\nabla V_{i}|=|\nabla V_{i}|_{q_{i}}$). Taking a subsequence if necessary we can assume that $\dist_{i}(x_{j_{i}},\partial S)\rightarrow d_{*}$ (where $d_{i}=d_{q_{i}}$). 

We are clearly in the hypothesis of Corollary \ref{SINGMODELUNDER}. Choosing $d_{1}$ (see the hypothesis of Corollary \ref{SINGMODELUNDER}) as $d_{1}=d_{*}+(d_{*}-1/8)/2$, we conclude that there is a sequence of neighbourhoods $\mathcal{B}_{i}$ containing $B_{i}(x_{j_{i}},d_{1})$ such that $(\mathcal{B}_{i};q_{i})$ metrically collapses to a segment of length $d_{1}$ (where $B_{i}=B_{q_{i}}$). The neighbourhood $\mathcal{B}_{i}$ essentially wraps around a geodesic $\beta_{i}$ joining $x_{j_{i}}$ and $\partial S$ and minimising the distance between them, and `covering' the part of it at a distance less or equal than $d_{1}$ from $x_{j_{i}}$. Hence, for $i$ large enough, the boundary of the $\mathcal{B}_{i}$ is inside the ball $B_{i}(\partial S,1/8)$. Therefore for $i$ large enough, the curve $\hat{\alpha}_{i}$ must enter $B_{i}(\partial S,1/8)$ before going to infinity. We reach thus a contradiction.

We have then that for each $j$, the scaled curvature $\kappa_{r_{j}}$ is bounded on each of the annuli $\mathcal{A}^{c}_{r_{j}}(p_{j};1/8,8)$. Consider the areas $A_{r_{j}}$ of the annuli $\mathcal{A}^{c}_{r_{j}}(p_{j};1/8,8)$ with respect to $q_{r_{j}}$. If $A_{r_{j}}$ tend to zero then the annuli $(\mathcal{A}^{c}_{r_{j}}(p_{j};1/8,8),q_{r_{j}})$ collapse with bounded curvature and thus become thiner and thiner finite cylinders. The end $S$ is then (except for a set of compact closure) a concatenation of the annuli $\mathcal{A}^{c}_{r_{j}}(p_{j};1/8,8)$ and the quadratic curvature decay in the whole end follows as well as the quadratic decay of $|\nabla V|^{2}$ follows. If instead a sequence $A_{r_{j_{i}}}$ of the areas is bounded below away from zero then, due to the Bishop-Gromov monotonicity $A(B(\partial S,r))/r^{2}\searrow$ and the curvature bound, the geometry of the annuli $(\mathcal{A}^{c}_{r_{j}}(p_{j};1/8,8); q_{r_{j}})$ becomes more and more that of a flat annulus. Once a piece sufficiently close to a flat annulus forms then the whole end must be asymptotic to a flat annulus (for a detailed proof in dimension three see \cite{MR3233267}). Again, the quadratic decay of $\kappa$ and $|\nabla V|^{2}$ on the whole end follows.
\end{proof}

The following version of Proposition \ref{SUPO} but when $\Omega\neq 0$ is now straight forward after Proposition \ref{SUPO2} and the proof of Proposition \ref{SUPO} itself.

\begin{Proposition}\label{SUPO3} Let $(\qM;q,U,V)$ be a metrically complete (reduced) static data set with $\Omega\neq0$, $\qM$ non-compact and $\partial \qM$ compact. Then there is a set $K$ with compact closure, such that 
\be
\qM=K\cup\big(\cup_{i=1}^{i=n} E_{i}\big)
\ee
where every end $E_{i}$ is diffeomorphic to $[0,\infty)\times \Sa$.
\end{Proposition}

Taking into account the previous proposition we say that $(E;q,U,V)$ is a (reduced) static ends if $E\sim [0,\infty)\times \Sa$ and $(E;q)$ is metrically complete.  
\vs

From the description of the asymptotic geometry of (reduced) static ends $(E;q,U,V)$, $(E\sim [0,\infty)\times \Sa$), we can easily find a simple end cut $\{\ell_{j};j=1,2,\ldots\}$. Each $\ell_{j}$ is of course isotopic to $\partial E$ and embedded in $\mathcal{A}(2^{1+2j},2^{2+2j})$. Let us be a bit more precise. Let $r_{j}=2^{1+2j}$ and as usual let $q_{r_{j}}=q/r_{j}^{2}$. If $\mu=0$ then the annuli $(\mathcal{A}_{r_{j}}(1,2);q_{r_{j}})$ metrically collapse to the segment $[1,2]$ and therefore the loops $\ell_{j}$ can be chosen to have $q_{r_{j}}$-length tending to zero. If instead $\mu>0$ then the loops can be chosen to converge to the radial circle $\{x=3/2\}$ as the annuli $(\mathcal{A}_{r_{j}}(1,2);q_{r_{j}})$ converge to the annulus $([1,2]\times \Sa;dx^{2}+4\mu^{2}x^{2}d\varphi^{2})$ as explained earlier.   

Let $\Sigma$ be the three-manifold whose quotient by the $\Sa$-Killing field is $E$. Let $\pi:\Sigma\rightarrow E$ be the projection. The tori $S_{j}:=\pi^{-1}(\ell_{j})$ form obviously a simple cut of $(\Sigma;\hg)$. Let us state this in a proposition that will be recalled later.

\begin{Proposition}\label{SIMPLECUTUS} Let $(\Sigma;\hg,U)$ be a free $\Sa$-symmetric metrically complete static data set such that the reduced data set $(E;q,U,V)$ is a reduced end. Then, $\Sigma$ and $E$ admit simple cuts.
\end{Proposition}

The next proposition shows that $U$ tends uniformly to a constant $U_{\infty}$, on any (reduced) static end $(E;q,U,V)$. The constant $U_{\infty}$ satisfies $-\infty\leq U_{\infty}\leq \infty$. The proposition will be used in Section \ref{FTKASS}.

\begin{Proposition}\label{LPRO} Let $(E; \hqg,U,V)$ be a metrically complete reduced static end. Then, $U\rightarrow U_{\infty}$ where the arrow signifies uniform convergence and the constant $U_{\infty}$ satisfies $-\infty\leq U_{\infty}\leq \infty$.
\end{Proposition}

\begin{proof} Note that the maximum principle is also applicable to $U$ because (\ref{ES3}) can be written as $div(e^{V}\nabla U)=0$. We will use this several times below.

Let $\{\ell_{j},j=0,1,2,\ldots\}$ be a simple cut of $E$ as described above. Let $r_{j}=2^{1+2j}$. 

Assume that $\mu=0$. Then, as said, the $q_{r_{j}}$-length of the loops $\ell_{j}$ tends to zero. At the same time the norm $|\nabla U|_{r_{j}}$ restricted to the loops $\ell_{j}$ remains uniformly bounded. Therefore, by a simple integration along the $\ell_{j}$ it is deduced that,
\be\label{ONCEMORE}
(\max\{U(q):q\in \ell_{j}\}-\min\{U(q):q\in \ell_{j}\})\rightarrow 0
\ee  
If instead $\mu>0$ then the $q_{r_{j}}$-length of the loops $\ell_{j}$ remains uniformly bounded while the norm $|\nabla U|_{r_{j}}$, over the loops $\ell_{j}$, tends to zero. So by a simple integration along the loops $\ell_{j}$ we deduce again (\ref{ONCEMORE}). 

Now suppose that for a certain sequence $p_{i}\in \ell_{j_{i}}$, $U(p_{i})$ tends to a constant $-\infty\leq U_{\infty}\leq \infty$. Then by (\ref{ONCEMORE}), the maximum and the minimum of $U$ over $\ell_{j_{i}}$ also tend to $U_{\infty}$. We use now the maximum principle to write for any $i<i'$
\begin{align}
\max\{U(q):q\in \ell_{j_{i}}\cup \ell_{j_{i'}}\}\geq & \max\{U(q):q\in \mathcal{L}_{j_{i},j_{i'}}\}\geq \\
\geq & \min\{U(q):q\in \mathcal{L}_{j_{i},j_{i'}}\}\geq \\
\geq & \min\{U(q):q\in \ell_{j_{i}}\cup \ell_{j_{i'}}\}
\end{align}
where $\mathcal{L}_{j_{i},j'_{i}}$ is the compact region enclosed by $\ell_{j_{i}}$ and $\ell_{j_{i'}}$. Letting $i'$ tend to infinity we deduce,
\begin{align}
\label{COSAP1}\max\{\max\{U(q):q\in & \ell_{j_{i}}\},U_{\infty}\}\geq \max\{U(q):q\in \mathcal{L}_{j_{i},\infty}\}\geq \\
\label{COSAP2}\geq & \min\{U(q):q\in \mathcal{L}_{j_{i},\infty}\}\geq \min\{\min\{U(q):q\in \ell_{j_{i}}\},U_{\infty}\}
\end{align}
where $\mathcal{L}_{j_{i},\infty}$ is the region enclosed by $\ell_{j_{i}}$ and infinity. As the left hand side of (\ref{COSAP1}) and the right hand side of (\ref{COSAP2}) tend to $U_{\infty}$ then $U$ must tend also uniformly to $U_{\infty}$.
\end{proof}

\subsection{Reduced data sets arising as collapsed limits}\label{RDSACL}

In this last subsection about $\Sa$-symmetric data sets, it is worth to discuss the geometry of reduced data arising from scaled limit of data sets. This discussion will be recalled later in Section \ref{POKA} where we prove that the asymptotic of static black hole data sets with sub-cubic volume growth is Kasner.

Let $(\Sigma;\hg,U)$ be a data set, and let $\gamma$ be a ray from $\partial \Sigma$. Let $p_{n}\in \gamma$ be a divergent sequence of points. Suppose there are neighbourhoods $\mathcal{B}_{n}$ of $\mathcal{A}^{c}_{r_{n}}(p_{n}, 1/2,2)$ such that $(\mathcal{B}_{n};\hg_{r_{n}})$ collapses to a two-dimensional orbifold. Having this, by a diagonal argument, one can find a subsequence of it (also indexed by $n$) and neighbourhoods $\mathcal{B}_{n}$ of $\mathcal{A}^{c}_{r_{n}}(p_{n};1/2,2^{k_{n}})$, with $k_{n}\rightarrow \infty$, and collapsing to a two-dimensional orbifold $(S_{\infty};q_{\infty})$. As the collapse is along $\Sa$-fibers (hence defining asymptotically a symmetry), we obtain, in the limit, a well defined reduced data $(S;q,\bar{U},V)$ where $\overline{U}$ is obtained as the limit of $U_{n}:=U-U(p_{n})$. On smooth points the scalar curvature $\kappa$ is non-negative. Orbifold points are conical with total angles an integer fraction of $2\pi$ ($2\pi/2$, $2\pi/3$, $2\pi/4$, etc) hence can be thought as having also non-negative curvature (they can be rounded off to have a smooth metric with $\kappa\geq 0$). Therefore $(S;q)$ has only a finite number of ends. Note that it has at least one end containing a limit, say $\overline{\gamma}$, of the ray $\gamma$. Let us denote that end by $S_{\overline{\gamma}}$. 

We claim that every end has only a finite number of orbifold points. This is the result of a simple application of Gauss-Bonnet. Indeed, let $S$ be an end. Let $\ell_{j}$, $j=1,2,3,\ldots$, be one-manfiolds embedded for each $j$ in $\mathcal{A}(2^{2j},2^{2j+3})$ such that $\ell_{1}$ and $\ell_{j}$ enclose a connected manifold $\Omega_{1j}$. Let $\mathcal{O}$ be the set of orbifold points in $S$. By Gauss-Bonnet we have
\be
-\int_{\ell_{1}}kdl-\int_{\ell_{j}}kdl=\int_{\Omega_{1j}\setminus \mathcal{O}}\kappa dA+\sum_{p\in \Omega_{1j}\cap \mathcal{O}}2\pi\bigg(\frac{i(p)-1}{i(p)}\bigg)
\ee
where $k$ is the mean-curvature (or first variation of logarithm of length) on the one-manifolds $\ell_{j}$ and the angle at each orbifold point $p\in \mathcal{O}$ is $2\pi/i(p)$. As the right hand side is greater or equal than the number of orbifold points in $\Omega_{1j}$, that is $\sharp\{\Omega_{1j}\cap \mathcal{O}\}$. Thus, if the left hand side remains bounded as $j\rightarrow \infty$ then the number of orbifold points must be finite. To see the existence of such one-manifolds $\ell_{j}$ for which the left hand side remains bounded just argue as follows. First note that the left hand side is scale invariant. Second observe that as for each $j$ the scaled annuli $(\mathcal{A}(2^{2j},2^{2j+3});q_{2^{2j}})$ in $S$ are scaled limits of annuli in $(\Sigma;\hg)$, (which has quadratic curvature decay), then one can always chose a suitable subsequence $j_{i}$ such that as $i\rightarrow \infty$ the annuli $(\mathcal{A}(2^{2j_{i}},2^{2j_{i}+3});q_{2^{2j_{i}}})$ either converge of collapse to a segment. The selection of the $\ell_{i}$ is then evident.   
 
\section{Volume growth and the asymptotic of ends}\label{VWAE}

The asymptotic of ends is markedly divided by the volume growth. We discuss first cubic volume growth, which is the simplest and that implies AF. Then we discuss sub-cubic volume growth which implies (under certain hypothesis) AK. This last case requires an elaborated and long analysis.

\subsection{Cubic volume growth and asymptotic flatness}\label{ENDSAF}

Suppose $(\Sigma;\hg,U)$ is a static end with cubic volume growth. Cubic volume growth, non-negative Ricci curvature and quadratic curvature decay, implies that the end is asymptotically conical, (i.e. the metric is asymptotic to a metric of the form $dr^{2}+a^{2}r^{2}d\Omega^{2}$ in $\mathbb{R}^{3}$). Hence, outside an open set of compact closure, $\Sigma$ is diffeomorphic to $\mathbb{R}^{3}$ minus a ball. It was proved in \cite{MR3233266},\cite{MR3233267} that the data is then asymptotically flat (indeed asymptotically Schwarzschild). 

\subsection{Sub-cubic volume growth and Kasner asymptotic}\label{ENDSAK}

The goal of this section will be to prove that the asymptotic of any static black hole data set with sub-cubic volume growth is Kasner different from a Kasner $A$ or $C$. Observe that the claim is for the asymptotic of black hole data sets, and not just that of any end with sub-cubic volume growth.

We aim therefore to prove the following theorem.
\begin{Theorem}\label{KAFR} Let $(\Sigma;\hg,U)$ be a static black hole data set with sub-cubic volume growth. Then the data is asymptotically Kasner, different from a Kasner $A$ or $C$.
\end{Theorem}
To achieve this we provide first a necessary and sufficient condition for Kasner asymptotic different from $A$ or $C$. This is the content of Proposition \ref{KASYMPTOTIC} for which we dedicate the whole subsection \ref{SPTWOP1}. In second place, we analise the asymptotic of free $\Sa$-symmetric static ends $(\Sigma;\hg,U)$ under the natural condition that $U(p)\leq U_{\infty}$ (recall that $U_{\infty}$, the limit of $U$ at $\infty$, exists by Proposition \ref{LPRO}). We dedicate subection \ref{FTKASS} to show Theorem \ref{SSKAA} claiming that, for such a data, either the asymptotic is Kasner different from $A$ or $C$, or the whole data is flat and $U$ is constant. The proof requires the results we have obtained for reduced data sets in section \ref{S1S}, as well as the development of an interesting monotonic quantity along the leaves of the level sets of $U$, that in turn will be used again in the proof of Theorem \ref{KAFR}. Finally, subsection \ref{POKA} uses the results of the previous two subsections to prove the desired Theorem \ref{KAFR}. 

\subsubsection{Preliminaries, $C^{k}$-norms on two-tori} \label{SPTWOP0}

We prove here a series of results on the $C^{k}$-norm of tensor field on two-tori that will be used in section \ref{SPTWOP1}.

We begin recalling the definition of the $C^{k}$-norms of a tensor with respect to a background metric. Let $(M;g)$ be a smooth Riemannian manifold. Let $W$ be a smooth tensor of any valence. We denote by $|W|_{g}(x)$ the $g$-norm of $W$ at $x\in M$. Given $k\geq 0$, the $C^{k}$-norm of $W$ with respect to $g$ is defined as
\be
\|W\|^{2}_{C^{k}_{g}}:=\sup_{x\in M}\bigg\{\sum_{i=0}^{i=k}|\nabla^{(i)}W|^{2}_{g}(x)\bigg\} \quad \text{where}\quad \nabla^{(i)}W=\underbrace{\nabla\ldots\nabla}_{\text{i-times}} W
\ee

\begin{Proposition}\label{ENDK1} Let $(T;h_{F})$ be a flat  two-torus. Let $W$ be a smooth tensor field (of any valence), equal to zero at some point. Then for any $0\leq j\leq k$ we have
\be\label{TOP}
\|W\|_{C^{j}_{h_F}}\leq c(k_{ij})\ \diam_{h_F}^{k-j}(T)\ \|W\|_{C^{k}_{h_F}}
\ee
\end{Proposition}     

\begin{proof} We will prove the inequality for functions. To prove it for tensors use the expansion $W=\sum f_{I}\omega_{I}$, where $\omega_{I}$ is an orthonormal and parallel basis (i.e. $\delta_{II'}=<\omega_{I},\omega_{I'}>_{g}$ and $\nabla \omega_{I}=0$), and then use the result obtained for functions.   

We will work in $(\mathbb{R}^{2}; g_{\mathbb{R}^{2}})$ thought as the universal cover of $(T;h_F)$. In particular $\pi^{*}h_{F}=g_{\mathbb{R}^{2}}$ where $\pi:\mathbb{R}^{2}\rightarrow T$ is the projection. On a Cartesian coordinate system $(x_1,x_2)$ 
we have
\be\label{EUCMETR}
g_{\mathbb{R}^{2}}=dx_{1}^{2}+dx_{2}^{2}
\ee
and 
\be\label{NOSI}
\|f\|^{2}_{C^{j}_{h_{F}}}=\|f\|^{2}_{C^{j}_{g_{\mathbb{R}^{2}}}}=\sup_{x\in \mathbb{R}^{2}}\bigg\{\sum_{|I|=0}^{|I|=j} |\partial_{I} f|^{2}(x)\bigg\}
\ee
where for any multi-index $I=(i_{1},\ldots,i_{|I|})$, $i_{l}\in \{1,2\}$, we denote $\partial_{I}=\partial_{x_{i_1}}\ldots\partial_{x_{i_{|I|}}}$.

We will need to rely on the existence of a coordinate system $(\overline{x}_{1},\overline{x}_{2})$ on $\mathbb{R}^{2}$ on which the metric $g_{\mathbb{R}^{2}}$ is written as
\be\label{METRIC}
g_{\mathbb{R}^{2}}=d\overline{x}_{1}^{2}+\alpha (d\overline{x}_{1}d\overline{x}_{2}+d\overline{x}_{2}d\overline{x}_{1})+d\overline{x}_{2}^{2},
\ee  
where $\alpha$ is a constant such that $|\alpha|\leq 1/2$, and where the directions $\partial_{\overline{x}_{1}}$ and $\partial_{\overline{x}_{2}}$ are periodic of period less than $6\diam_{h_F}(T)$, that is, any line in the direction of either $\partial_{x_{1}}$ or $\partial_{x_{2}}$ projects into a circle in $T$ of length less that $6\diam_{h_F}(T)$. For the calculations that follow we assume that the coordinates $(\overline{x}_{1},\overline{x}_{2})$ are given. We will prove their existence at the end. 

Observe that the norm (\ref{NOSI}), which is defined with respect to the metric (\ref{EUCMETR}) and the norm 
\be\label{NOSII}
\|f\|^{2}_{C^{j}_{\overline{g}_{\mathbb{R}^{2}}}}=\sup\bigg\{\sum_{|I|=0}^{|I|=j} |\overline{\partial}_{I} f|^{2}\bigg\},\qquad \overline{\partial}_{I}=\partial_{\overline{x}_{i_{1}}}\ldots\partial_{\overline{x}_{i_{|I|}}},
\ee
which is defined with respect to the metric 
\be
\overline{g}_{\mathbb{R}^{2}}=d\overline{x}^{2}_{1}+d\overline{x}^{2}_{2},
\ee
are equivalent, namely $c_{1}(j) \|f\|_{C^{j}_{g_{\mathbb{R}^{2}}}} \leq \|f\|_{C^{j}_{\overline{g}_{\mathbb{R}^{2}}}}\leq c_{2}(j) \|f\|_{C^{j}_{g_{\mathbb{R}^{2}}}}$. This is proved by noting that the family of metrics (\ref{METRIC}) with $|\alpha|\leq 1/2$ is compact. Thus, to prove (\ref{TOP}) it is enough to prove 
\be\label{TOPP}
\|W\|_{C^{j}_{\overline{g}_{\mathbb{R}^{2}}}}\leq c(k)\ \diam_{h_F}^{k-j}(T)\ \|W\|_{C^{k}_{\overline{g}_{\mathbb{R}^{2}}}}
\ee
Again we prove this for functions. We show it in what follows. 

We claim first that for any function $\psi$ which is zero at some point, say $(\overline{x}_{1}^{0},\overline{x}_{2}^{0})$, we have
\be\label{ITER}
\sup\big\{|\psi|\big\}\leq 12\diam_{h_F}(T)\sup\big\{|\partial_{\overline{x}_{1}} \psi|, |\partial_{\overline{x}_{2}} \psi|\big\}
\ee
This is seen by just writing
\be
\psi(\overline{x}_{1},\overline{x}_{2})=\int_{0}^{\overline{x}_{1}-\overline{x}_{1}^{0}}\partial_{\overline{x}_{1}} \psi\bigg|_{(\overline{x}_{1}^{0}+s,\overline{x}_{2}^{0})}ds  + \int_{0}^{\overline{x}_{2}-\overline{x}_{2}^{0}} \partial_{\overline{x}_{2}} \psi\bigg|_{(\overline{x}_{1},\overline{x}_{2}^{0}+s)}ds
\ee
and using that $|\overline{x}_{1}-\overline{x}^{0}_{1}|$ and $|\overline{x}_{2}-\overline{x}_{2}^{0}|$ are less or equal than $6\diam_{h_F}(T)$. 

We use (\ref{ITER}) to prove (\ref{TOPP}). We observe first that, for any $\psi$ and multi-index $I$, ($|I|\geq 1$), the function $\overline{\partial}_{I}\psi$ has also a zero. To see this just fix $\overline{x}_{i}$, for all $i\neq i_{1}$ (at any values), and observe that the function $\psi$ as a function of $\overline{x}_{i_{1}}$ is the $\overline{x}_{i_{1}}$-derivative of a periodic function. Having the observation at hand, start with $\psi=f$ and use (\ref{ITER}) repeatedly to obtain (\ref{TOPP}).

It remains to show the existence of the coordinates $(\overline{x}_{1},\overline{x}_{2})$. In the cartesian system $(x_{1},x_{2})$, the balls $B((4\diam_{\mathbb{R}^{2}},0),\diam_{\mathbb{R}^{2}}(T))$ and $B((0,4\diam_{\mathbb{R}^{2}}),\diam_{\mathbb{R}^{2}}(T))$, possess points $q_{1}$ and $q_{2}$ projecting (in $T$) to the same point as the point $q_{0}=(0,0)$ does. Define the directions $\partial_{\overline{x}_{1}}$ and $\partial_{\overline{x}_{1}}$ as, respectively, those defined by $q_{0}, q_{1}$ and $q_{0}, q_{2}$, and finally define the origin of the coordinates $(\overline{x}_{1},\overline{x}_{2})$ to be $(x_{1},x_{2})=(0,0)$. It is direct to check that the coordinates $(\overline{x}_{1},\overline{x}_{2})$ thus constructed enjoy the required properties.  
\end{proof}

\begin{Proposition}\label{FLATTGA} Let $(T;h)$ be a Riemannian two-torus and let $p\in T$. Then there is a unique flat metric $h_{F}$, conformally related to $h$ and equal to $h$ at $p$. Moreover, for any integer $k\geq 1$, and reals $K_{1}>0$ and $K_{k}>0$ there is $D(K_{1})>0$ (small enough) and $C(k,K_{k})>0$ such that if
\be
\|\gcur\|_{C^{1}_{h}}\leq K_{1},\quad \|\gcur\|_{C^{k}_{h}}\leq K_{k},\quad \text{and}\quad \diam_{h}(T)\leq D
\ee
where $\gcur$ is the Gaussian curvature, then,
\be\label{FLATTGAEQQ}
e^{-C}h_{F}\leq h\leq e^{C}h_{F}
\ee
and
\be\label{FLATTGAEQ}
\|h\|_{C^{k}_{h_F}}\leq C.
\ee
\end{Proposition}

\begin{proof} We will use that there is $D(K_{1})$, (small enough), such that if $\diam_{h}(T)\leq D(K_{1})$ then there is a finite cover $\pi:(\tilde{T};\tilde{h})\rightarrow (T;h)$, (i.e. $\pi:\tilde{T}\rightarrow T$ and $\tilde{h}=\pi^{*}h$), such that, (i) $\diam_{\tilde{h}}(\tilde{T})\leq 1$, and (ii) $inj_{\tilde{h}}(p)\geq i_{0}(K_{1})$ for all $p\in \tilde{T}$. Because $(\tilde{T};\tilde{h})$ is a cover of $(T;h)$ we also have (iii) $\|\tilde{\gcur}\|_{C^{k}_{\tilde{h}}}\leq K_{k}$. The claims, (i) and (ii), are well known from the standard theory of diameter-collapse with bounded curvature. In simple terms they follow easily from the fact that the exponential map $exp:\mathcal{T}_{p}T\rightarrow T$ restricted to a small ball in $\mathcal{T}_{p}T$ is an immersion and then finding an appropriate fundamental domain on $\mathcal{T}_{p}T$ around $p$ that will define $\tilde{T}$. We will not discuss this further, rather we will use it from now on.  

The properties (i) and (ii) imply that the geometry of $(\tilde{T};\tilde{h})$ is controlled\footnote{To be precise: A geometry is controlled in $C^{k}$ by $K$ if there is a cover of $\tilde{T}$ by $n(K)$-harmonic charts, with Lebesgue number $\delta(K)$, such that, on each chart $(x_{1},x_{2})$, we have (i) $e^{K'(K)}\delta_{ij}\leq \tilde{h}_{ij}\leq e^{K'(K)}\delta_{ij}$ and (ii) $\|\tilde{h}_{ij}\|_{C^{k}_{\delta_{ij}}}\leq K'(K)$. See \cite{MR2243772}.} in $C^{2}$ by $K_{1}$. Moreover if the geometry of $(\tilde{T};\tilde{h})$ is controlled in $C^{2}$ by $K_{1}$ and (iii) above holds, then the geometry of $(\tilde{T};\tilde{h})$ is controlled in $C^{k+1}$ by $K_{k}$. This allows us to make standard elliptic analysis in $(\tilde{T};\tilde{h})$ as if working in a fixed manifold.

Let $\tilde{\phi}$ be the solution to
\be\label{LAPPHI}
\Delta_{\tilde{h}}\tilde{\phi}=\tilde{\gcur},\quad \text{with}\quad \int_{\tilde{T}}\tilde{\phi}\, dA_{\tilde{h}}=0
\ee 
With such $\tilde{\phi}$, the conformal metric $\tilde{h}_{F}=e^{2\tilde{\phi}}\tilde{h}$ is flat. Multiply (\ref{LAPPHI}) by $\tilde{\phi}$, integrate and use Cauchy-Schwarz to obtain
\be\label{RRRR}
\int_{\tilde{T}}|\nabla \tilde{\phi}|^{2}_{\tilde{h}}\, dA_{\tilde{h}}\leq \big(\int_{\tilde{T}}\tilde{\gcur}^{2}\, dA_{\tilde{h}}\big)^{\frac{1}{2}}\big(\int_{\tilde{T}}\tilde{\phi}^{2}\, dA_{\tilde{h}}\big)^{\frac{1}{2}}
\ee  
Now, we can use the Poincar\'e inequality
\be\label{POINCA}
\int_{\tilde{T}}\tilde{\phi}^{2}\, dA_{\tilde{h}}\leq I(K_{1})\int_{\tilde{T}}|\nabla \tilde{\phi}|^{2}\, dA_{\tilde{h}}
\ee
in the right hand side of (\ref{RRRR}) to obtain an upper bound on $\|\nabla \tilde{\phi}\|_{L^{2}_{\tilde{h}}}$, (that $I=I(K_{1})$ is justified because the geometry of $\tilde{T}$ is controlled in $C^{2}$). Such bound can be used in turn again in (\ref{POINCA}) to obtain $\|\tilde{\phi}\|_{L^{2}_{\tilde{h}}}\leq B_{1}(K_{1})$. Using this $L^{2}$-bound together with standard elliptic estimates on (\ref{LAPPHI}) we obtain 
\be\label{BNEC}
\|\tilde{\phi}\|_{C^{k}_{\tilde{h}}}\leq B_{2}(k,K_{k}).
\ee 
As $k\geq 1$, we deduce 
\be\label{BOUNNO}
|\tilde{\phi}|\leq B_{2}(k,K_{k}), 
\ee
This implies that for a $C_{1}(k,K_{k})>0$ we have
\be\label{BNECI}
e^{-C_{1}}\tilde{h}\leq \tilde{h}_{F}\leq e^{C_{1}}\tilde{h}
\ee
Moreover the covariant derivative $\partial$ of $\tilde{h}_{F}$ is related to the covariant derivative $\nabla$ of $\tilde{h}$ by
\be\label{BNECII}
\partial_{A}=\nabla_{A}+(\nabla_{A}\phi) h_{B}^{C}+(\nabla_{B}\phi) h_{A}^{C}-(\nabla^{C}\phi) h_{AB}
\ee
Now (\ref{BNEC}), (\ref{BNECI}) and (\ref{BNECII}) (to compute $\partial^{(j)}$) imply the bound
\be\label{BOUNN}
\|\tilde{\phi}\|_{C^{k}_{\tilde{h}_{F}}}\leq B_{3}(k,K_{k}).
\ee
By the uniqueness of solutions to (\ref{LAPPHI}), $\tilde{\phi}$ has to coincide with its average by the Deck-transformations. Hence, $\tilde{\phi}$ and  $\tilde{h}_{F}$ descend respectively to a function $\overline{\phi}$ and a flat metric $h_{F}$. As the bound (\ref{BOUNN}) is local we also have
\be\label{BOUNN22}
\|\overline{\phi}\|_{C^{k}_{h_{F}}}\leq B_{3}(k,K_{k}).
\ee
Finally define $\phi=\overline{\phi}-\overline{\phi}(p)$. With this definition we have $h(p)=h_{F}(p)$. From the bound $|\overline{\phi}|\leq B_{2}(k,K_{k})$ we obtain the bound $|\phi|\leq 2B_{2}(k,K_{k})$ hence (\ref{FLATTGAEQQ}). Also from $|\phi|\leq 2B(k,K_{k})$ and (\ref{BOUNN22}) we deduce,
\be\label{BOUNN3}
\|\phi\|_{C^{k}_{h_{F}}}\leq B_{4}(k,K_{k}).
\ee 
hence (\ref{FLATTGAEQ}).
\end{proof}

\subsubsection{Necessary and sufficient condition for KA different from $A$ or $C$.}\label{SPTWOP1}

Before passing to the next crucial propositions we make a pair of geometric observations about the Kasner solutions and introduce some terminology. 
 
On $\mathbb{R}^{+}\times \mathbb{R}^{2}$ consider a Kasner solution
\begin{align}
& \hg=dx^{2}+x^{2a}dy^{2}+x^{2b}dz^{2},\\
& U=c\ln x
\end{align}
and assume that $c\in (0,1/2]$. Quotient the space by a $\mathbb{Z}\times \mathbb{Z}$ action to obtain a Kasner solution on $\mathbb{R}^{+}\times \T^{2}$. For every $x$ let $T_{x}$ be the two-torus $\{x\}\times \T^{2}$. Fixed $c$, there are two possibilities for $(a,b)$, $(a_{-},b_{-})$ and $(a_{+},b_{+})=(b_{-},a_{-})$. In either case, and because $c\in (0,1/2]$, we have $0<a<1$, $0<b<1$. Let 
\be\label{ASDEF}
a_{*}=\max\big\{e^{2/(1-a)},e^{2/(1-b)}\big\}
\ee
Observe that, as $a+b=1$ we have $a^{*}\geq 4$. Note that if $\overline{a}\geq a_{*}$ then
\be\label{DIAMHAL}
\diam_{\hg_{\overline{a}}}(T_{\overline{a}})\leq \frac{1}{e^{2}}\diam_{\hg}(T_{1})
\ee
where, recall former notation, $\hg_{\overline{a}}=\hg/\overline{a}^{2}$. To see this simply note that
\be\label{ASTRA}
\frac{1}{\overline{a}^{2}}(\overline{a}^{2a}dy^{2}+\overline{a}^{2b}dz^{2})=(\overline{a}^{2a-2}dy^{2}+\overline{a}^{2b-2}dz^{2})\leq \frac{1}{e^{4}}(dz^{2}+dy^{2})
\ee 
so (\ref{DIAMHAL}) holds no matter how we quotient $\mathbb{R}^{2}$. Thus, the diameter of $T_{\overline{a}}$ with respect to $\hg_{\overline{a}}$, is at least $1/e^{2}$ of the diameter of $T_{1}$ with respect to $\hg$. 
\vs

In the following propositions we will use the notation $\rho=|\nabla U|$, $\rho_{r}=|\nabla U|_{r}$ and $\lambda=1/\rho$, $\lambda_{r}=1/\rho_{r}$. Also, given $0<\rho^{*}\leq 1/2$ we let 
\be
a^{*}=\max\{a_{*}(c): c \in [\rho^{*}/4, 1/2]\}. 
\ee
The reader must keep this notation in mind.

We will also use the following definition. Let $W$ be a tensor of any valence defined at just one point $x$ of a flat torus $(T;h_{F})$. Then the {\it $h_{F}$-extension} of $W$ is the tensor field defined by translating $W$ to all $T$ by its isometry group. 

\begin{Proposition}\label{PORSI} Let $(\Sigma; \hg, U)$ be a static end, and let $\gamma$ be a ray emanating from $\partial \Sigma$. Let $0<\rho^{*}\leq 1/2$ and let integers $j^{*}\geq 0$ and $m^{*}\geq 1$. Then, there exist positive constants $\epsilon^{*}, \mu^{*}\leq \rho^{*}/2, r^*, c^*$, such that if at a point $p\in \gamma$ with $r=r(p)\geq r^*$ we have,
\begin{enumerate}[labelindent=\parindent, leftmargin=*, label={\rm (\alph*)}, widest=a, align=left]
\item\label{PORSI-a} $\dist_{GH}\big(\big(\mathcal{A}^{c}_{r}(p;1/2,2); d_{r}\big),\big([1/2,2];|\ldots|\big)\big)\leq \epsilon^*$, and,
\item\label{PORSI-b} $|\rho_{r}(p)-\rho^*|\leq \mu^{*}$,
\end{enumerate}
then,
\begin{enumerate}[labelindent=\parindent, leftmargin=*, label={\rm (\Roman*)}, widest=II, align=left]
\item\label{PORSI-I} there is a neighbourhood $\mathcal{U}_{p}$ of $\mathcal{A}_{r}^{c}(p;1/(2a^{*}),2a^{*})$ foliated by level sets of $U$ each of which is a two-torus, and, 
\item\label{PORSI-II} there is a Kasner space $(\mathcal{U}^{\mathbb{K}};\hg^{\mathbb{K}},U^{\mathbb{K}})$, $\mathbb{Z}^{2}$-quotient of the Kasner,
\be
\tilde{\mathcal{U}}^{\mathbb{K}}=I\times \mathbb{R}^{2},\quad \tilde{\hg}^{\mathbb{K}}=dx^{2}+x^{2a}dy^{2}+x^{2b}dz^{2},\quad \tilde{U}^{\mathbb{K}}=d+ c\ln x 
\ee
($I$ is some interval) and a smooth diffeomorphism $\phi:\mathcal{U}_{p}\rightarrow \mathcal{U}^{\mathbb{K}}=I\times \T^{2}$ such that
\begin{align}\label{ESTIMATE}
& \phi_{*}U=U^{\mathbb{K}},\\
& \|\phi_{*} \hg_r - \hg^{\mathbb{K}}\|_{c^{j^{*}}_{\hg^{\mathbb{K}}}} \leq C^{*} \diam_{\hg^{\mathbb{K}}}^{m^*}\big(\phi(T_{p})\big)
\end{align}
where $T_p$ is the level set of $U$ containing $p$.
\end{enumerate}
\end{Proposition}

\begin{proof}\

\ref{PORSI-I} Proceeding by contradiction, assume that for every $\epsilon^{*}_{i}=1/i$, $\mu^{*}_{i}=1/i$ and $r^{*}=i$, with $i\geq i_{0}$, there is $p_{i}\in\gamma$ with $r_{i}=r(p_{i})\geq r^{*}_{i}$ for which \ref{PORSI-a} and \ref{PORSI-b} hold but for which the neighbourhood ${\mathcal{U}}_{p}$ with the desired properties does not exist. But if \ref{PORSI-a} holds and $p_{i}$ belongs to a ray then the space $(\mathcal{A}_{r_{i}}^{c}(p_{i};1/(2a^{*}),2a^{*});d_{r_{i}})$ necessarily metrically collapses to a segment of length $2a^{*}-1/(2a^{*})$. Thus there are neighbourhoods $\mathcal{B}_{i}$ of $\mathcal{A}_{r}^{c}(p_{i};1/(2a^{*}),2a^{*})$ and covers $\pi_{i}:\tilde{\mathcal{B}}_{i}\rightarrow \mathcal{B}_{i}$ such that $(\tilde{\mathcal{B}}_{i};\tilde{g}_{r_{i}},\tilde{U}_{i})$ converges to a $\Sa\times \Sa$-symmetric data set. The limit data set has non-constant $\rho$ because by \ref{PORSI-b} it must be $\tilde{\rho}_{r_{i}}(p_{i})\rightarrow \rho^{*}$ and $0<\rho^{*}\leq 1/2$. Hence the limit space is a Kasner space different from $A$ and $C$. Therefore for $i$ large enough the level sets of $\tilde{U}$ foliate $\tilde{\mathcal{B}}_{i}$ and hence $\mathcal{B}_{i}$. Thus the neighbourhoods $\mathcal{U}_{p_{i}}$ with the desired properties exist for $i$ large enough, which is a contradiction. 
\vs 

It is direct to see from the argumentation above that, after choosing $\epsilon^{*}$ smaller if necessary and $r^{*}$ bigger if necessary, $\rho_{r}$ is uniformly bounded above and below away from zero; that is, for some $0<\underline{\rho}<\overline{\rho}$, the bound $0<\underline{\rho}\leq \rho_{r}\leq \overline{\rho}$ holds on $\mathcal{U}_{p}$, for any $p\in \gamma$ with $r(p)\geq r^{*}$ for which \ref{PORSI-a} and (b) hold. In the proof of part \ref{PORSI-II} we will assume that $\epsilon^{*}$ and $r^{*}$ were chosen accordingly. As the proof progresses the values of $\epsilon^{*}$ and $r^{*}$ will be adjusted a few times.

Note that the estimates (\ref{ESTCHEC}) of Part I and the uniform bound for $\rho_{r}$ show that for any $i\geq 0$, $|\nabla^{(i)}\rho_{r}|_{r}$ is uniformly bounded (without the need to adjust $\epsilon^{*}$ or $r^{*}$ for each $i$). Similarly for any $i\geq 0$, $|\nabla^{(i)}\lambda_{r}|_{r}$ is uniformly bounded. 

{\it Terminology}: It is natural then to introduce the following terminology that will be used throughout the proof of \ref{PORSI-II} below. Let $\mathcal{G}$ be a geometric quantity defined on each of the neighbourhoods $\mathcal{U}_{p}$ (for instance $\mathcal{G}=\lambda_{r}$). Then $\mathcal{G}$ is {\it uniformly bounded} if one can find a constant $C>0$ such that $\mathcal{G}\leq C$ holds on $\mathcal{U}_{p}$, for any $p$ with $r(p)\geq r^{*}$ for which \ref{PORSI-a} and (b) hold.
\vs

\ref{PORSI-II} The construction of the Kasner space and the map $\phi$ is done in the three progressive steps \ref{PORSI-II}-A, \ref{PORSI-II}-B and \ref{PORSI-II}-C below. In (II)-A we define a map $\phi$ from $\mathcal{U}$ into a product space $I\times \T^{2}$. Then, also in \ref{PORSI-II}-A, we define a product metric $\hg_{F}$ on $I\times \T^{2}$ that will be used as a support metric, and prove its main properties. In \ref{PORSI-II}-B, we use $\hg_{F}$ to define a good $\Sa\times \Sa$-symmetric approximation $\breve{\hg}$ to $\phi_{*}\hg$. Finally in \ref{PORSI-II}-C we show that $(\breve{\lambda},\breve{\hg})$ 'almost' satisfy the ODEs defining Kasner metrics that were discussed in subsection \ref{UNIQ} and make the error explicit. We show that the Kasner solution defined out of such ODEs with an initial data equal to that of $(\breve{\lambda},\breve{\hg})$ at an initial slice, gives the wished Kasner metric $\hg^{\mathbb{K}}$ and $U^{\mathbb{K}}$. 
\vs

{\it Notation}: Throughout this part \ref{PORSI-II} we will be working on the neighbourhoods $\mathcal{U}_{p}$ and at the scaled geometry, namely dealing with $\hg_{r}$ rather than $\hg$. However to prevent a cumbersome notation we will omit the subindex $r$ everywhere. The reader should be aware of that.

\ref{PORSI-II}-A. {\sc The trivialisation $\phi$ and the flat metric $\hg_{F}$.} Given $q\in \mathcal{U}_{p}$ let $\zeta_{q}(U)$ be the integral curve of the vector field $\nabla^{a} U$, extending throughout $\mathcal{U}_{p}$ and parametrised by $U$. Then define $\phi:\mathcal{U}_{p}\rightarrow T_{p}\times I$ by $\phi(q)=(T_{p}\cap \zeta_{q},U(q))$. We will be identifying $\mathcal{U}_{p}$ with $T_{p}\times I$ via the diffeomorphism $\phi$.

On $T_{p}\times I$ the metric $\hg$ is written as
\be
\hg=\lambda^{2} dU^{2}+h
\ee
where $\lambda=1/\rho$ and $h$ is the induced metric on the tori $T_{U}:=T_{p}\times U$. Denote by $D$ the intrinsic covariant derivative on the $T_U$'s. As $T_{U(p)}$ will appear often we will use the simpler notation $T_{p}$.

Let $h_{F}$ be the metric on $T_{p}$ that is conformally related to $h|_{T_{p}}$ and that is equal to $h$ at $p$ (Proposition \ref{FLATTGA}). On $T_{p}\times I$ define
\be
\hg_{F}=dU^{2}+h_{F}.
\ee
Around any point $q\in T_{p}$ we can consider coordinates $(z_{1},z_{2})$ such that $h_{F}=dz_{1}^{2}+dz_{2}^{2}$. On every patch $(z_{1},z_{2},U)$ we have $\hg_{F}=dU^{2}+dz_{1}^{2}+dz_{2}^{2}$. For this reason the $h_{F}$-covariant derivative on the tori $T_{U}$ will be denoted by $\partial_{A}$ or simply $\partial$.

We claim that,
\begin{enumerate}[labelindent=\parindent, leftmargin=*, label=(\roman*), widest=IV, align=left]
\item\label{PORSI-i} $e^{-C_{0}}h_{F}\leq h\leq e^{C_{0}}h_{F}$, where $C_{0}>0$ is uniform,
\item\label{PORSI-ii} for any $i\geq 0$ and $l\geq 0$, $|\partial^{l}_{U}\partial^{(i)} h|_{h_{F}}$ and $|\partial^{l}_{U}\partial^{(i)}\lambda|_{h_{F}}$ are uniformly bounded. 
\end{enumerate}

Of course these uniform bounds should be understood to hold at every point of every $T_{U}$ in $\mathcal{U}_{p}$. 

We prove first \ref{PORSI-i}. We start showing that for every $i\geq 0$, $|D^{i}\lambda|_{h}$, $|D^{(i)}\Theta|_{h}$ and $|D^{i}\theta|_{h}$ are uniformly bounded. Let $v$ and $w$ be two unit vectors tangent to a $T_{U}$ at one point. A normal unit vector to $T_{U}$ is $n^{a}=\lambda\nabla^{a}U$. Then we compute,
\be\label{SECFUNDBC}
\Theta(v,w)=\langle \nabla_{v} (\lambda\nabla U),w\rangle = \big(\lambda\nabla_{a}\nabla_{b}U+(\nabla_{a}\lambda) \nabla_{b}U\big)v^{a}w^{b}
\ee
By the estimates (\ref{ESTCHEC}) of Part I, $|\nabla_{a}U|_{\hg}$ and $|\nabla_{a}\nabla_{b}U|_{\hg}$ are uniformly bounded. Similarly, as mentioned in \ref{PORSI-I}, $\lambda$ and $|\nabla \lambda|_{\hg}$ are uniformly bounded. Hence $|\Theta|_{h}$ is uniformly bounded. For the same reason $\nabla$-derivatives of $\Theta$ are uniformly bounded, and therefore are the $D$-derivatives because $\nabla$ and $D$ differ from each other in $\Theta$. These bounds imply the uniform bounds also for $|D^{i}\lambda|_{h}$ and $|D^{i}\theta|_{h}$.

Recall that the Gaussian curvature $\kappa$ of the metric $h$ on a slice $T_{U}$ is given by
\be
2\kappa=-|\Theta|^{2}-\theta^{2}-\frac{2}{\lambda^{2}}.
\ee
The previous estimates then show that for every $i\geq 0$, $|D^{(i)}\kappa|_{h}$ is also uniformly bounded.

So far these uniform bounds hold without the need to adjust $\epsilon^{*}$ or $r^{*}$, because they are due essentially to the bounds (\ref{ESTCHEC}) of Part I and the uniform bounds for $\rho$. In the sequel we may need however further adjustment. Chose then $\epsilon^{*}$ sufficiently small such that $\diam_{h}(T_{p})$ is small enough that we can use Proposition \ref{FLATTGA} on $T_{p}$ to conclude first that
\be\label{INCLUDE}
e^{-K_{0}}h_{F}\leq h\big|_{T_{p}}\leq e^{K_{0}}h_{F}
\ee
where $K_{0}>0$ is uniform and second that for any $i\geq 1$, $|\partial^{(i)} h|_{T_{p}}|_{h_{F}}$ is uniformly bounded. 

Now we explain how \ref{PORSI-i} is a simple consequence of the boundedness of the second fundamental forms. Recall that
\be\label{GGQQ}
\partial_{U}h=2\lambda\Theta
\ee
As $\lambda$ is uniformly bounded and as $e^{-K_{1}}h\leq \Theta\leq e^{K_{1}}h$ at every $T_{U}$ and for some uniform $K_{1}>0$, we deduce that $e^{-K_{2}}h\leq \partial_{U} h\leq e^{K_{2}}h$ for some uniform $K_{2}>0$. After integration in $U$ we obtain $e^{-K_{3}}h|_{T_{p}}\leq h \leq e^{K_{3}}h|_{T_{p}}$ for some uniform $K_{3}>0$, which is equivalent to $e^{-C_{0}}h_{F}\leq h \leq e^{C_{0}}h_{F}$ for a uniform $C_{0}>0$ because of (\ref{INCLUDE}). 

We turn to prove \ref{PORSI-ii}. We have mentioned already that $|\nabla \lambda|_{\hg}$ is uniformly bounded. Thus, $|\partial_{U} \lambda|(=|\rho^{2}\langle\nabla U,\nabla \lambda\rangle|)$ is uniformly bounded and so is $|\partial \lambda|_{h_{F}}$ by (\ref{INCLUDE}). We prove then that $|\partial_{U}h|_{h_{F}}$ and $|\partial h|_{h_{F}}$ are uniformly bounded. The uniform bound for $|\partial_{U}h|_{h_{F}}$ follows directly from the formula (\ref{GGQQ}), the uniform bound of $\lambda$ and of $|\Theta|_{h}$, and \ref{PORSI-i}. Let us turn now to prove the uniform bound for $|\partial h|_{h_{F}}$. We work in coordinates. We compute
\be
\partial_{U}\partial_{C}h_{AB}=2(\partial_{C}\lambda)\Theta_{AB}+2\lambda \partial_{C}\Theta_{AB}
\ee
where we can write
\be
\partial_{C}\Theta_{AB}=D_{C}\Theta_{AB}+\Gamma_{CA}^{M}\Theta_{MB}+\Gamma_{CB}^{M}\Theta_{AM}
\ee
with the Levi-Civita connection $\Gamma_{AB}^{C}$ being
\be
\Gamma_{AB}^{C}=\frac{1}{2}\{\partial_{A}h_{MB}+\partial_{B}h_{AM}-\partial_{M}h_{AB}\}h^{MC}
\ee
Hence, relying on the estimates previously obtained we can write
\be
\partial_{U}(\partial_{C}h_{AB})=X_{CAB}^{\ \ \ C'A'B'}(\partial_{C'}h_{A'B'})+Y_{CAB}
\ee
where $|X_{CAB}^{\ \ \ C'A'B'}|$ and $|Y_{CAB}|$ are uniformly bounded. Using this system of first order ODEs and the uniform bound for $|\partial h|_{T_{p}}|_{h_{F}}$ at the initial slice $T_{p}$, we get directly the desired uniform boundedness of $|\partial_{C}h_{AB}|$. 

Proving that  for every $i\geq 0$ and $l\geq 0$, $|\partial^{l}_{U}\partial^{(i)}\lambda|_{h_{F}}$ and $|\partial_{U}^{l}\partial^{(i)} h|_{h_{F}}$ are uniformly bounded, is done by the iteration of the same arguments. 
\vs

\ref{PORSI-II}-B. {\sc A `good' $\Sa\times \Sa$-symmetric approximation $\breve{\hg}$ of $\hg$.} We explain first how to define $\breve{\hg}$ and then we explain how well it does approximate $\hg$. Let $p_{0}$ be a point in $T_{p}$ where the Gaussian curvature is zero. The choice of $p_{0}$ will play some role that we will explain later. Then define
\be
\breve{\hg}=\breve{\lambda}^{2}dU^{2}+\breve{h}
\ee
where $\breve{\lambda}$ and $\breve{h}$ are, at every leaf $T_U$, simply the $h_{F}$-extensions of $\lambda(\zeta_{p_{0}}(U))$ and $h|_{\zeta_{p_{0}}(U)}$ respectively (recall the notion of $h_{F}$-extension right before the statement of the proposition). Note, in particular, that $h-\breve{h}$ and $\lambda-\breve{\lambda}$ are zero all over $\zeta_{p_{0}}(U)$.

We prove now that for every $i\geq 0$ and $l\geq 0$ there is a uniform $C>0$ such that 
\begin{gather}
\label{FOR1}
|\partial_{U}^{l}\partial^{(i)}(h-\breve{h})|_{h_{F}}\leq C\diam^{m^{*}}_{h_{F}}(T_{p}),\\ 
\label{FOR11}
|\partial_{U}^{l}\partial^{(i)}(\lambda -\breve{\lambda})|_{h_{F}}\leq C\diam^{m^{*}}_{h_{F}}(T_{p})
\end{gather}
Fix $i$ and $l$. In a coordinate patch $(z_{1},z_{2},U)$ around $\zeta_{p_{0}}=p_{0}\times I$, ($p_{0}=(0,0)$), we have
\be
\breve{h}_{AB}(z_{1},z_{2},U)=h_{AB}(0,0,U),\quad \breve{\lambda}(z_{1},z_{2},U)=\lambda(0,0,U)
\ee
for all $(z_{1},z_{2},U)$. Taking $\partial_{U}$-derivatives we deduce that for every $l'\geq 0$, also $(\partial_{U}^{l'}\breve{h})|_{T_{U}}$ and $(\partial_{U}^{l'}\breve{\lambda})|_{T_{U}}$ are the $h_{F}$-extensions of $(\partial_{U}^{l'}h)|_{\zeta_{p_{0}}(U)}$ and $(\partial_{U}^{l'}\lambda)|_{\zeta_{p_{0}}(U)}$ respectively. Therefore $\partial^{l'}_{U}(h-\breve{h})$ and $\partial^{l'}_{U}(\lambda-\breve{\lambda})$ are zero at every point on $\zeta_{p_{0}}(U)$. If we prove that in addition for every $i'\geq 0$ and $l'\geq 0$, $|\partial^{(i')}\partial_{U}^{l'} (h-\breve{h})|_{h_{F}}$ is uniformly bounded then the $C^{i+m^{*}}_{h_{F}}$-norm of $\partial_{U}^{l}(h-\breve{h})$ on every $T_{U}$ would be uniformly bounded. We could then use Proposition \ref{ENDK1} at every tori $T_{U}$, (in Proposition \ref{ENDK1} use $W=\partial^{l}_{U}(h-\breve{h})$, $k=i+m^{*}$ and $j=i$), to conclude (\ref{FOR1}) and (\ref{FOR11}). Let us prove then these bounds.

First, as $(\partial_{U}^{l'}\breve{h})|_{T_{U}}$ is the $h_{F}$ extension of $(\partial_{U}^{l'}h)|_{\zeta_{p_{0}}(U)}$, then at every point $q$ in a torus $T_{U}$ we have $|\partial^{l'}_{U}\breve{h}|_{h_{F}}(q)=|\partial^{l'}_{U}\breve{h}|_{h_{F}}(\zeta_{p_{0}}(U))=|\partial^{l'}_{U}h|_{h_{F}}(\zeta_{p_{0}}(U))$. But by (ii), for every $l'\geq 0$, $|\partial^{l'}_{U}h|_{h_{F}}$ is uniformly bounded, hence $|\partial^{l'}_{U}(h-\breve{h})|_{h_{F}}(\leq |\partial^{l'}_{U}h|_{h_{F}}+|\partial^{l'}_{U}\breve{h}|_{h_{F}})$, is uniformly bounded.

In second place, as $(\partial_{U}^{l'}\breve{h})|_{T_{U}}$ and $(\partial_{U}^{l'}\breve{\lambda})|_{T_{U}}$ are the $h_{F}$-extensions of $(\partial_{U}^{l'}h)|_{\zeta_{p_{0}}(U)}$ and $(\partial_{U}^{l'}\lambda)|_{\zeta_{p_{0}}(U)}$ respectively then for any $i'\geq 1$ we have $\partial^{(i')}\partial^{l'}_{U}\breve{h}=0$ and $\partial^{(i')}\partial^{l'}_{U}\breve{\lambda}=0$. Therefore,
\be
\partial^{l}_{U}\partial^{(i')}(h-\breve{h})=\partial^{l}_{U}\partial^{(i')}h\quad {\rm and} \quad\partial^{l}_{U}\partial^{(i')}(\lambda-\breve{\lambda})=\partial^{l}_{U}\partial^{(i')}\lambda
\ee 
By the estimates \ref{PORSI-i} and \ref{PORSI-ii} in \ref{PORSI-I}, the $h_{F}$-norm of the right hand side of each of these expressions is uniformly bounded. This concludes the proof of the bounds that we claimed above.

These estimates imply now, for any $i\geq 0$ and $l\geq 0$, we have
\begin{gather}
\label{FOR2}
|\partial^{l}_{U}\partial^{(i)}DD(\lambda-\breve{\lambda})|_{h_{F}}\leq C_{li}\diam_{h_{F}}^{m^{*}}(T_{p}),
\end{gather}
where the $C_{li}$ are uniform (use that $D=\partial +\Gamma$). This is the estimate that will be used in \ref{PORSI-II}-C.
\vs

\ref{PORSI-II}-C. {\sc The Kasner approximation $\hg^{\mathbb{K}}$ of $\hg$.} In coordinates $(z_{1},z_{2},U)$ the static equations are
\begin{align}
\label{SSA1} & \partial_{U}h_{AB}=2\lambda \Theta_{AB},\\
\label{SSA2} & \partial_{U}\Theta_{AB}=-D_{A}D_{B}\lambda +\lambda(2\kappa h_{AB}-\theta \Theta_{AB}+2\Theta_{AC}\Theta^{C}_{\ B}),\\
\label{SSA5} & \partial_{U}\bigg(\frac{\sqrt{|h|}}{\lambda}\bigg)=0,\\
\label{SSA3} & \Theta_{AB}\Theta^{AB}-\theta^{2}=-\frac{2}{\lambda^{2}}-2\kappa,\\
\label{SSA4} & D^{A}\Theta_{AB}=D_{B}\theta,
\end{align}
where, as earlier, $\theta=\Theta_{A}^{\ A}$. The equation (\ref{SSA5}) is the same as $\Delta U=0$ and is equivalent to 
\be\label{SSA6}
\partial_{U}\lambda=\lambda^{2}\theta
\ee
We will use this equation instead of (\ref{SSA5}). 

Evaluating (\ref{SSA1}), (\ref{SSA2}), (\ref{SSA3}), (\ref{SSA6}) and (\ref{SSA4}) at $\zeta_{p_{0}}(U)$ and (\ref{SSA3}) at $p_{0}$ we get,
\begin{align}
\label{SSA1B} & \partial_{U}\breve{h}_{AB}=2\breve{\lambda} \breve{\Theta}_{AB},\\
\label{SSA2B} & \partial_{U}\breve{\Theta}_{AB}=\breve{\lambda}(2\overline{\kappa}\breve{h}_{AB}-\breve{\theta} \breve{\Theta}_{AB}+2\breve{\Theta}_{AC}\breve{\Theta}^{C}_{\ B})+O^{\infty}_{AB}(\diam_{h_{F}}^{m^{*}}(T_{p})),\\
\label{SSA5B} &\partial_{U}\breve{\lambda}=\breve{\lambda}^{2}\breve{\theta},\\
\label{SSA3B} & \big(\breve{\Theta}_{AB}\breve{\Theta}^{AB}-\breve{\theta}^{2}\big)\bigg|_{p_{0}}=-\frac{2}{\breve{\lambda}^{2}}\bigg|_{p_{0}},\\
\label{SSA4B} &\breve{\partial}^{A}\breve{\Theta}_{AB}=\partial_{B}\breve{\theta},
\end{align}
where $\overline{\kappa}$ is defined as
\be
\overline{\kappa}=\bigg[-\frac{1}{\breve{\lambda}^{2}}-\frac{1}{2}\big(\breve{\Theta}_{AB}\breve{\Theta}^{AB}-\breve{\theta}^{2}\big)
\bigg]\bigg|_{p_{0}}
\ee
(and is not the Gaussian curvature of $\breve{h}$ which is zero) and where $O^{\infty}_{AB}$ is
\be\label{ODDI}
O^{\infty}_{AB}=-D_{A}D_{B}\lambda.
\ee
This notation is to stress that, as was shown in (\ref{FOR2}), for any $l\geq 0$ we have
\be
|\partial^{l}_{U}O^{\infty}_{AB}|_{h_{F}}\leq C_{l}\diam^{m^{*}}_{h_{F}}(T_{p})
\ee
where $C_{l}$ is uniform.

Consider now the metric 
\be\label{GKMETR}
\hg^{\mathbb{K}}=(\lambda^{\mathbb{K}})^{2}dU^{2}+h^{\mathbb{K}},
\ee
where $\lambda^{\mathbb{K}}=\lambda^{\mathbb{K}}(U)$ and $h^{\mathbb{K}}=h^{\mathbb{K}}(U)$ solve
\begin{align}
\label{RTE41} & \partial_{U}h^{\mathbb{K}}_{AB}=2\lambda^{\mathbb{K}} \Theta^{\mathbb{K}}_{AB},\\
\label{RTE42} & \partial_{U}\Theta^{\mathbb{K}}_{AB}=\lambda^{\mathbb{K}}(-\theta^{\mathbb{K}}\Theta^{\mathbb{K}}_{AB}+2\Theta^{\mathbb{K}}_{AC}\Theta^{\mathbb{K}C}_{\ B}),\\
\label{RTE43} & \partial_{U}\lambda^{\mathbb{K}}=(\lambda^{\mathbb{K}})^{2}\theta^{\mathbb{K}}
\end{align}
subject to the initial data
\be
h^{\mathbb{K}}_{AB}(0)=\breve{h}_{AB}(0),\quad \Theta_{AB}^{\mathbb{K}}(0)=\breve{\Theta}_{AB}(0)\quad \text{and}\quad \lambda^{\mathbb{K}}(0)=\breve{\lambda}(0).
\ee
Following the discussion in Section \ref{UNIQ}, we see that $(\lambda^{\mathbb{K}}(U),h^{\mathbb{K}}(U))$ satisfy (\ref{RTE1}), (\ref{RTE2}) and (\ref{RTE43}) for all $U$, and (\ref{RTE3}) at the initial time, hence is a Kasner solution. Hence 
\be\label{ZERO}
0=-\frac{1}{(\lambda^{\mathbb{K}})^{2}}-\frac{1}{2}\big(\Theta^{\mathbb{K}}_{AB}\Theta^{\mathbb{K}AB}-(\theta^{\mathbb{K}})^{2}\big)
\ee
on each $T_{U}$. Thus, (\ref{RTE42}) is equivalent to 
\be
\partial_{U}\Theta^{\mathbb{K}}_{AB}=\lambda^{\mathbb{K}}(2\overline{\kappa}^{\mathbb{K}}-\theta^{\mathbb{K}}\Theta^{\mathbb{K}}_{AB}+2\Theta^{\mathbb{K}}_{AC}\Theta^{\mathbb{K}C}_{\ B}),
\ee
where $\overline{\kappa}^{\mathbb{K}}$ is the right hand side of (\ref{ZERO}) and is zero. Therefore, thought as ODE's, the system (\ref{SSA1B}), (\ref{SSA2B}), (\ref{SSA5B}) is a perturbation of the system  (\ref{RTE41}), (\ref{RTE42}), (\ref{RTE43}) where the 'perturbation' is $O^{\infty}_{AB}$ and should be thought as depending only on $U$. Both systems have also the same initial data. Therefore, using (\ref{ODDI}) and standard ODE analysis we obtain
\begin{gather}
\label{FCASE1} |\partial^{l}_{U}(\breve{h}-h^{\mathbb{K}})|_{h_{F}}\leq C_{l}^{*}\diam^{m^{*}}_{h_{F}}(T_{p}),\\
\label{FCASE2}|\partial^{l}_{U}(\breve{\lambda}-\lambda^{\mathbb{K}})|\leq C^{*}_{l}\diam^{m^{*}}_{h_{F}}(T_{p})
\end{gather}
for any $l\geq 0$, where the $C^{*}_{l}$ are uniform. Now note that because $\partial^{(i)}h_{\mathbb{K}}=\partial^{(i)}\breve{h}=0$ then for every $i\geq 1$ we have
\begin{gather}
\partial_{U}^{l}\partial^{(i)}(h-h^{\mathbb{K}})=\partial_{U}^{l}\partial^{(i)}(h-\breve{h}),\\
\partial_{U}^{l}\partial^{(i)}(\lambda-\lambda^{\mathbb{K}})=\partial_{U}^{l}\partial^{(i)}(\lambda-\breve{\lambda})
\end{gather}
Thus, from (\ref{FOR1}) and (\ref{FOR11}) we obtain
\begin{gather}
\label{FCASE3} |\partial_{U}^{l}\partial^{(i)}(h-h^{\mathbb{K}})|_{h_{F}}\leq C_{li}^{*}\diam^{m^{*}}_{h_{F}}(T_{p}),\\
\label{FCASE4} |\partial_{U}^{l}\partial^{(i)}(\lambda-\lambda^{\mathbb{K}})|_{h_{F}}\leq C_{li}^{*}\diam^{m^{*}}_{h_{F}}(T_{p})
\end{gather}
where the $C^{*}_{li}$ are uniform. 

The estimates (\ref{ESTIMATE}) claimed in \ref{PORSI-II} are equivalent to (\ref{FCASE3}),(\ref{FCASE4}). This finishes the proof of the Proposition.
\end{proof} 

\begin{Proposition}\label{PORE} Let $(\Sigma; \hg, U)$ be a static end, and let $\gamma$ be a ray. Let $0<\rho^{*}\leq 1/2$ and let $j^{*}\geq 1$ and $m^{*}\geq 2$. Let $\epsilon^{*}, \mu^{*}, r^*, C^*$, be as in Proposition \ref{PORSI}. Then, there exist positive $\delta^*$, $\ell^*$ and $B^{*}$ such that if $p$ is a point in $\gamma$ with $r=r(p)\geq r^*$ satisfying,
\begin{enumerate}[labelindent=\parindent, leftmargin=*, label={\rm (\alph*)}, widest=a, align=left]
\item\label{PORE1} $\dist_{GH}\big(\big(\mathcal{A}^{c}_{r}(p;1/2,2); d_{r}\big),\big([1/2,2];|\ldots|\big)\big)\leq \epsilon^*$,
\item\label{PORE2} $|\rho_{r}(p)-\rho^*|\leq \mu^{*}$,
\item\label{PORE3} $|\theta_{r}(p)-1|\leq \delta^*$,
\item\label{PORE4} $\diam_{\hg^{\mathbb{K}}}(\phi(T_p))\leq \ell^*$,
\end{enumerate}
(where $\U_{p}$, $(\U^{\K};\hg^{\K},U^{\K})$ and $\phi:\U_{p}\rightarrow \U^{\K}$ are respectively the neighbourhood, the Kasner data and the diffeomoprhisms from \ref{PORSI-I} and \ref{PORSI-II} of Proposition \ref{PORSI}),
and $p'$ is a point in $\gamma$ with $r':=r(p')=a^* r$, then the following holds,
\begin{enumerate}[labelindent=\parindent, leftmargin=*, label={\rm (\Roman*)}, widest=a, align=left]
\item\label{POREI} $\dist_{GH}\big(\big(\mathcal{A}^{c}_{a^* r}(p';1/2,2); d_{a^* r}\big),\big([1/2,2];|\ldots|\big)\big)\leq \epsilon^*/2$,
\item\label{POREII} $\diam_{\hg_{a^{*}}^{\mathbb{K}}}(T_{p'})\leq \diam_{\hg^{\mathbb{K}}}(T_{p})/2$,
\item\label{POREIII} $|\theta_{r'}(p')-1|\leq B^{*}\diam^{2}_{\hg^{\mathbb{K}}}(T_{p})+|\theta_{r}(p)-1|/2$,
\item\label{POREIV} $|\rho_{r'}(p')-\rho_{r}(p)|\leq B^{*}\diam^{2}_{\hg^{\mathbb{K}}}(T_{p})+|\theta_{r}(p)-1|/2$. 
\end{enumerate}
\end{Proposition} 
 
\begin{proof} 
Proceeding by contradiction we assume that for each $\delta^{*}_{i}=1/i$, $\ell^{*}_{i}=1/i$ and $B^{*}_{i}=i$, there is $p_{i}\in \gamma$ with $r(p_{i})\geq r^{*}$ satisfying \ref{PORE1}-\ref{PORE4}, and there is $p'_{i}\in \gamma$ with $r'_{i}=r(p'_{i})=a^{*}r(p_{i})$ such that either \ref{POREI}, \ref{POREII}, \ref{POREIII} or \ref{POREIV} does not hold. 

We prove now that for $i\geq i_{0}$ with $i_{0}$ large enough, indeed all \ref{POREI}, \ref{POREII}, \ref{POREIII} and \ref{POREIV} must hold.

\ref{POREI} As $\diam_{\hg_{i}^{\mathbb{K}}}(\phi(T_{p_{i}}))\rightarrow 0$, then the metric distance between $(\mathcal{U}_{p_{i}};\hg_{r_{i}})$ and $(T_{p_{i}}\times I_{i}; \hg_{i}^{\mathbb{K}_{i}})$ tends to zero as $i\rightarrow \infty$ and at the same time the spaces $(T_{p_{i}}\times I_{i}; \hg_{i}^{\mathbb{K}_{i}})$ collapse metrically to a segment of length $(2a^{*}-1/(2a^{*}))$. Hence so does $(\mathcal{U}_{p_{i}};\hg_{r_{i}})$. As $\mathcal{U}_{p_{i}}$ contains $\mathcal{A}^{c}_{r_{i}}(p_{i};1/(2a^{*}),2a^{*})$ and therefore $\mathcal{A}^{c}_{r'_{i}}(p'_{i};1/2,2)$, these last annuli metrically collapse to a segment of length $2-1/2$. Hence \ref{POREI} must hold for $i$ sufficiently large.

\ref{POREII} Let $c_{i}$ be the Kasner parameter of the Kasner space $\mathbb{K}_{i}$. Then by \ref{PORE2}, for sufficiently large $i$ we have $c_{i}>  \rho^{*}/4$ hence \ref{POREII} must hold by the definition (\ref{ASDEF}) of $a^{*}$.

\ref{POREIII}
We write
\be\label{LOLA}
|\theta_{r'_{i}}(p'_{i})-1|\leq |\theta_{r'_{i}}(p'_{i})-\theta^{\mathbb{K}_{i}}_{a^{*}}(p'_{i})|+|\theta^{\mathbb{K}_{i}}_{a^{*}}(p'_{i})-1|
\ee
where $\theta^{\mathbb{K}_{i}}_{a^{*}}(T_{p'_{i}})$ is the mean curvature of the slice $T_{p'_{i}}$ with respect to the Kasner metric $(1/a^{*})^{2}
\hg^{\mathbb{K}_{i}}$, namely $\theta^{\mathbb{K}_{i}}_{a^{*}}(T_{p'_{i}})=a^{*}\theta^{\mathbb{K}_{i}}(T_{p'_{i}})$. Similarly, as $r'_{i}=a^{*}r_{i}$ we have $\theta_{r'_{i}}(p'_{i})=a^{*}\theta_{r_{i}}(p'_{i})$. Therefore for the first term in the right hand side of (\ref{LOLA}) we can write 
\be\label{COMB22}
|\theta_{r'_{i}}(p'_{i})-\theta^{\mathbb{K}}_{a^{*}}(p'_{i})| = a^{*}|\theta_{r_{i}}(p'_{i})-\theta^{\mathbb{K}}(p'_{i})|\leq C^{*}_{1}\diam^{2}_{h_{\mathbb{K}_{i}}}(T_{p_{i}})
\ee
where the last inequality is from \ref{POREII} in Proposition \ref{PORE} with $m^{*}\geq 2$, $j^{*}\geq 1$. 

Write the Kasner metric $\hg^{\mathbb{K}_{i}}$ as
\be
\hg^{\mathbb{K}_{i}}=dx^{2}+x^{2a_{i}}d\varphi_{1}^{2}+x^{2b_{i}}d\varphi_{2}^{2}=(\lambda^{\mathbb{K}_{i}})^{2}dU^{2}+h^{\mathbb{K}_{i}}
\ee
and let $x(p_{i})=x_{i}$ and $x(p'_{i})=x'_{i}$. Then, 
\be
\theta^{\mathbb{K}_{i}}(p_{i})=\frac{1}{x_{i}},\quad {\rm and}\quad \theta^{\mathbb{K}_{i}}(p'_{i})=\frac{1}{x'_{i}}
\ee
and,
\be\label{TOSUB1}
x_{i}'-x_{i}=\int \lambda^{\mathbb{K}_{i}}dU
\ee
where the integral is along any integral line of $\nabla^{a} U$. 

On the other hand the $\hg_{r_{i}}$-length of the segment of $\gamma$ between $p_{i}$ and $p'_{i}$, is equal to $a^{*}-1$. This length is equal, up to an $O(\diam^{2}_{h^{\mathbb{K}_{i}}}(T_{p_{i}}))$ to the $\hg_{r_{i}}$-length of any integral line of $\nabla^{a} U$ between $T_{p_{i}}$ and $T_{p'_{i}}$. So,
\be\label{TOSUB2}
a^{*}-1=\int \lambda_{r_{i}}dU+O(\diam^{2}_{h^{\mathbb{K}_{i}}}(T_{p_{i}}))
\ee
But by Proposition \ref{PORSI} we have $|\lambda_{r_{i}}-\lambda^{\mathbb{K}_{i}}|\leq O(\diam^{2}_{h_{\mathbb{K}_{i}}}(T_{p_{i}}))$. Subtract (\ref{TOSUB1}) and (\ref{TOSUB2}) to get
\be\label{EQFX}
x_{i}'=x_{i}+(a^{*}-1)+O(\diam^{2}_{h_{\mathbb{K}_{i}}}(T_{p_{i}}))
\ee
Thus
\be
\theta^{\mathbb{K}_{i}}_{a^{*}}(p'_{i})=\frac{a^{*}}{x_{i}+a^{*}-1+O(\diam^{2}_{h_{\mathbb{K}_{i}}}(T_{p_{i}}))}
\ee
Then we calculate
\begin{align}
\label{COMB3} |\theta^{\mathbb{K}_{i}}_{a^{*}}(p'_{i})-1| & =\bigg|\frac{x_{i}-1+O(\diam^{2}_{h_{\mathbb{K}_{i}}}(T_{p_{i}}))}{x_{i}+a^{*}-1+O(\diam^{2}_{h_{\mathbb{K}_{i}}}(T_{p_{i}}))}\bigg|\\
\label{COMB4} & \leq \frac{1}{2}\bigg|\frac{1}{x_{i}}-1\bigg| + C^{*}_{3}\diam^{2}_{h_{\mathbb{K}_{i}}}(T_{p_{i}})
\end{align}
where to obtain the bound we used that $x_{i}\rightarrow 1$ and that $a^{*}\geq 4$ (see definition of $a^{*}$). But 
\be
\frac{1}{x_{i}}=\theta^{\mathbb{K}_{i}}(p_{i})=\theta^{\mathbb{K}_{i}}(p_{0i})=\theta_{r_{i}}(p_{0i})
\ee
where $p_{0i}$ is the point over $T_{p_{i}}$ that is used in the construction of $\hg^{\mathbb{K}_{i}}$ in \ref{PORSI-II}-C in Proposition \ref{PORSI}. But again by Proposition \ref{PORSI} we have,
\be
|\theta_{r_{i}}(p_{i})-\theta_{r_{i}}(p_{0i})|\leq C_{4}^{*}\diam^{2}_{h_{\mathbb{K}_{i}}}(T_{p_{i}})
\ee
and thus
\be\label{COMB5}
\bigg|\frac{1}{x_{i}}-1\bigg|=|\theta_{r_{i}}(p_{0i})-1|\leq |\theta_{r_{i}}(p_{i})-1|+C_{4}^{*}\diam^{2}_{h_{\mathbb{K}_{i}}}(T_{p_{i}})
\ee
Combining now (\ref{LOLA}), (\ref{COMB22}), (\ref{COMB3})-(\ref{COMB4}) and (\ref{COMB5}) we deduce that \ref{POREIII} also holds for $i$ sufficiently large.

\ref{POREIV} This follows the same arguments as in \ref{POREIII}. Write,
\begin{align}
\label{POREIVE1} |\rho_{r'_{i}}(p'_{i})-\rho_{r_{i}}(p_{i})| \leq &\ |\rho_{r'_{i}}(p'_{i})-\rho^{\mathbb{K}_{i}}_{a^{*}}(p'_{i})|+|\rho^{\mathbb{K}_{i}}(p_{i})-\rho_{r_{i}}(p_{i})|\\
\label{POREIVE2} &+|\rho^{\mathbb{K}_{i}}_{a^{*}}(p'_{i})-\rho^{\mathbb{K}_{i}}(p_{i})|
\end{align}
The two terms on the right hand side of (\ref{POREIVE1}) are bounded by $O(\diam^{2}_{h_{\mathbb{K}_{i}}}(T_{p_{i}}))$ by Proposition \ref{PORSI} with $m^{*}\geq 2$. On the other hand following notation as in \ref{POREIII}, write $U^{\mathbb{K}_{i}}=c_{i}\ln x$ with $c_{i}\rightarrow \rho^{*}$. Then the term in (\ref{POREIVE2}) is equal to 
\be
\bigg|a^{*}\frac{c_{i}}{x'_{i}}-\frac{c_{i}}{x_{i}}\bigg|
\ee
and using (\ref{EQFX}) we can easily manipulate this expression to obtain the bound 
\be
|x_{i}-1|/2+O(\diam^{2}_{h_{\mathbb{K}_{i}}}(T_{p_{i}}))
\ee
because $a^{*}\geq 4$ and $\rho^{*}>0$. Finally use (\ref{COMB5}) to bound this expression once more and obtain \ref{POREIV}.
\end{proof} 

\begin{Theorem}\label{KASYMPTOTIC} {\rm (A characterisation of KA $\neq A, C$)} Let $(\Sigma;\hg,U)$ be a static end. Let $\gamma$ be a ray and suppose that there is a sequence $p_{i}\in \gamma$ such that $\rho_{r_{i}}(p_{i})\rightarrow \rho^{*}$, with $0<\rho^{*}\leq 1/2$ and that $(\mathcal{A}^{c}_{r_{i}}(p_{i};1/2,2); \hg_{r_{i}})$ metrically collapses to a segment $([1/2,2];|\ldots|)$. Then the end is asymptotically Kasner different from $A$ and $C$.. 
\end{Theorem}

\begin{proof} For the $\rho^{*}$ given in the hypothesis and for any integers $j^{*}\geq 1$ and $m^{*}\geq 2$ let $\epsilon^{*}$, $\mu^{*}$, $r^{*}$ and $C^{*}$ be as in Proposition \ref{PORSI}, and let $\delta^{*}$, $\ell^{*}$ and $B^{*}$ be as in Proposition \ref{PORE}. We begin proving that there are $\mu^{**}\leq \mu^{*}$, $\delta^{**}\leq \delta^{*}$ and $\ell^{**}\leq \ell^{*}$ such that if for $i$ big enough the point $p^{0}:=p_{i}$ is such that,
\begin{enumerate}[labelindent=\parindent, leftmargin=*, label={\rm (\alph*')}, widest=a, align=left]
\item\label{PORE1b} $\dist_{GH}\big(\big(\mathcal{A}^{c}_{r}(p^{0};1/2,2); d_{r}\big),\big([1/2,2];|\ldots|\big)\big)\leq \epsilon^*$,
\item\label{PORE2b} $|\rho_{r^{0}}(p^{0})-\rho^*|\leq \mu^{**}$,
\item\label{PORE3b} $|\theta_{r^{0}}(p^{0})-1|\leq \delta^{**}$,
\item\label{PORE4b} $\diam_{h_{\mathbb{K}^{0}}}(\phi(T_{p^{0}}))\leq \ell^{**}$,
\end{enumerate}
then for all $p^{n}\in \gamma$ such that $\mathfrak{r}_{n}:=r(p^{n})=(a^{*})^{n}r(p_{0})$ we have
\begin{enumerate}[labelindent=\parindent, leftmargin=*, label={\rm (\alph*)}, widest=a, align=left]
\item\label{PORE1b2} $\dist_{GH}\big(\big(\mathcal{A}^{c}_{\mathfrak{r}_{n}}(p^{n};1/2,2); d_{\mathfrak{r}_{n}}\big),\big([1/2,2];|\ldots|\big)\big)\leq \epsilon^*$,
\item\label{PORE2b2} $|\rho_{\mathfrak{r}_{n}}(p^{n})-\rho^*|\leq \mu^{*}$,
\item\label{PORE3b2} $|\theta_{\mathfrak{r}_{n}}(p^{n})-1|\leq \delta^{*}$,
\item\label{PORE4b2} $\diam_{h_{\mathbb{K}^{n}}}(\phi(T_{p^{n}}))\leq \ell^{*}2^{-n}$.
\end{enumerate}
To choose $\epsilon^{**}$, $\delta^{**}$ and $\mu^{**}$ we make the following observation. Suppose that for some $\mu^{**}\leq \mu^{*}$, $\delta^{**}\leq \delta^{*}$ and $\ell^{**}\leq \ell^{*}$, \ref{PORE1},\ref{PORE2},\ref{PORE3} and \ref{PORE4} hold for $p^{n}$ for $n=0,1,2,3,\ldots,m\geq 1$. Then, after using the conclusions \ref{POREI},\ref{POREII} and \ref{POREIII} in Proposition \ref{PORE} $m$-times (each time use Prop \ref{PORE} with $p=p^{n},p'=p^{n+1}$) one obtains without difficulty the bounds,
\begin{align}
\label{RHSI} & \diam_{h_{\mathbb{K}^{m}}}(\phi(T_{p^{m}}))\leq \frac{\ell^{**}}{2^{m-1}},\\
\label{RHSII} & |\theta_{r^{m}}(p^{m})-1|\leq \frac{mB^{*}\ell^{**}}{2^{m-1}}+\frac{\delta^{**}}{2^{m}},\\
\label{RHSIII} & |\rho_{\mathfrak{r}_{n}}(p^{n})-\rho_{r^{0}}(p^{0})|\leq \sum_{n=1}^{n=m}\bigg(\frac{B^{*}(\ell^{**})^{2}}{2^{2(n-1)}}+\frac{nB^{*}\ell^{**}}{2^{n}}+\frac{\delta^{**}}{2^{n+1}}\bigg)
\end{align}
With this information at hand, choose $\mu^{**}=\mu^{*}/4$, and $\delta^{**}\leq \delta^{*}$ and $\ell^{**}\leq \ell^{*}$ such that the right hand side of (\ref{RHSII}) is less or equal than $\delta^{*}/2$ for all $m\geq 1$ and, when in (\ref{RHSIII}) we consider $m=\infty$ (i.e. the infinite sum), this sum is less or equal than $\mu^{*}/4$. Chosed that way it is then trivial that \ref{PORE1},\ref{PORE2},\ref{PORE3} and \ref{PORE4} in this theorem indeed hold for all $p^{n}$, $n=0,1,2,3,\ldots,\infty$.

Having now \ref{PORE1b2} and \ref{PORE2b2} for all $p^{n}$, we can use Proposition \ref{PORSI} to conclude that, for each $n$, 
\begin{enumerate}
\item there are neighbourhoods $\mathcal{U}_{n}$, each covering $\mathcal{A}_{r_{n}}(p^{n};1/(2a^{*}),2a^{*})$ and their union  covering the end of $\Sigma$, and, 
\item there are Kasner spaces $(\mathcal{U}^{\mathbb{K}_{n}};\hg^{\mathbb{K}_{n}},U^{\mathbb{K}_{n}})$, $\T^{2}$-quotients of,
\be\label{KCOSAS}
I_{n}\times \mathbb{R}^{2},\quad \tilde{\hg}^{\mathbb{K}_{n}}=dx^{2}+x^{2a_{n}}dy^{2}+x^{2b_{n}}dy^{2},\quad \tilde{U}^{\mathbb{K}_{n}}=d_{n}+c_{n}\ln x
\ee 
and,
\item there are diffeomorphisms $\phi_{n}:\mathcal{U}_{n}\rightarrow \mathcal{U}^{\mathbb{K}_{n}}=I_{n}\times T^{2}_{n}$ ($T^{2}_{n}$ is the quotient of $\mathbb{R}^{2}$) such that,
\begin{align}
& \label{IINI} \phi_{n*}U=U^{\mathbb{K}_{n}},\\
& \label{IINII} \|\phi_{n*}\hg_{\mathfrak{r}_{n}}-\hg^{\mathbb{K}_{n}}\|_{C^{j^{*}}_{\hg^{\mathbb{K}_{n}}}(\mathcal{U}^{\mathbb{K}_{n}})}\leq C^{*}\diam^{m^{*}}(\phi_{*}(T_{p^{n}}))
\end{align}
\end{enumerate}

What we have so far is close to the Definition \ref{KADEF} of Kasner asymptotic, except that we still need one single Kasner space and one global map $\phi$. Its construction is what we do next. 
\vs

We will work with the `un-scaled' metrics defined by,
\be\label{USCAL}
\hat{\hg}^{\mathbb{K}_{n}}:=(\mathfrak{r}_{n})^{2}\hg^{\mathbb{K}_{n}}=\hg^{\mathbb{K}_{n}}_{1/\mathfrak{r}_{n}},
\ee
and leave $U^{\mathbb{K}_{n}}=U$ unchanged. Thus, we will work with the data sets, $(\mathcal{U}^{\mathbb{K}_{n}}; \hat{\hg}^{\mathbb{K}_{n}},U^{\mathbb{K}_{n}})$. 

From the construction of the trivialisations $\phi$ in step \ref{PORSI-II}-A of Proposition \ref{PORSI}, we obtain that the transition functions,
\be
\phi_{n-1}\circ \phi_{n}^{-1}:\phi_{n}(\mathcal{U}_{n-1}\cap \mathcal{U}_{n})(\subset \mathcal{U}^{\K_{n}})\rightarrow \phi_{n-1}(\mathcal{U}_{n-1}\cap \mathcal{U}_{n})(\subset \mathcal{U}^{\K_{n-1}})
\ee
are defined by just one map $\psi_{n-1,n}:T^{2}_{n}\rightarrow T^{2}_{n-1}$, namely, there is $\psi_{n-1,n}$ such that if 
\be
\phi_{n-1}\circ \phi_{n}^{-1}((x_{n},t_{n}))=(x_{n-1},t_{n-1})
\ee
where $x_{n-1}\in I_{n-1}$, $t_{n-1}\in T^{2}_{n-1}$, $x_{n}\in I_{n}$ and $t_{n}\in T^{2}_{n}$, then 
\be
t_{n-1}=\psi_{n-1,n}(t_{n}). 
\ee
We can use this fact to extend the Kasner data $(\mathcal{U}^{\mathbb{K}_{n}}; \hat{\hg}^{\mathbb{K}_{n}},U^{\mathbb{K}_{n}})$ to a Kasner data on $\mathcal{U}^{\mathbb{K}_{n}}\cup_{\#} \mathcal{U}^{\mathbb{K}_{n-1}}$ where $\#$ means we use the identification $\phi_{n-1}\circ \phi_{n}^{-1}$. The extension is performed as follows. Instead of the coordinate $x$ we use $U$. So, on $\mathcal{U}^{\K_{n-1}}$, $U$ ranges between $U_{1}^{n-1}$, and $U_{2}^{n-1}$, and on $\mathcal{U}^{\K_{n}}$, $U$ ranges between $U_{1}^{n}$ and $U_{2}^{n}$. Now, extend $\hat{g}^{\mathbb{K}_{n}}$ given by (\ref{USCAL}) from $[U^{n}_{1},U^{n}_{2}]\times T^{2}_{n}$ to $[U^{n-1}_{1},U^{n}_{2}]\times T^{2}_{n}$ in the obvious way (by using in (\ref{KCOSAS}) use $U$ instead of $x$), and then identify $[U^{n-1}_{1},U^{n-1}_{2}]\times T^{2}_{n}$ to $[U^{n-1}_{1},U^{n-1}_{2}]\times T^{2}_{n-1}$ by $(U,t_{n})\rightarrow (U,\psi_{n-1,n}(t_{n}))$. In this way we can extend uniquely $\hat{\hg}^{\mathbb{K}_{n}}$, from $\mathcal{U}^{\mathbb{K}_{n}}$, to $\mathcal{U}^{\mathbb{K}_{n-1}}$, then to $\mathcal{U}^{\mathbb{K}_{n-2}}$ and so on until, say, $\mathcal{U}^{\mathbb{K}_{n_{0}}}$. Thus we have a Kasner data,
\be
(\U^{\K_{n_{0}}}\cup_{\#}\ldots\cup_{\#} \U^{\K_{n}}; \hat{\hg}^{\K_{n}},U)
\ee
and we can use the map,
\be
\phi:\U_{n_{0}}\cup\ldots\cup \U_{n}\rightarrow \U^{\K_{n_{0}}}\cup_{\#}\ldots\cup_{\#} \U^{\K_{n}}
\ee
defined by,
\be
\phi(p)=\phi_{j}(p),\quad {\rm if\ } p\in \U_{j}
\ee
to translate it back to a Kasner data on $\mathcal{U}_{n_{0}}\cup\ldots \cup\mathcal{U}_{n}$. 
It is important to keep in mind in the following that, for each $n$, the metrics $\phi^{*}\hat{\hg}^{\K_{n}}$ are defined indeed on $\mathcal{U}_{n_{0}}\cup\ldots \cup\mathcal{U}_{n}$. The point now is that, as was mentioned earlier, one can take a convergent subsequence of the metrics $\hat{\hg}^{\K_{n}}$, and that will define the Kasner metric we were looking for. We pass to explain the calculations justifying the convergence.

We observe the following inequalities for any $n_{0}<i<n$,
\be\label{OBSI1}
\|\hat{\hg}^{\mathbb{K}_{i}}-\hat{\hg}^{\mathbb{K}_{n}}\|_{C^{j^{*}}_{\hat{\hg}^{\mathbb{K}_{i}}}(\mathcal{U}^{\K_{i}})} \leq 
\|\hat{\hg}^{\mathbb{K}_{i}}-\hat{\hg}^{\mathbb{K}_{i+1}}\|_{C^{j^{*}}_{\hat{\hg}^{\mathbb{K}_{i}}}(\mathcal{U}^{\K_{i}})}
+\|\hat{\hg}^{\mathbb{K}_{i+1}}-\hat{\hg}^{\mathbb{K}_{n}}\|_{C^{j^{*}}_{\hat{\hg}^{\mathbb{K}_{i}}}(\mathcal{U}^{\K_{i}})}
\ee
\be\label{OBSI2}
\|\hat{\hg}^{\mathbb{K}_{i}}-\hat{\hg}^{\mathbb{K}_{i+1}}\|_{C^{j^{*}}_{\hat{\hg}^{\mathbb{K}_{i}}}(\mathcal{U}^{\K_{i}})}\leq 
c_{0}\|\hat{\hg}^{\mathbb{K}_{i}}-\hat{\hg}^{\mathbb{K}_{i+1}}\|_{C^{j^{*}}_{\hat{\hg}^{\mathbb{K}_{i}}}(\mathcal{U}^{\K_{i}}\cap_{\#} \mathcal{U}^{\K_{i+1}})}
\ee
\be\label{OBSI3}
\|\hat{\hg}^{\mathbb{K}_{i}}-\hat{\hg}^{\mathbb{K}_{i+1}}\|_{C^{j^{*}}_{\hat{\hg}^{\mathbb{K}_{i}}}(\mathcal{U}^{\K_{i}}\cap_{\#} \mathcal{U}^{\K_{i+1}})}\leq \frac{c_{1}}{2^{im^{*}}}
\ee

\be\label{OBSI4}
\|\hat{\hg}^{\mathbb{K}_{i+1}}-\hat{\hg}^{\mathbb{K}_{n}}\|_{C^{j^{*}}_{\hat{\hg}^{\mathbb{K}_{i}}}(\mathcal{U}^{\K_{i}})}\leq c_{2}\|\hat{\hg}^{\mathbb{K}_{i+1}}-\hat{\hg}^{\mathbb{K}_{n}}\|_{C^{j^{*}}_{\hat{\hg}^{\mathbb{K}_{i+1}}}(\mathcal{U}^{\K_{i+1}})}
\ee
Briefly: The inequality (\ref{OBSI1}) is the triangle inequality. The inequality (\ref{OBSI3}) follows by first adding and subtracting $\phi_{i*}\hg$ inside the norm, then use the triangle inequality and finally use (\ref{IINII}) after noting that scaling the metric by $(\mathfrak{r}_{n})^{2}$ decreases the norm (observe that (\ref{IINII} involves scaled metrics). The inequalities (\ref{OBSI2}) and (\ref{OBSI4}) follow by noting that, because Kasner metrics are determined by an ODE, the norms on the right or the left hand sides are controlled by the $C^{j}$-norms at just one level set of $U$ i.e. just one torus. The constant $c_{0}$ and $c_{2}$ are independent on $n$ and $m^{*}$, and $c_{1}$ is independent on $n$ but may depend on $j^{*}$ and $m^{*}$. 

Putting all together we have the following recursive inequality,
\be
\|\hat{\hg}^{\mathbb{K}_{i}}-\hat{\hg}^{\mathbb{K}_{n}}\|_{C^{j^{*}}_{\hat{\hg}^{\mathbb{K}_{i}}}(\mathcal{U}^{\K_{i}})} \leq 
\frac{c_{4}}{2^{im^{*}}}
+c_{2}\|\hat{\hg}^{\mathbb{K}_{i+1}}-\hat{\hg}^{\mathbb{K}_{n}}\|_{C^{j^{*}}_{\hat{\hg}^{\mathbb{K}_{i+1}}}(\mathcal{U}^{\K_{i+1}})}
\ee
from which we deduce,
\begin{align}
\|\hat{\hg}^{\mathbb{K}_{i}}-\hat{\hg}^{\mathbb{K}_{n}}\|_{C^{j^{*}}_{\hat{\hg}^{\mathbb{K}_{i}}}(\mathcal{U}^{\K_{i}})} & \leq \frac{c_{4}}{2^{im^{*}}}+\frac{c_{4}c_{2}}{2^{(i+1)m^{*}}}+\frac{c_{4}c_{2}^{2}}{2^{(i+2)m^{*}}}+\ldots+\frac{c_{4}c_{2}^{n-i}}{2^{(i+n-i)m^{*}}}\\
& = \frac{c_{4}}{2^{im^{*}}}\sum_{l=0}^{l=n-i}\bigg(\frac{c_{2}}{2^{m^{*}}}\bigg)^{l}
\end{align}
Playing with the fact that $c_{2}$ does not depend on $m^{*}$, we take $m$ such that $2^{m^{*}}>c_{2}$, thus making the series $\sum_{l=0}^{l=\infty}\big(\frac{c_{2}}{2^{m^{*}}}\big)^{l}$ convergent. The work is essentially done. Using this bound we let $n\rightarrow \infty$ and we can take a subsequence of the metrics $\hat{\hg}^{\K_{n}}$ convergent in $C^{j^{*}-1}_{\hat{\hg}^{\K_{i}}}(\U^{\K_{i}})$ for every $i>n_{0}$. 

Say the limit is $\hat{\hg}^{\K_{\infty}}$. Then we have the bounds,
\be
\|\hat{\hg}^{\mathbb{K}_{i}}-\hat{\hg}^{\mathbb{K}_{\infty}}\|_{C^{j^{*}-1}_{\hat{\hg}^{\mathbb{K}_{i}}}(\mathcal{U}^{\K_{i}})}\leq \frac{c_{5}}{2^{im^{*}}}
\ee
for every $i>n_{0}$. Making $2=a_{*}^{\ln 2/\ln a_{*}}$ we get $2^{im^{*}}=a_{*}^{i(m^{*}\ln 2/\ln a_{*})}$ and recalling that $\mathfrak{r}_{n}=r(p_{n})=a_{*}^{n}r(p_{0})$ we get without difficulty,
\be
\|\hg-\phi^{*}\hat{\hg}^{\K_{\infty}}\|_{C^{j^{*}}_{\hg}}(p)\leq \frac{c_{m^{*},j^{*}}}{(\dist_{\hg}(p,\partial \Sigma))^{m^{*}(\ln 2/\ln a^{*})}}
\ee
Playing with the freedom in $j^{*}$ and $m^{*}$ and passing back to the variables $(g,N)$, KA is obtained as wished.
\end{proof}

\subsubsection{The asymptotic of free $\Sa$-symmetric data sets}\label{FTKASS}

Free $\Sa$-symmetric ends have a well defined limit of $U$ at infinity that we denoted by $U_{\infty}$ (Proposition \ref{LPRO}). In this section we study free $\Sa$-symmetric ends with the property that, 
\be\label{MAXU}
U(p)\leq U_{\infty} 
\ee
for all $p$. We aim to prove the following theorem. 
\begin{Theorem}\label{SSKAA} Let $(\Sigma;\hg,U)$ be a free $\Sa$-symmetric static end such that $U(p)\leq U_{\infty}$ for all $p\in \Sigma$. Then, either the data set is flat and $U$ is constant, or is asymptotic to a Kasner different from $A$ and $C$.
\end{Theorem} 

Suppose $(\Sigma;\hg,U)$ is a data set as in the last proposition. If $U(p)=U_{\infty}$ at some $p\in \Sigma^{\circ}$ then $U$ is constant by the maximum principle and the data set is flat. Due to this, from now on we are concerned with the case when $U<U_{\infty}$.

A large part of the proof of Theorem \ref{SSKAA} is indeed quite general and is valid too for a class of data sets that will show up again crucially in the next section. They are the $\hgls$-static ends that we define below (the `$\hgls$" is just a notation). 

The level sets of $U$ will be denoted as follows,
\be
U^{-1}_{*}=\{p\in \Sigma: U(p)=U_{*}\}
\ee
For instance $U_{1}^{-1}=\{p\in \Sigma:U(p)=U_{1}\}$ and so forth. As for the critical and regular values of $U$, it follows from Theorem 1 in \cite{MR0308345} that the set of critical values of $U$ is discrete. We will use this information below. Besides of this, the critical set $\{|\nabla U|=0\}$ is well understood but this won't be necessary here, (see \cite{2015arXiv150404563A}).  

\begin{Definition}[$\star$-static end] \label{DEFSSS} Let $(\Sigma;\hg,U)$ be a (non-necessarily free $\Sa$-symmetric) static end. Then, we say that $(\Sigma;\hg,U)$ is a $\hgls$-static end iff
\begin{enumerate}
\item\label{URVA1} the limit of $U$ at infinity exists (denote it by $U_{\infty}\leq \infty$),
\item\label{URVA2} $U<U_{\infty}$ everywhere,
\item\label{URVA3} there is a regular value $U_{0}$ of $U$, with $U_{0}> \sup \{U(p):p\in \partial S\}$, such that for any regular value $U_{1}\geq U_{0}$, $U^{-1}_{1}$ is a compact and connected surface of genus greater than zero.
\end{enumerate}
\end{Definition} 

Note that condition \ref{URVA2} implies that $\hgls$-ends are non-flat. It is also easy to see that any two regular values $U_{2}>U_{1}$ greater or equal than $U_{0}$, enclose a compact region $\Omega_{12}$, that is $\partial \Omega_{12}=U_{1}^{-1}\cup U_{2}^{-1}$.

The proof of Theorem \ref{SSKAA} follows from the next three propositions.
\begin{Proposition}\label{BATAT1} Let $(\Sigma;\hg,U)$ be a free $\Sa$-symmetric static end such that $U(p)< U_{\infty}$ for all $p$. Then $(\Sigma;\hg,U)$ is a $\hgls$-static end and has a simple cut $\{\mathcal{S}_{j}\}$.
\end{Proposition} 
\begin{Proposition}\label{BATAT2} Let $(\Sigma;\hg,U)$ be a static free $\Sa$-symmetric end such that $U(p)< U_{\infty}$ for all $p$. Then the end is asymptotic to a Kasner different from $A$ and $C$, or has sub-quadratic curvature decay.
\end{Proposition}
\begin{Proposition}\label{CORONA} Let $(\Sigma; \hg,U)$ be a $\hgls$-static end and let $\gamma$ be a ray. Suppose that the data set has a simple cut $\{\mathcal{S}_{i}\}$. Then the curvature does not decay sub-quadratically along $\gamma\cup (\cup_{j}\mathcal{S}_{j})$. 
\end{Proposition}
\begin{proof}[Proof of Theorem \ref{SSKAA}] Direct from Propositions \ref{BATAT1}, \ref{BATAT2} and \ref{CORONA}.
\end{proof}

Propositions \ref{BATAT1} and \ref{BATAT2} concern only free $\Sa$-symmetric ends  and are simple to prove.   

\begin{proof}[Proof of Proposition \ref{BATAT1}] We need to show only \ref{URVA2} of Definition \ref{DEFSSS}, items \ref{URVA1} and \ref{URVA3} are verified by hypothesis. Without loss of generality we can assume that the quotient manifold $S$ is diffeomorphic to $\Sa\times [0,\infty)$ (Propositions \ref{SUPO}, \ref{SUPO3}). We work on $(S;q,U,V)$ in particular we think $U$ as a function from $S$ into $\mathbb{R}$. Clearly there is a regular value $U_{0}$ such that for any regular value $U_{1}\geq U_{0}$, $U_{1}^{-1}$ is compact, that is, a collection of circles. None of such circles can be contractible otherwise we would violate the maximum principle. But if there are two such circles, then they enclose a compact manifold (a finite cylinder) hence the maximum principle would be also violated. Therefore $U_{1}^{-1}$ is just diffeomorphic to $\Sa$. Now thinking $U$ as a function from $\Sigma$ to $\mathbb{R}$, we have that $U_{1}^{-1}$ is diffeomorphic to a torus, hence of genus greater than zero. The existence of a simple cut $\{\mathcal{S}_{i}\}$ was shown in Proposition \ref{SIMPLECUTUS}.
\end{proof}
 
\begin{proof}[Proof of Proposition \ref{BATAT2}] We work on $(S;q,U,V)$. Let $\mu:=\lim A(B(\partial S,r))/r^{2}$. If $\mu>0$ then $(S;q)$ is asymptotic to a two-dimensional cone. Hence $\kappa$ decays sub-quadratically and therefore so does $|\nabla U|^{2}$ by (\ref{KAPPAF}). Suppose now that $\mu=0$. Let $\gamma$ be a ray from $\partial S$. If $\mu=0$ then any sequence of annuli $(\mathcal{A}^{c}_{r_{i}}(p_{i};1/2,2);q_{r_{i}})$, with $p_{i}\in \gamma$, metrically collapses to the segment $[1/2,2]$. For this reason, if $|\nabla U|^{2}$ decays sub-quadratically along any sequence $p_{i}\in \gamma$ then indeed $|\nabla U|^{2}$ decays sub-quadratically along the end. On the other hand if for a certain sequence $p_{i}$, $|\nabla U|_{r_{i}}^{2}(p_{i})\geq \rho_{*}>0$ ($\rho_{*}$ a given constant), then the end $(\Sigma;\hg,U)$ is indeed asymptotic to a Kasner different from $A$ and $C$ by Proposition \ref{KASYMPTOTIC}. (There is a caveat here. Proposition \ref{KASYMPTOTIC} requires that for $i$ large enough, the annulus $(\mathcal{A}_{r_{i}}(p_{i};1/2,2);\hg_{r_{i}})$ (annulus in $\Sigma$) to be metrically close to the segment $[1/2,2]$. For $i$ large enough the annulus $(\mathcal{A}_{r_{i}}(p_{i};1/2,2);q_{r_{i}})$ (annulus in $S$) is close to the segment $[1/2,2]$, then, if necessary, just make a scaling as in (\ref{SECSCA}), with $\lambda_{i}=1, \mu_{i}=0$ and with $\nu_{i}$ small enough that also the annulus $(\mathcal{A}_{r_{i}}(p_{i};1/2,2);\hg_{r_{i}})$ is close to $[1/2,2]$. Note that such scaling only changes the $\hg$-length of the $\Sa$-fibers in $\Sigma$ and so doesn't affect the norm $|\nabla U|^{2}$).
\end{proof} 
 
The proof of Proposition \ref{CORONA} will be carried out through several steps (Proposition \ref{ANTES}, \ref{GPI}, \ref{P70}, Corollary \ref{P75}, and Proposition \ref{FCOK}).

\begin{Proposition}\label{ANTES} Let $(\Sigma; \hg, U)$ be a $\hgls$-static end. Let $U_{0}$ be a regular value as in Definition \ref{DEFSSS} and consider another regular value $U_{1}\geq U_{0}$. Then, the set of points in $U_{0}^{-1}$ reaching $U_{1}^{-1}$ in time $U_{1}-U_{0}$ under the flow of $\partial_{U}=\nabla^{i}U/|\nabla U|$ is a set of total measure on $U_{0}^{-1}$ and its image under the flow is a set of total measure in $U_{1}^{-1}$.
\end{Proposition}

\begin{proof} Denote by $\Omega_{01}$ the manifold enclosed by $U_{0}^{-1}$ and $U_{1}^{-1}$. Let $\mathcal{C}=\{p:\nabla U(p)=0\}\cap \Omega_{01}$ be the set of critical points in $\Omega_{01}^{\circ}$. The closed set of points $C$ (note the font) in $U_{0}^{-1}$ that do not reach $U_{1}^{-1}$ in time $U_{1}-U_{0}$ under the flow of $\partial_{U}=\nabla^{i} U/|\nabla U|^{2}$, end in a smaller time at a point in $\mathcal{C}$. Let $\phi(x,t):C\times [0,\infty)\rightarrow \Omega_{01}$ be the map generated by the flow of the vector field $\nabla^{i} U$, (not the collinear field $\partial_{U}$), that is, that takes a point $x$ in $C$ and moves it a time $t$ by the flow of $\nabla^{i} U$ (note that indeed if $x\in C$, then the orbit under the flow of $\nabla^{i} U$ remains in $\Omega_{01}$ and is defined for all time). Suppose that the area of $C$ is positive. Then the set
\be
C_{1}=\{\phi(x,t):x\in C,0\leq t\leq 1\}
\ee
has positive volume $V(C_{1})$. But as $U$ is harmonic the flow of $\nabla^{i} U$ preserves volume and so we have $V(\phi(C_{1},t))=V(\phi(C_{1},0))$ for all $t\geq 0$. Let $\epsilon>0$ be small enough that
\be
V(B(\mathcal{C},\epsilon)\setminus \mathcal{C})<V(C_{1})/2
\ee
where $B(\mathcal{C},\epsilon)$ is the ball of points at a distance less than epsilon from $\mathcal{C}$. Then a contradiction is reached by choosing $t$ large enough that $\phi(C_{1},t)\subset B(\mathcal{C},\epsilon)\setminus \mathcal{C}$ because then it would be
\be
V(C_{1})=V(\phi(C_{1},t))\leq V(B(\mathcal{C},\epsilon)\setminus \mathcal{C})<V(C_{1})/2
\ee
To show that the image of $U^{-1}_{0}\setminus C$ under the flow of $\partial_{U}$ is a set of total measure in $U^{-1}_{1}$ just reverse the argument using the flow of $-\partial_{U}$ from $U^{-1}_{1}$ to $U^{-1}_{0}$.
\end{proof}

The following function of the level sets of $U$, ($U\geq U_{0}$), will be central in the analysis later,
\be\label{GMONOTONIC}
G(U):=\int_{U^{-1}}|\nabla U|^{2}dA
\ee
The function $G(U)$ is well defined at least for regular values of $U$. It is also well defined at the critical values but this won't be needed. As mentioned before Definition \ref{DEFSSS}, critical values of $U$ are discrete and, as we will show next, the lateral limits of $G(U)$ at any critical value $U_{c}$ coincide (and are finite). Let us see this property. Let $U_{2}>U_{1}$ be any two regular values with $U_{2}>U_{c}>U_{1}\geq U_{0}$ and let $\Omega_{12}$ be the region enclosed them. As in Proposition \ref{ANTES} let $C$ be the closed set of points in $U_{1}^{-1}$ that do not reach $U_{2}^{-1}$ in time $U_{2}-U_{1}$ under the flow of $\partial_{U}$. For any $\epsilon>0$ small enough let $R(\epsilon)$ be an open region in $U_{1}^{-1}$, with smooth boundary, containing $C$, and inside the ball $B(C,\epsilon)$. Let $C_{1}(\epsilon)=U_{1}^{-1}\setminus R(\epsilon)$. Let $\Omega_{12}(\epsilon)$ be the union of the set of integral curves (inside $\Omega_{12}$) of $\partial_{U}$ starting from points in $C_{1}(\epsilon)$ and ending in $U_{2}^{-1}$, and let $C_{2}(\epsilon)$ be the union of the end-points in $U_{2}^{-1}$ of these integral curves. Then the divergence theorem gives
\be
\int_{C_{1}(\epsilon)}|\nabla U|^{2}dA-\int_{C_{2}(\epsilon)}|\nabla U|^{2}dA=\int_{\Omega_{12}(\epsilon)}\langle\nabla\nabla U,\frac{\nabla U}{|\nabla U|} \nabla U\rangle dV
\ee
Take the limit $\epsilon\rightarrow 0$ and use Proposition \ref{ANTES} to deduce,
\be
G(U_{2})-G(U_{1})=\int_{\Omega'_{12}}\langle\nabla\nabla U,\frac{\nabla U}{|\nabla U|} \nabla U\rangle dV
\ee
where $\Omega'_{12}$ is the union of the set of integral curves of $\partial_{U}$ starting from points in $U_{1}^{-1}\setminus C$ and ending in $U_{2}^{-1}$ and is equal to $\Omega'_{12}$ minus a set of measure zero. Observe that the integrand is bounded. Take finally the limit $U_{1}\uparrow U_{c}$ and $U_{2}\downarrow U_{c}$ and note that the volume of $\Omega_{12}$ tends to zero (this is easy to see) to get
\be
\lim_{U_{1}\uparrow U_{c}}G(U)=\lim_{U_{2}\downarrow U_{c}}G(U)
\ee
as claimed.

The function $G(U)$ will be thought as defined for all $U\geq U_{0}$, continuous everywhere and differentiable except perhaps on a discrete set (the critical values of $U$). The continuity will be used implicitly several times in what follows.

\begin{Proposition}\label{GPI} Let $(\Sigma; \hg, U)$ be a $\hgls$-static end. Let $U_{0}$ be a regular value as in Definition \ref{DEFSSS}. Then for any two regular values $U_{2}>U_{1}\geq U_{0}$ we have,
\be
G'(U_{2})\geq G'(U_{1})
\ee
where $G'=dG/dU$.
\end{Proposition}

\begin{proof} Let $U_{*}$ be a regular value. Identify nearby level sets $U^{-1}$ to $U^{-1}_{*}$ through the flow of $\partial_{U}:=\nabla^{i} U/|\nabla U|^{2}=n/|\nabla U|$ where $n$ is the unit normal to $U^{-1}$. As $U$ is harmonic, the form $|\nabla U|dA$ is preserved. Abusing notation we write $|\nabla U|dA=|\nabla U_{*}|dA_{*}$. Thus, 
\be
G(U)=\int_{U^{-1}} |\nabla U||\nabla U_{*}|dA_{*}
\ee
Therefore 
\be
G'(U)=\int_{U^{-1}} (\nabla_{n}|\nabla U|)\frac{|\nabla U_{*}|}{|\nabla U|}dA_{*}=\int_{U^{-1}} \nabla_{n}|\nabla U| dA
\ee
Let $\Omega_{12}$ be the region enclosed by $U^{-1}_{1}$ and $U^{-1}_{2}$. Now let $\epsilon^{2}>0$ be a regular value of $|\nabla U|^{2}$ smaller than the minimum of $|\nabla U|^{2}$ over $U^{-1}_{1}$ and $U^{-1}_{2}$. Let $E=\{p\in \Omega_{12}:|\nabla U|(p)\leq \epsilon\}$. The divergence theorem gives us
\be
\int_{U^{-1}_{2}}\nabla_{n}|\nabla U|dA=\int_{U^{-1}_{1}}\nabla_{n}|\nabla U|dA+\int_{\Omega_{12}\setminus E^{\circ}}\Delta |\nabla U| dV+\int_{\partial E}\nabla_{n_{out}} |\nabla U|dA
\ee
The last term on the right hand side is positive, and the second from last is non-negative because $\Delta |\nabla U|\geq 0$ (use Bochner or just see \cite{MR1809792} Lemma 3.5). The proposition follows.
\end{proof}

\begin{Proposition}\label{P70} Let $(\Sigma; \hg,U)$ be a $\hgls$-static end. Let $U_{0}$ be a regular value as in Definition \ref{DEFSSS}. Then, for any two regular values $U_{2}\geq U_{1}\geq U_{0}$, we have
\be\label{LNGP}
\bigg(\frac{G'}{G}\bigg)(U_{2})\geq \bigg(\frac{G'}{G}\bigg)(U_{1})
\ee
where $G'=dG/dU$.
\end{Proposition}

\begin{proof} First, recall that the set of critical values of $U$ is discrete. We start proving that for any two regular values $U_{2}>U_{1}$ with no critical value in between, the inequality (\ref{LNGP}) holds.

We write
\be
\hg=\frac{1}{|\nabla U|^{2}}dU^{2}+h
\ee
where $h$ is a two-metric over the leaves $U^{-1}$ between $U_{1}^{-1}$ and $U^{-1}_{2}$. Denote with a prime ($'$) the derivative with respect to $\partial_{U}=\nabla^{i} U/|\nabla U|^{2}$. We will use again the notation $\lambda:=1/|\nabla U|$. Let $\Theta$ and $\theta$ be the second fundamental form and mean curvature respectively of the leaves $U^{-1}$. 

Fix a leaf $U_{*}^{-1}$. Identify the leaves $U^{-1}$ to $U^{-1}_{*}$ through the flow of $\partial_{U}$. As $U$ is harmonic we have $|\nabla U|dA=|\nabla U_{*}|dA_{*}$.  Hence 
\be
G=\int_{U^{-1}}|\nabla U|^{2}dA=\int_{U^{-1}}\frac{1}{\lambda}|\nabla U_{*}|dA_{*}. 
\ee
As $dA=\lambda |\nabla U_{*}|dA_{*}$ and $\theta =(\partial_{n} dA)/dA$ we deduce $\theta=-(1/\lambda)'$. Thus,
\be
G'=-\int_{U^{-1}}\theta |\nabla U_{*}|dA_{*}
\ee

\be
G''=-\int_{U^{-1}}\theta '|\nabla U_{*}|dA_{*}=-\int_{U^{-1}}\frac{\theta'}{\lambda}dA
\ee
We use now that in dimension three $\theta'$ has the standard expression,
\be
\theta'=-\Delta \lambda -(-2\kappa+tr_{h}Ric+\theta^{2})\lambda
\ee
to deduce,
\be
G''=-4\pi\chi+\int_{U^{1}}\big(\frac{|\nabla \lambda|^{2}}{\lambda^{2}}+tr_{h}Ric\big)dA+\int_{U^{-1}}\theta^{2}dA
\ee
where $\chi$ is the Euler characteristic of the leaves $U^{-1}$. On the right hand side of this expression the first two terms are non-negative. For the last term we have
\be
\int_{U^{-1}}\theta^{2}dA=\int_{U^{-1}}\theta^{2}\lambda |\nabla U_{*}|dA_{*}\geq \frac{\bigg(\int_{U^{-1}}\theta |\nabla U_{*}|dA_{*}\bigg)^{2}}{\int_{U^{-1}} \frac{1}{\lambda}|\nabla U_{*}|dA_{*}}=\frac{G'^{2}}{G}
\ee
Therefore,
\be
G''\geq \frac{G'^{2}}{G}
\ee
which is equivalent to $(G'/G)'\geq 0$ from which (\ref{LNGP}) follows.

We prove now that (\ref{LNGP}) also holds when $U_{2}>U_{1}$ are two regular values, and between them there is only one critical value $U_{c}$. This would complete the proof of the proposition. To see this we just compute,
\begin{align}
\label{LNGP1} \bigg(\frac{G'}{G}\bigg)(U_{2}) & \geq \lim_{U\rightarrow U^{+}_{c}}\bigg(\frac{G'}{G}\bigg)(U)=\bigg(\frac{\lim_{U\rightarrow U^{+}_{c}} G'(U)}{G(U_{c})}\bigg)\\
\label{LNGP2} & \geq \bigg(\frac{\lim_{U\rightarrow U^{-}_{c}} G'(U)}{G(U_{c})}\bigg)=\lim_{U\rightarrow U^{-}_{c}}\bigg(\frac{G'}{G}\bigg)(U)\\
& \geq \bigg(\frac{G'}{G}\bigg)(U_{1})
\end{align}
where to pass from (\ref{LNGP1}) to (\ref{LNGP2}) we use Proposition \ref{GPI} (note $G(U)>0$ for all $U$).
\end{proof}

\begin{Corollary}\label{P75} Let $(\Sigma; \hg,U)$ be a $\hgls$-static end. Then, there is a divergent sequence of points $p_{i}$, and constants $C>0$ and $D>0$ such that
\be\label{TOTOT}
|\nabla e^{CU}|(p_{i})\geq D
\ee
\end{Corollary}

\begin{proof} From Proposition \ref{P70} we get
\be
G(U)\geq G(U_{0})e^{-C(U-U_{0})}
\ee
where $C=-G'(U_{0})/G(U_{0})$. If $C\leq 0$ then $G(U)\geq G(U_{0})$. But 
\be
G(U)=\int_{U^{-1}}|\nabla U||\nabla U_{0}|dA_{0}
\ee
which has a fixed integration measure $|\nabla U_{0}|dA_{0}$. It follows that there must be a divergent sequence of points $p_{i}$ for which $|\nabla U|(p_{i})$ is bounded away from zero (which is not the case). Thus $C>0$. In this case we have
\be
G(U)e^{CU}\geq G(U_{0})e^{CU_{0}}>0.
\ee
But as 
\be
G(U)e^{CU}=\int_{U^{-1}}\frac{1}{C}|\nabla e^{CU}||\nabla U_{0}| dA_{0}
\ee
again we conclude that there must be a divergent sequence of points $p_{i}$ and a constant $D>0$ for which (\ref{TOTOT}) holds.
\end{proof}

\begin{Proposition}\label{FCOK} Let $(\Sigma; \hg,U)$ be a $\hgls$-static end and let $\gamma$ be a ray. Suppose that the data set has a simple cut $\{\mathcal{S}_{i}\}$ and that the curvature decays sub-quadratically along $\gamma\cup (\cup_{j}\mathcal{S}_{j})$.  Then, for any constant $C>0$, $|\nabla e^{CU}|$ tends to zero at infinity.
\end{Proposition}

\begin{proof} Let $\gamma(s)$ be a ray from $\partial \Sigma$ and parametrised by arc-length $s$, (i.e. $\dist(\gamma(s),\partial \Sigma)=s$). As we have done before, we will use the notation $r(p)=\dist(p,\partial \Sigma)$, for $p\in \Sigma$. Thus $r(\gamma(s))=s$. 

As $|\nabla U|^{2}$ decays faster than quadratically along $\gamma$ we have,
\be
r|\nabla U|(r)\rightarrow 0\quad {\rm as}\quad r\rightarrow \infty,
\ee
where we have denoted $|\nabla U|(\gamma(r))$ by $|\nabla U|(r)$. Let $r_{0}$ be such that for all $r\geq r_{0}$ we have $|\nabla U|(r)\leq 1/(2Cr)$. Integrating we obtain
\be
|U(r)-U(r_{0})|\leq \frac{1}{2C}\ln \frac{r}{r_{0}}
\ee
where to simplify notation we made $U(r):=U(\gamma(r))$. Thus,
\be\label{LATERTOU}
e^{CU(r)}\leq c_{1}r^{1/2}
\ee
We will use this inequality below.

The ray $\gamma$ intersects $\mathcal{S}_{j}$ and $\mathcal{S}_{j+1}$. So let $\alpha_{j,j+1}$ be the segment of $\gamma$ intersecting $\mathcal{S}_{j}$ and $\mathcal{S}_{j+1}$ only at its end points. Let $r_{j}$ be the number such that $\gamma(r_{j})$ is the end point of $\alpha_{j,j+1}$ in $\mathcal{S}_{j}$. The connected set 
\be
Z_{j}=\mathcal{S}_{j}\cup \alpha_{j,j+1}\cup \mathcal{S}_{j+1}
\ee
is included inside $\mathcal{A}(2^{1+2j},2^{4+2j})$. So by Proposition \ref{MAXMINU} (with $Z=Z_{j}$) we deduce,
\be\label{ANTERIOR}
U(q)\leq \eta+U(\gamma(r_{j}))
\ee
for any $q$ in $\mathcal{S}_{j}\cup \mathcal{S}_{j+1}$, and where $\eta$ does not depend on $j$.

Let $\mathcal{U}_{j,j+1}$ be the compact manifold enclosed by $\mathcal{S}_{j}$ and $\mathcal{S}_{j+1}$. By the maximum principle, the maximum of $U$ on $\mathcal{U}_{j,j+1}$ takes place at a point, say $x_{j}$, in $\mathcal{S}_{j}\cup \mathcal{S}_{j+1}$. So,
\be
U(x)\leq U(x_{j}) 
\ee
for any $x\in \mathcal{U}_{j,j+1}$. Combining this with (\ref{ANTERIOR}), with $q=x_{j}$, we obtain,
\be\label{BEUS11}
e^{CU(x)}\leq c_{2}e^{CU(\gamma(r_{j}))}
\ee
for any $x\in \mathcal{U}_{j,j+1}$ and where the constant $c_{2}$ does not depend on $j$. 

Now, $\mathcal{S}_{j}$ is included in $\mathcal{A}(2^{1+2j},2^{2+2j})$ and so we have,
\be\label{BOUNDSTT}
r_{j}\leq 2^{2+2j}
\ee
which plugged in (\ref{LATERTOU}) gives
\be\label{BOUNDSTOU}
e^{CU(\gamma(r_{j}))}\leq c_{1}2^{1+j}
\ee
Combining this bound and (\ref{BEUS11}) we deduce
\be\label{BEUS1}
e^{CU(x)}\leq c_{4}2^{j}
\ee
for any $x\in \mathcal{U}_{j,j+1}$ and where $c_{4}$ does not depend on $j$. 

On the other hand we also have $\Delta |\nabla U|^{2}\geq 0$ and thus the maximum of $|\nabla U|^{2}$ over $\mathcal{U}_{j,j+1}$ is reached again at $\mathcal{S}_{j}\cup \mathcal{S}_{j+1}$. From this fact we conclude that for every point $x\in \mathcal{U}_{j,j+1}$ it must be,
\be\label{BEUS2}
|\nabla U|(x)\leq \max\{|\nabla U|(q):q\in \mathcal{S}_{j}\cup \mathcal{S}_{j+1}\}\leq \frac{c_{5}}{2^{2j}}
\ee  
where the constant $c_{5}$ does not depend on $j$ and where to obtain the last inequality it was used that $|\nabla U|(q)\leq K/r(q)$ (Anderson's estimate) and the bound $r(q)\geq 2^{1+2j}$ for any $q\in \mathcal{S}_{j}\cup \mathcal{S}_{j+1}$ because $\mathcal{S}_{j}\cup \mathcal{S}_{j+1}$ is included in $\mathcal{A}(2^{1+2j},2^{4+2j})$.

Let $p_{j}$ be any divergent sequence such that $p_{j}\in \mathcal{U}_{j,j+1}$ for each $j$. Then, using (\ref{BEUS1}) and (\ref{BEUS2}) we reach,
\be
|\nabla e^{CU}|(p_{j})=Ce^{CU(p_{j})}|\nabla U|(p_{j})\leq \frac{c_{6}}{2^{j}}
\ee
where $c_{6}$ does not depend on $j$. Thus $|\nabla e^{CU}|(p_{j})$ tends to zero as $j$ goes to infinity. As the sequence $p_{j}$ is arbitrary we have proved the proposition.
\end{proof}

\subsubsection{Proof of the KA of static black hole ends}\label{POKA}

In this section we aim to prove finally Theorem \ref{KAFR} stating that a static black hole data set with sub-cubic volume growth is indeed AK. 

{\it Terminology}. Let $\Sigma$ be the manifold of a static black hole data. An embedded connected surface $\mathcal{S}$ is {\it disconnecting} if $\Sigma\setminus \mathcal{S}$ has two connected components one of which contains $\partial \Sigma$ and the other infinity. The closure of the component of $\Sigma\setminus \mathcal{S}$ containing $\partial \Sigma$ is denoted by $\Omega(\partial \Sigma,\mathcal{S})$. For instance, the surfaces $\mathcal{S}_{j}$ of a simple cut are disconnecting.

For any disconnecting surface $S$ on a static black hole data we have,
\be\label{MAXEQ}
\max\{U(p):p\in \Omega(\partial \Sigma,S)\}=\max\{U(p):p\in S\}
\ee
by the maximum principle. We will use this simple fact in the proof of the next proposition.
 
\begin{Proposition}\label{DFNKA} Let $(\Sigma; \hg, U)$ be a static black hole end with sub-cubic volume growth. Let $\gamma$ be a ray and let $\{\mathcal{S}_{j}\}$ be a simple cut. Then the end is either asymptotically Kasner different from $A$ and $C$ or the curvature decays sub-quadratically along the set $\gamma\cup (\cup_{j} \mathcal{S}_{j})$.
\end{Proposition}
\begin{proof} Suppose that there is a sequence of points $p_{n}\in \gamma\cup (\cup_{j} \mathcal{S}_{j})$ such that for some $\rho_{*}>0$,
\be\label{UCONDITI}
|\nabla U|_{r_{n}}(p_{n})\geq \rho_{*} 
\ee
If a subsequence of the annuli $(\mathcal{A}^{c}_{r_{n}}(p_{n}, 1/2,2);\hg_{r_{n}})$ collapses to a segment then $\gamma$ must pass through the annuli $\mathcal{A}^{c}_{r_{n}}(p_{n}, 1/2,2)$ and the end must be asymptotically Kasner by Theorem \ref{KASYMPTOTIC}. If no subsequence of these annuli metrically collapses to a segment then one can find a subsequence (also indexed by $n$) and neighbourhoods $\mathcal{B}_{n}$ of $\mathcal{A}^{c}_{r_{n}}(p_{n}, 1/2,2)$ such that $(\mathcal{B}_{n};\hg_{r_{n}})$ collapses to a two-dimensional orbifold. Having this, by a diagonal argument, one can find a subsequence of it (also indexed by $n$) and neighbourhoods $\mathcal{B}_{k_{n}}$ of $\mathcal{A}^{c}_{r_{n}}(p_{n};1/2,2^{k_{n}})$, with $k_{n}\rightarrow \infty$, and collapsing to a two-dimensional orbifold $(S_{\infty};q_{\infty})$. As the collapse is along $\Sa$-fibers (hence defining asymptotically a symmetry), we obtain, in the limit, a well defined reduced data $(S;q,\bar{U},V)$ where $\overline{U}$ is obtained as the limit of $U_{n}:=U-U(p_{n})$. This data has $|\nabla \overline{U}|_{q}\not\neq 0$ by (\ref{UCONDITI}) and therefore is non flat. Moreover it has at least one end containing a limit, say $\overline{\gamma}$, of the ray $\gamma$. Let us denote that end by $S_{\overline{\gamma}}$. 

As observed in subsection \ref{RDSACL} the limit orbifold has only a finite number of conic points, therefore the basic structure of the asymptotic of the reduced data on the end $S_{\overline{\gamma}}$ is described by Propositions \ref{REDCUR}, \ref{SUPO}, \ref{SUPO2} and \ref{SUPO3}. Furthermore $\overline{U}$ has a limit value $\overline{U}_{\infty}\leq \infty$ at infinity by Proposition \ref{LPRO}.   

We claim that we must have $\overline{U}\leq \overline{U}_{\infty}$. Let us see this. Assume $U_{\infty}<\infty$ otherwise there is nothing to prove. Let $j_{n}$ be an integer such that $j_{n}\leq r_{n}=r(p_{n})\leq j_{n}+1$. As $\gamma$ intersects all the surfaces $\mathcal{S}_{j}$, then fixed an integer $k\geq 1$, the surfaces $\mathcal{S}_{j_{n}+k}$ `collapse into sets' in $S_{\overline{\gamma}}$ as $n\rightarrow \infty$. The bigger $k$ is, the farther away the sets `collapse'. As $\overline{U}\rightarrow \overline{U}_{\infty}$ over the end $S_{\overline{\gamma}}$ then one can find a sequence $k_{n}\rightarrow \infty$ such that $U_{n}$ converges to $\overline{U}_{\infty}$ (as $n\rightarrow \infty$) when restricted to the surfaces $S_{j_{n}+k_{n}}$. Then, by (\ref{MAXEQ}), we have    
\be\label{MAXEQ2}
\max\{U_{n}(p):p\in \Omega(\partial \Sigma,S_{j_{n}+k_{n}})\}=\max\{U_{n}(p):p\in S_{j_{n}+k_{n}}\}\rightarrow \overline{U}_{\infty}
\ee 
and the claim follows because if $\overline{U}(q)\geq \overline{U}_{\infty}+\epsilon$ for some $\epsilon>0$ and for some $q\in S_{\overline{\gamma}}$ then there is a sequence of points $q_{n}\in \Omega(\partial \Sigma,S_{j_{n}+k_{n}})$ with $U_{n}(q_{n})>\overline{U}_{\infty}+\epsilon/2$ if $n\geq n_{0}$, that would eventually violate (\ref{MAXEQ2}).  

As $(S_{\overline{\gamma}}; q,\overline{U},V)$ is non-flat then it has to be AK different from the Kasner $A$ and $C$ by Proposition \ref{SSKAA}. Therefore one can find a sequence $k_{n}$ such that the annuli 
\be
(\mathcal{A}^{c}(\gamma(r_{n}2^{k_{n}});r_{n}2^{k_{n}-1},r_{n}2^{k_{n}+1});\hg_{r_{n}2^{k_{n}}}),
\ee
neighbouring the points $\gamma(r_{n}2^{k_{n}})$, collapse to a segment $[1/2,2]$ while having 
\be
|\nabla U|_{\hg_{r_{n}2^{k_{n}}}}(\gamma(r_{n}2^{k_{n}}))\geq \rho^{**}
\ee
for some $\rho^{**}>0$. Then the end must be asymptotically Kasner by Theorem \ref{KASYMPTOTIC}. We reach thus a contradiction. Hence, the curvature decays sub-quadratically along the set $\gamma\cup (\cup_{j} \mathcal{S}_{j})$.
\end{proof}

\begin{Corollary}\label{PARAFINA} Let $(\Sigma; \hg, U)$ be a static black hole data set with sub-cubic volume growth that is not AK. Then 
\be\label{QPARAFINA}
(\max\{U(p):p\in \mathcal{S}_{j}\cup \mathcal{S}_{j+1}\}-\min\{U(p):p\in \mathcal{S}_{j}\cup \mathcal{S}_{j+1}\})\rightarrow 0
\ee  
where $\{\mathcal{S}_{j}\}$ is a simple cut.
\end{Corollary}

\begin{proof} If the data is not AK, then we deduce by Proposition \ref{DFNKA} that for any sequence of points $p_{j}\in \mathcal{S}_{j}$ we have 
\be\label{777}
|\nabla U|_{r_{j}}(p_{j})\rightarrow 0,
\ee
where $r_{j}=r(p_{j})$ as usual. Now, if $p_{j}\in \mathcal{S}_{j}$ then $2^{1+2j}\leq r_{j}\leq 2^{4+2j}$, thus (\ref{777}) implies right away that,
\be
\max\{|\nabla U|_{\hat{r}_{j}}(q):q\in \mathcal{S}_{j-1}\cup \mathcal{S}_{j+2}\}\rightarrow 0,
\ee
as $j\rightarrow \infty$, where we made $\hat{r}_{j}=2^{2j}$. Now, as the maximum of $|\nabla U|_{\hat{r}_{j}}$ on $\mathcal{U}_{j-1,j+2}$ is reached at $\mathcal{S}_{j-1}\cup \mathcal{S}_{j+2}$ we conclude that,
\be\label{CAPICCI}
\max\{|\nabla U|_{\hat{r}_{j}}(q):q\in \mathcal{U}_{j-1,j+2}\}\rightarrow 0
\ee
as $j\rightarrow \infty$. Observe that because $\mathcal{S}_{j}$ and $\mathcal{S}_{j+1}$ are intersected by any ray $\gamma$ ($\{\mathcal{S}_{j}\}$ is a simple cut), they belong to the same connected component of $\mathcal{A}(2^{1+2j},2^{4+2j}) =\mathcal{A}_{\hat{r}_{j}}(2,4)$. Denote that component by $\mathcal{A}_{\hat{r}_{j}}^{c}(2,4)$. We have,
\be
\mathcal{S}_{j}\cup \mathcal{S}_{j+1}\subset \mathcal{A}_{\hat{r}_{j}}^{c}(2,4)\subset \mathcal{A}_{\hat{r}_{j}}^{c}(1/2,2^{6})\subset \mathcal{U}_{j-1,j+2}
\ee
and remember that by (\ref{CAPICCI}) the maximum of $|\nabla U|_{\hat{r}_{j}}$ over $\mathcal{A}_{\hat{r}_{j}}^{c}(1/2,2^{6})$ tends to zero. So (\ref{QPARAFINA}) is exactly item \ref{posterior} in Proposition \ref{MAXMINU} with $a=2$, $b=4$ and $Z_{j}=\mathcal{S}_{j}\cup \mathcal{S}_{j+1}$.
\end{proof}

\begin{Proposition}\label{EXISTLIM} Let $(\Sigma; \hg, U)$ be a static black hole data set with sub-cubic volume growth. Then $U$ tends uniformly to a constant $U_{\infty}\leq \infty$ at infinity. 
\end{Proposition}

\begin{proof} The claim is obviously true if the end is AK. Let us assume then that the end is not AK. Let $\{\mathcal{S}_{j}\}$ be a simple cut and $\gamma$ a ray. By Corollary \ref{PARAFINA} we have,
\be
(\max\{U(p):p\in \mathcal{S}_{j}\cup \mathcal{S}_{j+1}\}-\min\{U(p):p\in \mathcal{S}_{j}\cup \mathcal{S}_{j+1}\})\rightarrow 0
\ee  
And by the maximum principle,
\begin{align}
\max\{U(p):&\ p\in \mathcal{S}_{j}\cup \mathcal{S}_{j+1}\} \geq \max\{U(p):p\in \mathcal{U}_{j,j+1}\}\geq \\ 
& \geq \min\{U(p):p\in \mathcal{U}_{j,j+1}\}\geq \min\{U(p):p\in \mathcal{S}_{j}\cup \mathcal{S}_{j+1}\}
\end{align}
Therefore the function $U$ is becoming more and more constant over the manifolds $\mathcal{U}_{j,j+1}$ enclosed by $\mathcal{S}_{j}$ and $\mathcal{S}_{j+1}$. A simple application of this fact is that if there is a sequence of manifolds $\mathcal{U}_{j_{i},j_{i}+1}$ over which $U$ tends to infinity then $U$ must tend to infinity over any other sequence $\mathcal{U}_{j'_{i},j'_{i}+1}$, as if not then for some $i_{1}<i_{2}$ the minimum of $U$ over the manifold $\mathcal{U}_{j_{i_{1}},j_{i_{2}}}$ enclosed by $\mathcal{S}_{j_{i_{1}}}$ and $\mathcal{S}_{j_{i_{2}}}$ would not be reached at a point on either $\mathcal{S}_{j_{i_{1}}}$ or $\mathcal{S}_{j_{i_{2}}}$, but rather at a point on a manifold $\mathcal{U}_{j,j+1}$ with $j_{i_{1}}<j$ and $j+1<j_{i_{2}}$. This would violate the maximum principle. For the same reason if $U$ tends to a finite constant over a sequence of manifolds $\mathcal{U}_{j,j+1}$ then it must tend also to the same constant over any other sequence.
\end{proof}

{\it Notation}: Let $\Sigma$ and ${\rm H}$ be two manifolds with diffeomorphic compact boundaries and let $\#:\partial {\rm H}\rightarrow \partial \Sigma$ be a diffeomorphism. Then we denote by $\Sigma\# {\rm H}$ the manifold that is obtained by identifying $\partial \Sigma$ to $\partial {\rm H}$ through $\#$. 

We know from Part I that, if a static black hole data set $(\Sigma; g,N)$ is not the Boost, then every horizon component is a two-sphere. In that case, to every connected component of $\partial \Sigma$ we can glue a three-ball (in a unique way). Following the notation above we denote by ${\rm H}$ the set of balls and by $\Sigma\# {\rm H}$ the resulting boundary-less manifold. We will use this notation below.

\begin{Proposition}\label{TOROROS} Let $(\Sigma; \hg, U)$ be a static black hole data set with sub-cubic volume growth. Then, either the data set is a Boost, is asymptotic to a Boost or there is a divergent sequence of disconnecting tori embedded in $\Sigma$ and enclosing solid tori in $\Sigma\# {\rm H}$.
\end{Proposition}

\begin{proof} Let us assume that the data set is not a Boost and in particular that the horizon components are spheres. Furthermore, let us assume that the data it is not asymptotic to $B$, a Boost.

By Galloway's \cite{MR1201655} (see comments in footnote \ref{FN} in section \ref{TSOTP} of Part I), if there is a divergent sequence of disconnecting tori $T_{i}$ having outwards mean curvature positive in $(\Sigma;g)$ (not in the space $(\Sigma;\hg)$), then they bound solid tori in $\Sigma\cup {\rm H}$ (see definition of ${\rm H}$ above). Let us prove below the existence of such tori under the mentioned assumptions. 

If the data is asymptotic to a Kasner space, hence different from $B$ by assumption, then the existence of disconnecting tori with outwards mean curvature positive is direct (see further comments in subsection \ref{TSOTP} of Part I). On the other hand if the asymptotic is not a Kasner space then by Proposition \ref{DFNKA} the curvature of $\hg$ must decay sub-quadratically along $\gamma\cup (\cup_{j} \mathcal{S}_{j})$ where $\gamma$ is a ray and $\{\mathcal{S}_{j}\}$ is a simple cut. So, let us assume furthermore this decay, and hence that the asymptotic is not Kasner different from $A$ and $C$.

Let $p_{j}$ be, for each $j$, a point in $\gamma\cap \mathcal{S}_{j}$. If for some subsequence $p_{j_{i}}$ the annuli $(\mathcal{A}^{c}_{r_{j_{i}}}(p_{j_{i}};1/2,2);\hg_{r_{j_{i}}})$ collapse to a segment $[1/2,2]$, then there are neighbourhoods $\mathcal{B}_{i}$ of $\mathcal{A}^{c}_{r_{i}}(p_{j_{i}};1/2,2)$ and finite coverings $\tilde{\mathcal{B}}_{i}$ such that the sequence $(\tilde{\mathcal{B}}_{i};\hg_{r_{j_{i}}})$ converges to a $\Sa\times \Sa$-symmetric flat space $([1/2,2]\times T;g_{F})$. The limit is flat due to the sub-quadratic curvature decay along the ray $\gamma$ that is crossing $\mathcal{B}_{i}$. For the same reason the lifts of $U-U(p_{j})$ to $\tilde{\mathcal{B}}_{i}$ converge to the constant zero. Hence the lifts of $N/N(p_{j})$ converge to the constant one.
Let $T_{i}$ be a sequence of embedded tori in $\mathcal{B}_{i}$ such that the coverings $\tilde{T}_{i}$ converge (in $C^{2}$) to the torus $\{1\}\times T$ on $[1/2,2]\times T$. Observe that, as the disconnecting surfaces $\mathcal{S}_{j_{i}}$ are embedded in $\mathcal{B}_{i}$, the tori $T_{i}$ are also disconnecting. If the outwards mean curvature of the torus $\{1\}\times T$ is negative, then so is the mean curvature of the tori $T_{i}$ for $i$ sufficiently large. But this is not possible because as $Ric\geq 0$ any ray from $T_{i}$ would develop a focal point at a finite distance from $T_{i}$. On the other hand if the outwards mean curvature is positive, then for $i\geq i_{0}$ with $i_{0}$ large enough the mean curvature of the tori $T_{i}$ calculated using $g$ is also positive because the lifts of $N/N(p_{j})$ converge to one (so $g$ and $\hg$ different essentially by a numeric factor). Thus, the $T_{i}$ are the tori we are looking for. 
So let us suppose that the mean curvature of the torus $\{1\}\times T$ is zero and that this occurs for every subsequence $p_{j_{i}}$ for which the annuli $(\mathcal{A}^{c}_{r_{j_{i}}}(p_{j_{i}};1/2,2);\hg_{r_{j_{i}}})$ collapse to the segment $[1/2,2]$. Note that in such case the $\Sa\times \Sa$-symmetric space $([1/2,2]\times T;g_{F})$ must be a flat metric product $g_{F}=dx^{2}+h_{F}$. In such hypothesis we claim that if there is one such sequence then the end must be diffeomorphic to $[0,\infty)\times T^{2}$. This intuitive fact was proved essentially in \cite{MR3302042}, so let us postpone explaining it until later. Now, if the end is diffeomorphic to $[0,\infty)\times T^{2}$, then, as proved in Proposition \ref{SSEU} below, it must be a $\star$-static end. Then, by Proposition \ref{CORONA}, the curvature cannot decay sub-quadratically along $\gamma\cup(\cup \mathcal{S}_{j})$ which is against the hypothesis. We reach thus a contradiction. 

So let us assume now that there are no sequence of points $p_{j_{i}}$ with the property that the annuli $(\mathcal{A}^{c}_{r_{j_{i}}}(p_{j};1/2,2);g_{r_{j_{i}}})$ collapse to a segment. Then, exactly as was done inside the proof of proposition \ref{DFNKA}, we can find a subsequence $p_{j_{i}}$ and a sequence of neighbourhoods $\mathcal{B}_{k_{i}}$ of $\mathcal{A}^{c}_{r_{j_{i}}}(p_{j_{i}};1/2,2^{k_{i}})$ over which the static data collapses to a reduced data $(S;q,\overline{U},V)$, having at least one end (without orbifold points), containing a limit of the ray $\gamma$. Furthermore, over that end we have $\overline{U}\leq \overline{U}_{\infty}\leq \infty$. Let $(\overline{\Sigma};\overline{\hg},\overline{U})$ be a static data reducing to $(E;q,\overline{U},V)$, found by taking a limit of covers (unwrappings) of (regions of) the manifolds $\mathcal{B}_{k_{i}}$. Then, by Proposition \ref{SSKAA}, either the end $(\overline{\Sigma};\overline{\hg},\overline{U})$ is asymptotic to a Kasner space different from $A$ and $C$, or it is flat and $\overline{U}$ is constant. If it is asymptotic to a Kasner space different from $A$ and $C$ then we one can easily find a sequence of points $p'_{i}$ in $\gamma\cap \mathcal{B}_{k_{i}}$ such that (make $r'_{i}=r(p'_{i})$) the annuli $(\mathcal{A}^{c}_{r'_{i}}(p'_{i};1/2,2);\hg_{r'_{i}})$ metrically collapses to the interval $[1/2,2]$ while $\rho(p'_{i})\rightarrow \rho^{*}>0$. It follows then from  Theorem \ref{KASYMPTOTIC} that the end is asymptotically Kasner different from $A$ and $C$ which is against the assumption made earlier. Suppose now that $(\overline{\Sigma};\overline{\hg})$ is flat. Again, we need to find a convex $\Sa$-symmetric torus on $\overline{\Sigma}$, from which we can obtain a sequence of convex disconnecting tori $T_{i}$ on the neighbourhoods $\mathcal{B}_{k_{i}}$. To prove this we will rely on the results obtained for reduced data. Indeed we know that $E$ is diffeomorphic to $\Sa\times [r_{0},\infty)$ and that, if the area growth is quadratic, then $q$ has the asymptotic form $dr^{2}+\mu^{2}r^{2}d\varphi^{2}$ with $\kappa$ and $|\nabla V|^{2}=|\nabla \ln \Lambda|^{2}$ decaying sub-quadratically. On the other hand $\overline{\hg}$ has de form,
\be
\overline{\hg}=q_{ij}dx^{i}dx^{j}+\Lambda^{2}(d\varphi+  \theta_{i}dx^{i})^{2}
\ee
where $(x_{1},x_{2})=(r,\theta)$. For $r_{0}$ sufficiently large, the mean curvature of the torus $\{r=r_{0}\}$ on $\overline{\Sigma}$ is approximated by $\partial_{r}\Lambda/\Lambda+1/r_{0}\sim 1/r_{0}$, hence positive. This provides the torus we are looking for. On the other hand if the area growth of $(E;q)$ is less than quadratic, then, for any divergent sequence of points $t_{i}$, the annuli 
$(\mathcal{A}(t_{i};1/2,2);q_{r_{i}})$ metrically collapse to $[1/2,2]$. Using this, one can easily find a sequence of points $p_{j_{l}}$ on $\gamma\cup \{\mathcal{S}_{j}\}$ such that the annuli $(\mathcal{A}^{c}_{r_{j_{l}}}(p_{j_{l}};1/2,2);\hg_{r_{j_{l}}})$ collapse to the segment $[1/2,2]$, reaching a contradiction as we are assuming such a sequence does not exist. 
\vs

To conclude let us explain the claim that was left to be proved. Assume then that for every subsequence $p_{j_{i}}$ (of the original sequence $p_{j}$) for which the annuli $(\mathcal{A}^{c}_{r_{j_{i}}}(p_{j_{i}};1/2,2);\hg_{r_{j_{i}}})$ collapse to the segment $[1/2,2]$ there are neighbourhoods $\mathcal{B}_{i}$ of $\mathcal{A}^{c}_{r_{i}}(p_{j_{i}};1/2,2)$ and finite coverings $\tilde{\mathcal{B}}_{i}$ such that the sequence $(\tilde{\mathcal{B}}_{i};\hg_{r_{j_{i}}})$ converges to a $\Sa\times \Sa$-symmetric flat space $([1/2,2]\times T;g_{F}=dx^{2}+h_{F}$) where $h_{F}$ is a flat metric on $T$. 

We recall first a fact from \cite{MR3302042}: for any $\delta>0$ sufficiently small and $a\leq 1/2$ and $b\geq 2$, there is $\epsilon(\delta)$, such that if $(\mathcal{A}(p_{j};a,b);\hg_{r_{j}})$ is $\epsilon(\delta)$-close in the Gromov-Hausdorff metric to the segment $[a,b]$, then there is a neighbourhood $\mathcal{B}_{j}$ of the annulus $\mathcal{A}(p_{j};a,b)$, diffeomorphic to $[a,b]\times T^{2}$, on which $\hg_{r_{j}}$ has the form,
\be
\hg_{r_{j}}=\alpha^{2}dx^{2}+h
\ee
and where the function $\alpha$ and the family of metrics $h(x)$ on $T^{2}$, satisfy,
\be\label{SUPERCONDITIONS}
1-\delta\leq \alpha\leq 1+\delta,\quad (1-\delta)h(x')\leq h(x)\leq (1+\delta)h(x')
\ee
Furthermore the Gromov-Hausdorff distance is controlled by the $h(x)$-diameter of $T^{2}$ (for any $x$), namely it is between $c_{1}\diam_{h(x)}(T^{2})$ and $c_{2}\diam_{h(x)}(T^{2})$ where $c_{1}$ and $c_{2}$ do not depend on $\epsilon$, $a$ or $b$.    

Using this fact, we fix $\delta$ small, let $k\geq 1$, and find $\epsilon(k)\leq \epsilon(\delta)$ small enough that if $(\mathcal{A}(p_{j};1/2,2);\hg_{r_{j}})$ is $\epsilon(k)$-close to $[1/2,2]$ then $(\mathcal{A}(p_{j};1/2,2^{k});\hg_{r_{j}})$ is $\epsilon(\delta)$-close to $[1/2,2^{k}]$. But then the relations (\ref{SUPERCONDITIONS}) hold for all $x\in [1/2,2^{k}]$, hence also for $x\in[1/2,2]$, and then the Gromov-Hausdorff distance from $(\mathcal{A}(p_{j};1/2,2^{k});\hg_{r_{j}})$ to $[1/2,2^{k}]$ is less or equal than $c_{3}\epsilon(k)$ where $c_{3}$ does not depend on $k$. Thus, if $k$ is big enough the Gromov-Hausdorff distance between $(\mathcal{A}(p_{j+k};1/2,2);\hg_{r_{j+k}})$ and $[1/2,2]$ is less than $\epsilon(k)/2$ (observe here that $r_{j+k}\geq 2^{2k}r_{j}$ and thus $\hg_{r_{j+k}}\leq 2^{-4k}\hg_{r_{j}}$). Therefore, if for some $j_{*}$ the GH-distance from $(\mathcal{A}(p_{j_{*}};1/2,2);\hg_{r_{j_{*}}})$ and the interval $[1/2,2]$ is less or equal than $\epsilon(k)$, then the same occurs for the annuli $(\mathcal{A}(p_{j_{*}+mk};1/2,2);\hg_{r_{j_{*}+mk}})$ for any $m\geq 1$. Then the end must then be diffeomorphic to $[0,\infty)\times T^{2}$.
\end{proof}

The following proposition completes the proof of the previous proposition. 

\begin{Proposition}\label{SSEU} Let $(\Sigma;\hg,U)$ be a static end with $\Sigma$ diffeomorphic to ${\rm T}^{2}\times [0,\infty)$. Suppose that the limit of $U$ at infinity exists and that $U<U_{\infty}(\leq \infty)$ everywhere. Then, $(\Sigma;\hg,U)$ is a $\star$-static end.
\end{Proposition} 

\begin{proof} Let $U_{0}$ be a regular value of $U$ sufficiently close to $U_{\infty}$ such that for any regular value $U_{1}>U_{0}$, $U^{-1}_{1}$ is a compact manifold without boundary embedded in $\Sigma^{\circ}$. Let ${\rm H}$ be a solid torus and consider $\Sigma\# {\rm H}$. Let us prove first that $U_{1}^{-1}$ is connected. Suppose $U_{1}^{-1}$ has components $S_{1},\ldots,S_{m}$, $m\geq 2$. Then, each $S_{k}$ encloses a bounded region $\Omega_{k}$ in $\Sigma \# {\rm H}$ (think $\Sigma \# {\rm H}$ as an open solid torus embedded in $\mathbb{R}^{3}$). If one of the $\Omega_{k}$ does not contain ${\rm H}$ then $U$ must be constant by the maximum principle which is against the hypothesis. Therefore if $S_{i}\neq S_{j}$ then either $\Omega_{i}\subset \Omega_{j}$ or $\Omega_{j}\subset \Omega_{i}$. Whatever the case, the surfaces $S_{i}$ and $S_{j}$ bound a compact region inside $\Sigma$ which again is impossible by the maximum principle. So $U_{1}^{-1}=S_{1}$, is connected. Finally, $U_{1}^{-1}$ cannot be a sphere because if so the region $\Omega$ enclosed by it must be a ball containing ${\rm H}$ but ${\rm H}$ is not contractible. 
\end{proof}
\begin{Proposition}\label{LALA} Let $(\Sigma;\hg,U)$ be a static black hole data set, asymptotic to a Boost $B$ but that is not a Boost. Then, $\Sigma$ is diffeomorphic to an open solid three-torus minus a finite number of open three-balls, and thus there is a divergent sequence of disconnecting tori $T_{i}$ enclosing solid tori in $\Sigma\# {\rm H}$.
\end{Proposition}
\begin{proof} Recall that a Boost has a data $\Sigma_{B}=[0,\infty)\times \T^{2}$, $g_{B}=dx^{2}+h$, $N_{B}=x$, where $h$ is a flat metric on $\T^{2}$. Following the Definition \ref{KADEF} of Kasner asymptotic, Let $\phi:\Sigma\setminus K\rightarrow \Sigma_{B}\setminus K^{\mathbb{K}}$ be a diffeomorphism into the image such that the components $(\phi_{*}g)_{ij}$ (and their derivatives) converge to the components $g_{B,ij}$ (and their derivatives, i.e. zero) faster than any inverse power of $x$. Denote by $T_{x}$ the tori $\phi^{-1}(\{x\}\times \T^{2})$ ($x\geq x_{0}$ such that $\{x\}\times \T^{2}\subset \Sigma^{\mathbb{K}}\setminus K^{\mathbb{K}}$). Note that by the fast decay, the Gaussian curvature and the second fundamental forms of the tori $T_{x}$ tend to zero faster than any inverse power of $x$ as $x\rightarrow \infty$. Let $A(T_{x})$ be the area of $T_{x}$ and $A(T_{\infty})=\lim_{x\rightarrow \infty} A(T_{x})$.

The key point to prove the theorem is to to show that there is a torus, say $T_{*}$, isotopic to the tori $T_{x}$ and with area less than $A(T_{\infty})$. If this is the case, then one can essentially use Galloway's arguments in \cite{MR1201655} to conclude that indeed the tori $T_{x}$ enclose solid tori in $\Sigma\# {\rm H}$ (see footnote \ref{FN} in Part I). Shortly, the argument would be as follows. Let $\Sigma_{x}$ be the closure of the connected component of $\Sigma\setminus T_{x}$ containing $\partial \Sigma$. Let $x_{1}$ be large enough that for any $x\geq x_{1}$, the region near $T_{x}$ is so close to flat that one can extend $\Sigma_{x}$ by a small (Riemannian) ring (diffeomorphic to $[0,1]\times \T^{2}$), in such a way that the new boundary has positive outwards mean curvature, and furthermore that, if a stable minimal surface intersects the ring then it has area greater than $A(T^{*})$. Granted this, there is always a sequence of tori $S_{i}$ isotopic to $T_{x}$, disjoint from the ring and minimising area within the class of tori isotopic to $T_{x}$. One can repeat Galloway's argument directly. 

Let us show now the existence of $T^{*}$. It will follow from proving that there is an suitable integrable congruence of geodesics $\{\bar{\gamma}\}$ with respect to the optical metric $\bar{g}=N^{-2}g$ over the end of $\Sigma$. Integrable here means that the distribution of planes perpendicular to the geodesics integrate to surfaces, in this case two-tori. The congruence will cover $\Sigma$ outside a bounded closed set. Furthermore if we let $\bar{T}_{t}$ be the family of `integral' tori, where $t$ is the $\bar{g}$-distance between $\bar{T}_{t}$ and $\bar{T}_{0}$, then the Gaussian curvature and second fundamental form of the $\bar{T}_{t}$ tend to zero $t\rightarrow \infty$. Suppose that we have such a congruence. Let $\theta$ be, at every point, the $g$-mean curvature of the tori $\bar{T}_{t}$ passing through that point and with respect to the normal $n=-\partial_{t}/|\partial_{t}|_{\bar{g}}$ (`inwards'). Then, it was proved in \cite{MR1201655} (see also \cite{MR3077927}) that, the mean curvature $\theta$ evaluated on a geodesic $\bar{\gamma}$ is monotonically decreasing as $t$ decreases. As the mean curvature $\theta$ of the tori $\bar{T}_{t}$ tends to zero as $t\rightarrow \infty$, then $\theta\leq 0$ everywhere. As the areas of the tori $\bar{T}$ tend to $A(T_{\infty})$ then, at any $t$, $A(\bar{T}_{t})<A(T_{\infty})$ unless the mean curvature is identically zero in all the region between $\bar{T}_{t}$ and infinity. If such is the case, it also follows from \cite{MR1201655} that in that region the metric is a flat product, which is not the case because by hypothesis the data is not a Boost. So $A(\bar{T}_{t})<A(T_{\infty})$ and we define $T_{*}=\bar{T}_{t}$.   

The construction of the congruence of $\bar{g}$ geodesics is as follows. Consider the congruence of geodesics with respect to $\bar{g}$, emanating from $T_{x}$ and perpendicularly to it, and towards $\partial \Sigma$ (`inwards'). Due to the fast decay of $\phi_{*}g$ into $g_{\mathbb{K}}$, and of $\phi_{*}N$ to the function $x$ (indeed the fast decays of $\phi_{*}N-x$ to zero), the congruence converges as $x\rightarrow \infty$, to a (smooth) congruence and covering $\Sigma$ outside a bounded closed set as wished.
\end{proof} 

The following proposition, that uses the previous ones, essentially proves that static black hole ends are $\star$-static ends. 

\begin{Proposition}\label{REGVALUO} Let $(\Sigma; \hg, U)$ be a static black hole data set with sub-cubic volume growth. Then, there is a regular value $U_{0}<U_{\infty}$, such that for any regular value $U_{1}$ of $U$ with $U_{\infty}>U_{1}\geq U_{0}$, $U_{1}^{-1}$ is a compact connected surface of genus greater than zero.  
\end{Proposition}

\begin{proof} If the data is a Boost then we are done, so let us assume from now on that it is not. 

Let $\Omega$ be a an open connected set with a closure $\overline{\Omega}$ containing $\partial \Sigma$. Let $U_{0}<U_{\infty}$ be a regular value such that the set $\{U<U_{0}\}$ contains $\overline{\Omega}$. Suppose that for some regular value $U_{1}>U_{0}$, $U_{1}^{-1}$ has the connected components $S_{1},\ldots,S_{m}$, $m\geq 2$. If for one of the components, say $S_{i}$, $\Sigma\setminus S_{i}$ is connected, then we can glue two copies of $\Sigma\setminus S_{i}$ along $S_{i}$ to make a static black hole data set with more than one end which is not possible. So for every $S_{j}$, $\Sigma\setminus S_{j}$ has two connected components, and because $U_{1}>U_{0}$, one of the them must contain $\overline{\Omega}$. Call the closure of that component $\Sigma_{j}$. We have $\partial \Sigma_{j}=\partial \Sigma\cup S_{j}$. Observe now that $\Sigma\setminus \Sigma_{j}^{\circ}$ must be connected because first, no component of it can be compact (that would violate the maximum principle) and second, no two of the components can be non-compact (because there would be at least two ends). Hence, if $\Sigma\setminus \Sigma_{j}^{\circ}$ is non-compact then $\Sigma_{j}$ must be compact (if not there would be two ends again). In sum, every $S_{j}$ is disconnecting, and $\partial \Sigma_{j}=\partial \Sigma \cup S_{j}$. Therefore if $m\geq 2$ then, either $\Sigma_{1}\setminus \Sigma_{2}$ or $\Sigma_{2}\setminus \Sigma_{1}$ is a compact manifold with $U=U_{1}$ on its boundary contradicting the maximum principle (here, following the notation above, $\Sigma_{1}$ is connected component of $\Sigma\setminus S_{1}$ containing $\overline{\Omega}$ and similarly for $\Sigma_{2}$). So $U_{1}^{-1}$ is connected for every regular value $U_{1}>U_{0}$. 

Now, by contradiction suppose that there is a sequence of regular values $U_{i}>U_{1}$ tending to $U_{\infty}$ such that each $U^{-1}_{i}$ is a sphere. Clearly such sequence of spheres is divergent (i.e. escapes any compact set). Also, by Propositions \ref{TOROROS} and \ref{LALA}, every sphere is embedded inside a solid torus in $\Sigma\# {\rm H}$. Hence, every $U^{-1}_{i}$ bounds a ball. Thus $\Sigma\# {\rm H}$ must be diffeomorphic to $\mathbb{R}^{3}$. Hence, the complement of an open set of $\Sigma$ is diffeomorphic to ${\rm S}^{2}\times [0,\infty)$ and the end must have cubic-volume growth by \cite{MR3302042} which is against the hypothesis.
\end{proof} 


The next Corollary is direct from Propositions \ref{EXISTLIM}, \ref{REGVALUO}.
\begin{Corollary}\label{LERO} Let $(\Sigma; \hg,U)$ be a static black hole data set with sub-cubic volume growth. Then $(\Sigma; \hg,U)$ is a $\hgls$-static end.
\end{Corollary}

We are now ready to prove the Theorem \ref{KAFR}.

\begin{proof}[Proof of Theorem \ref{KAFR}] Suppose that the data is not AK. Let $\{\mathcal{S}_{i}\}$ be a simple cut and let $\gamma$ be a ray. Then, by Proposition \ref{DFNKA}, the curvature decays sub-quadratically along $\gamma\cup(\cup \mathcal{S}_{i})$. By Corollary \ref{LERO} the data is $\hgls$-static and by Proposition \ref{CORONA}  the curvature cannot decay sub-quadratically along $\gamma\cup(\cup \mathcal{S}_{i})$. We obtain a contradiction. Therefore the data is AK.
\end{proof} 

\section{The proof of the classification theorem}\label{TCTH}

\begin{proof}[Proof of the classification theorem \ref{TCTHM}] Let $(\Sigma;\sg,N)$ be a static black hole data set. By Proposition \ref{SOFOR} we know that one of the following holds,
\begin{enumerate}
\item\label{CCC1} $\partial \Sigma = H$, where $H$ is a two-torus, or,
\item\label{CCC2} $\partial \Sigma =H_{1}\cup\ldots\cup H_{h}$, $h\geq 1$, where each $H_{j}$ is a two-sphere, and $(\Sigma;\hg)$ has cubic volume growth, or,
\item\label{CCC3} $\partial \Sigma =H_{1}\cup\ldots\cup H_{h}$, $h\geq 1$, where each $H_{j}$ is a two-sphere, and $(\Sigma;\hg)$ has sub-cubic volume growth.
\end{enumerate}
Then depending on whether $1, 2$ or $3$ holds, we can conclude the following, 
\begin{enumerate}
\item If $\partial \Sigma = H$,  then the data is a Boost as explained in Proposition \ref{SOFOR}.
\item In this case the data is asymptotically flat (with Schwarzschildian fall off), as discussed in Section \ref{ENDSAF}. By Galloway's \cite{MR1201655}, $\Sigma$ is diffeomorphic to $\mathbb{R}^{3}$ minus $h$-balls and the uniqueness theorem of Israel-Robinson-Bunting-Masood-um-Alam, shows that the solution is Schwarzschild. 
\item By Theorem \ref{KAFR} the data is asymptotically Kasner different from a Kasner $A$ or $C$. If the asymptotic is a Boost, that is $B$, then $\Sigma$ is diffeomorphic to a solid three-torus minus a finite number of open three-balls, Proposition \ref{LALA}. If the asymptotic is different from $B$ (and also from $A$ and $C$) then one can clearly find an embedded torus $T$ sufficiently far away that its outwards mean curvature is positive and that separate $\Sigma$ into (i) a compact manifold $\Sigma_{\partial}$ containing the horizons (i.e. $\partial \Sigma$) and another manifold (the `end') diffeomorphic to $[0,\infty)\times {\rm T}^{2}$. It follows again by Galloway's \cite{MR1201655}, that $\Sigma_{\partial}$ is diffeomorphic to a solid torus minus a finite number of open balls. Thus, $\Sigma$ is diffeomorphic to a solid three-torus minus a finite number of open balls. Hence, according to Definition \ref{KNTDEF}, $(\Sigma; g,N)$ is of Myers/Korotkin-Nicolai type.
\end{enumerate}
\end{proof}

\bibliographystyle{plain}
\bibliography{Master}

\end{document}